\documentclass[a4paper,11pt]{amsart}

\pdfoutput=1
\usepackage[top=3cm, bottom=3cm, left=2.5cm, right=2.5cm]{geometry}

 \usepackage{mathtools}
\usepackage{amsaddr}
\usepackage{amsfonts}
\usepackage{amsmath,amssymb,units}
\usepackage{enumitem}
\usepackage[all]{xy}
\usepackage{url,tikz}
\usepackage{blkarray}
\usetikzlibrary{trees}
\usetikzlibrary{arrows,shapes,automata,backgrounds,petri}
\usepackage{bm}
\usepackage{array} 

\usepackage[unicode,psdextra,pdfusetitle,
bookmarks, bookmarksdepth=2, colorlinks=true, linkcolor=blue!50!black, citecolor=black!70, urlcolor=blue!50!black]{hyperref}

\usepackage[sort,compress,numbers]{natbib}

\setlist{nosep}
 
\usepackage[version=3]{mhchem}

\usepackage{amsthm}
\usepackage[nameinlink]{cleveref}

\newtheorem{thm}{Theorem}[section]
\newtheorem{lem}[thm]{Lemma}
\newtheorem{prop}[thm]{Proposition}

\newtheorem{cor}[thm]{Corollary}
\theoremstyle{definition}
\newtheorem{rem}[thm]{Remark}
\newtheorem{example}[thm]{Example}
\newenvironment{ex}
  {\pushQED{\qed}\example}
  {\popQED\endexample}

\numberwithin{equation}{section}

\crefname{thm}{Theorem}{Theorems} 
\crefname{lem}{Lemma}{Lemmas}
\crefname{prop}{Proposition}{Propositions}
\crefname{cor}{Corollary}{Corollaries} 
\crefname{conj}{Conjecture}{Conjectures} 
\crefname{rem}{Remark}{Remarks}
\crefname{ex}{Example}{Examples}

\renewcommand{\k}{\kappa}
\renewcommand{\t}{\theta}
\newcommand{\R}{\mathbb{R}}
\newcommand{\Z}{\mathbb{Z}}
\newcommand{\C}{\mathbb{C}}

\DeclareMathOperator{\rank}{rank}
\DeclareMathOperator{\im}{im}
\newcommand{\wQ}{\widehat{Q}}

\newcommand{\wG}{\widehat{G}}

\newcommand{\fX}{\mathsf{X}}
\newcommand{\fY}{\mathsf{Y}}
\newcommand{\fZ}{\mathsf{Z}}
\newcommand{\Ak}{\mathcal{A}_\kappa}
\newcommand{\Ek}{\mathcal{E}_\kappa}
\newcommand{\mP}{\mathcal{P}}
\newcommand{\mG}{\mathcal{G}}

\newcommand{\mN}{\mathcal{N}}
\newcommand{\mH}{\mathcal{H}}
\newcommand{\mS}{\mathcal{S}}
\newcommand{\mT}{\mathcal{T}}
\newcommand{\mC}{\mathcal{C}}
\newcommand{\mZ}{\mathcal{Z}}
\newcommand{\mB}{\mathcal{B}}

\usepackage[textsize=footnotesize, 
  linecolor=gray, bordercolor=gray, backgroundcolor=gray!25, 
  textcolor=black]{todonotes}

\begin{document}

\title[Positive equilibria in mass action networks: geometry and bounds]{Positive equilibria in mass action networks: geometry and bounds}

\author{Murad Banaji$^1$ \and Elisenda Feliu$^2$}

\date{\today}

\footnotetext[1]{School of Mathematical Sciences, Lancaster University, United Kingdom. m.banaji@lancaster.ac.uk}
\footnotetext[2]{Department of Mathematical Sciences, University of Copenhagen, Denmark. efeliu@math.ku.dk}

\begin{abstract}
Any mass action network gives rise to a parameterised family of polynomial equations whose positive solutions are the positive equilibria of the network. Here, we consider alternative systems of equations, whose solutions are in smooth, one-to-one correspondence with positive equilibria of the network, and capture degeneracy or nondegeneracy of the corresponding equilibria. The construction leads us to consider partitions of networks in a natural sense, and we explore the implications of choosing different partitions. The alternative systems are in some situations simpler than the original mass action equations, which allows us to rapidly identify various algebraic and geometric properties of the positive equilibrium set. This includes the characterisation of toricity and local toricity, bounds on the number of positive nondegenerate equilibria on stoichiometric classes, semialgebraic descriptions of the parameter regions for multistationarity, and the study of bifurcations. After discussing the construction of the alternative systems, various consequences for particular classes of networks and numerous examples are presented. We also develop additional techniques specifically for quadratic networks, the most common class of networks in applications, and use these techniques to derive strengthened results for quadratic networks.

\vspace{0.5cm}

\end{abstract}


\maketitle

\section{Introduction}
Often motivated by problems in systems biology,  the question of when chemical reaction networks (CRNs) admit multiple equilibria, and if so how many, has been studied extensively. For example: there are many results giving necessary conditions for multistationarity; some results with sufficient conditions for multistationarity; results providing explicit bounds on the number of equilibria; results on bifurcations leading to multistationarity; and results on properties of the parameter set for multistationarity. Much of the work focuses on CRNs with mass action kinetics ({\em mass action networks} for short), and this is the setting of this work. A small selection of references on counting equilibria, both in mass action networks and networks with other kinetics, includes: classical early work \cite{hornjackson,Feinberg1995class}; a survey on conditions for multistationarity \cite{joshi-shiu-III}; papers on injectivity-based approaches to multistationarity \cite{craciun,banajicraciun2,feliuwiufSIADS2013,MR12,Muller2015injectivity,shinarfeinbergconcord1}; papers on parameter regions for multistationarity \cite{conradi2019total,Conradi.2016aa,Telek2023}; and papers on systems with special properties \cite{Wang.2008aa,Millan2012toricsteadystates,Tang2020MultistabilityOS,hernandez2023independent,Craciun2009toric,CraciunHeltonWilliams}.

Underlying questions about multistationarity are a deeper set of questions about the geometry of the positive equilibrium set of a mass action network. In particular, what can we infer about this geometry based on the number of species, the number of reactions, the rank of the network, the configuration of its sources, the way that it ``factors'' over various partitions, and other purely combinatorial properties of the network? In this paper, we derive bounds on the number of positive nondegenerate equilibria and show that easily computed combinatorial properties of a network can be used to infer and characterise toricity or local toricity of the equilibrium set. On the way to these results we explore partitions of a network, defined in a natural way, and how properties of the equilibrium set of a network relate to the way it factors over various partitions. We remark that the main construction in \cite{regensburger:gale}, and several consequences presented by the authors, are closely related to some of the results we present here. 

The broad starting point we take follows the principles of Gale duality \cite{bihan-sottile-complete-gale,sottile:book}, adapted to the specific context of mass action networks. While positive equilibria of a mass action network are, naturally, the positive solutions of a system of polynomial equations, we use these principles to write down {\em alternative} systems of equations whose solutions on some polyhedron are in smooth, one-to-one correspondence with positive equilibria of the network, perhaps on some stoichiometric class, see in particular \Cref{thminj,thm:degensolvability}. The basic idea of the construction is as follows (details are given in \Cref{secmain}): positive equilibria satisfy a polynomial equation of the form $\Gamma (\k \circ x^A)=0$, where $\Gamma$ is a matrix, $x^A$ is a vector of monomials and $\k$ a vector of positive constants multiplying the monomial vector componentwise. For any fixed $\k$, a positive vector $x$ satisfies this equation if and only if $\k \circ x^A=z$ for some $z$ in the positive part of $\ker \Gamma$. 
Given any positive $z$, one can easily recover all $x$ satisfying $\k \circ x^A=z$ as, after applying logarithms, this system becomes a linear system. Thus, imposing that $z$ lies in the positive part of $\ker \Gamma$, and that $\ln (z/\k)$ lies in $\mathrm{im}\,A$, one obtains a description of all solutions to the original system using linear algebra alone. 

This approach to solving polynomial systems, which extends to zeros over the nonzero complex numbers, appears in many studies (see e.g. the review in \cite{dickenstein-binomials}), and has proven useful for the study of fewnomial systems \cite{bihan:descartes:2,bihan:gale}. In the context of reaction networks, this idea appears in disguise for example in the deficiency one theorem of Feinberg \cite{Feinberg1995class}, and explicitly in the language of algebraic geometry in several works of Gatermann and coauthors e.g. \cite[Section 7]{gatermann-sparse} and \cite[Section 3.2]{gatermann-hopf}. The construction provided the basis to study parameter regions for multistationarity in \cite{conradi-flockerzi}, and to count positive equilibria in \cite{helmer:gale}. It was used in \cite{banajiborosnonlinearity,BBH2024smallbif} to write down parameterisations of equilibria with a view to studying bifurcations, and it plays a central part in various algorithms used in the study of mass action networks gathered in the library \cite{muradgithub}. In recent work \cite{regensburger:gale}, using these ideas, the authors give an explicit description of the solution set of a system of generalised polynomial inequalities of the form $\k \circ x^A \in C$, where $\k$ is a positive vector, $C$ is a convex cone in the positive orthant, and the exponents in $A$ are allowed to be real numbers. Additionally, the authors exploit possible direct product decompositions of the cone $C$. In this work, we restrict our attention to polynomial systems arising from reaction networks, but also consider decompositions of the positive part of $\ker \Gamma$. We arrive at explicit descriptions of the positive equilibrium set of an arbitrary reaction network, and hence at the alternative systems of equations.

Our main concern is to use properties of the alternative systems to make claims about various classes of mass action networks. The alternative systems may include fewer variables and equations and/or have lower degree, than the original systems while, crucially, capturing the cardinality of the equilibrium set and degeneracy or nondegeneracy of equilibria. In some important cases, even when the network has many species and reactions, the alternative system becomes trivial, or reduces to a single polynomial equation in one real variable, allowing one to use techniques such as Descartes' rule of signs and its generalisations, or Sturm's theorem (see Chapter~1 of \cite{Sturmfels02solvingsystems}, for example) to make claims about the positive equilibrium set. Even when not so simple, the alternative system may be easier to analyse than the original mass action equations using tools of computational algebra such as B\'ezout's theorem, the Bernstein–Khovanskii–Kushnirenko (BKK) theorem \cite{bernstein}, or cylindrical algebraic decomposition \cite[Chapter 11]{BasuAlgorithms}. 

Particular classes of networks for which we derive bounds on positive nondegenerate equilibria include networks with relatively few reactions or sources compared to the number of species, networks with few or no conservation laws, and networks with few affine dependencies amongst their sources (\Cref{thmnnn1,,thmfewsources,,thm:n_plus_2} for example). Specialising to quadratic networks, we get sharper or deeper results (\Cref{thmquadnn2n,,thmquadnn2nA} for example). Moreover, the theory sometimes allows us to find examples confirming that these bounds are achieved by examining only a relatively small number of candidate networks. 

As discussed above, and also highlighted in \cite{regensburger:gale}, the construction often leads to explicit parameterisations of the equilibrium set of a mass action network. The theory we develop here is useful in the analysis of bifurcations in mass action networks both because explicit parameterisations of equilibria simplify many computations, and because, as we will show, the alternative equations preserve degeneracy/nondegeneracy of solutions of the original mass action equations. In fact, the analyses of bifurcations in mass action networks developed in \cite{banajiborosnonlinearity,BBH2024smallbif} relied on special cases of the construction presented here. Moreover, as we show in examples, the approach here can be directly used to identify potential bifurcation sets for bifurcations such as fold and cusp bifurcations \cite{kuznetsov:2023} near which the cardinality of the equilibrium set changes. The role of the alternative equations in the study of bifurcations will be explored further in forthcoming work.

Also touched on in examples, but not explored in depth, is the potential of these techniques for characterising the parameter set for multistationarity. In some cases, the results allow an explicit semialgebraic characterisation of the set of rate constants at which a given mass action network admits a given number of equilibria on some stoichiometric class. Deriving such a characterisation by direct examination of the mass action equations can be quite nontrivial even for relatively simple networks.

We present numerous examples to illustrate the power of the techniques we develop; some of these are drawn from the biological literature (see \Cref{ex:2site,,exBIOMD,,ex664}, and \Cref{remodebase}). While our examples involve small to medium-sized networks, we expect that {\em inheritance} results, which tell us how networks inherit dynamical behaviours and bifurcations from their subnetworks, can be used alongside the results here to make claims about larger networks. In particular, nondegenerate multistationarity, and transversally unfolded bifurcations, identified using the techniques described here, can be ``lifted'' to larger, more complex, networks, see, e.g., \cite{banaji:boros:hofbauer:2024b,craciun-entrapped,banajipanteaMPNE,feliu:intermediates,joshishiu,cappelletti:flow}. Similarly, several of our results give easy criteria that guarantee the absence of multistationarity. These results can be used to identify minimal mechanisms underlying multistationarity in larger networks.

\section{Preliminaries on chemical reaction networks}\label{sec:background}

We write $\mathbb{R}^n_{\geq 0}:=\{x \in \mathbb{R}^n\,:\,x_i\geq 0,\,\,i=1,\ldots,n\}$ and $\mathbb{R}^n_+:=\{x \in \mathbb{R}^n\,:\,x_i> 0,\,\,i=1,\ldots,n\}$ for the nonnegative and positive orthants in $\mathbb{R}^n$ respectively; similar notation is used for these orthants in $\mathbb{Z}^n$. Vectors in $\mathbb{R}^n_{\geq 0}$ are termed {\bf nonnegative} and vectors in $\mathbb{R}^n_+$ are termed {\bf positive}. For $v\in \R^n$, $D_v$ is the $n \times n$ diagonal matrix with diagonal given by $v$. $\mathbb{C}^*$ denotes the complex torus $\mathbb{C}\setminus \{0\}$.

For a matrix $A=(a_{ij}) \in \mathbb{Z}_{\geq 0}^{m \times n}$ and $x\in \R^n_+$, we let $x^A\in \R^{m}_+$ denote the vector whose $i$th entry is $\prod_{j=1}^n x_j^{a_{ij}}$. Observe that $A$ encodes the exponents in \emph{rows}. Throughout the text, ``$\bm{1}$'' will denote a column vector of ones whose length is inferred from the context.   The symbol ``$\circ$'' will denote the entrywise product, and we will freely apply various functions and operations entrywise; in particular, $\ln v=(\ln v_1,\dots,\ln v_n)$ for $v\in \R^n_+$ and $u/v=(u_1/v_1,\dots,u_n/v_n)$ for $u\in \mathbb{C}^n$, $v\in (\mathbb{C}^*)^n$. Given a differentiable map $F$ between subsets of Euclidean space, $J_F$ refers to the Jacobian matrix of $F$.

Throughout, we say a property holds \emph{generically} in $X \subseteq \mathbb{C}^n$, if it holds in a nonempty Zariski open subset of $X$, that is, in a nonempty set of the form $ X\setminus V$ with $V$ being the zero locus of a set of polynomials.

{\bf Projections and partitions.} Given $X \subseteq \mathbb{R}^m$ and $B \subseteq \{1, \ldots, m\}$, write $X_B$ for the orthogonal projection of $X$ onto the coordinate hyperplane $\{x \in \mathbb{R}^m\,:\, x_i=0 \mbox{ if } i \not \in B\}$. When $X$ is a singleton, say $X=\{v\}$, we write $v_B$ rather than $\{v\}_B$, and note that $v_B = \bm{1}_B \circ v \in \R^m$. Given a partition $\mP = \sqcup_{i=1}^p P_i$ of $\{1, \ldots, m\}$, we refer to the subsets $P_i$ as {\bf blocks} of $\mP$. We refer to a partition with $p$ blocks as a ``$p$-partition''. Given two  partitions $\mP, \widetilde{\mP}$ of $\{1, \ldots, m\}$, we write $\mP \leq \widetilde{\mP}$ to mean that $\tilde{\mP}$ is a refinement of $\mP$. 

Given $X, Y_1, \ldots, Y_p \subseteq \mathbb{R}^m$, we abuse notation slightly and write $X = Y_{1} \oplus \cdots \oplus Y_{p}$ to mean that each element of $X$ can be written uniquely as a sum $y_1 + \cdots + y_p$ where $y_i \in Y_i$.  We say that $X \subseteq \mathbb{R}^m$ {\bf factors over $\mP$} if $X = X_{P_1} \oplus \cdots \oplus X_{P_p}$.
It is easily seen that if $X$ factors over two partitions, say $\mP$ and $\widetilde{\mP}$, then it factors over their common refinement.  
Assuming that $X \subseteq \R^m$ is nonempty and nonzero, we refer to the (unique) finest partition over which $X$ factors as the \textbf{finest partition} associated with $X$. We may also observe that if $X$ is a relatively open subset of a linear subspace $\mS \subseteq \R^m$, then if $X$ factors over a partition $\mP$, $\mS$ must also factor over $\mP$.   
Writing $\bm{1}_\mP$ for the $m \times p$ matrix whose $i$th column is $\bm{1}_{P_i}$, we will often use the easy fact that a linear subspace $\mS \subseteq \R^m$ factors over $\mP$ if and only if $(\bm{1}_\mP \lambda)\circ x \in \mS$ for all $x \in \mS, \lambda \in \R^p$. Consequently, if $\mS$ includes a vector none of whose components are zero, it must have dimension at least $p$.  Another immediate consequence is that if $\mS$ factors over $\mP$, then so does its orthogonal complement $\mS^\perp$.

 \smallskip
{\bf Basics on reaction networks. } To avoid extensive repetition, basic notation and terminology on reaction networks are presented only briefly, and the reader is referred to \cite{banajiborosnonlinearity,BBH2024smallbif} for more detail. These concepts will be illustrated in examples throughout the forthcoming sections.

A (chemical) {\bf complex} is a formal linear combination of (chemical) species with nonnegative integer coefficients termed {\bf stoichiometries}. The {\bf molecularity} of a complex is the sum of its coefficients. A (chemical) {\bf reaction} is the conversion of one complex, here termed the {\bf source}, into another distinct complex, here termed the {\bf product}. Mathematically, a reaction is simply an ordered pair of distinct complexes. A {\bf (chemical) reaction network} ({\bf CRN} for short) is a collection of reactions, and when the reactions are assumed to have mass action kinetics we refer to the CRN as a {\bf mass action network} (also sometimes referred to as a \emph{mass action system} in the literature).  Henceforth, whenever we refer to a ``reaction network'', or ``network'', without further qualification, this will mean a mass action network.

Having made some choices of orderings on species and reactions, a mass action network with $n$ species and $m$ reactions can be characterised by three objects:
\begin{itemize}
\item The {\bf stoichiometric matrix} $\Gamma \in \mathbb{Z}^{n \times m}$ whose $(i,j)$th entry is the net production of species $i$ in reaction $j$. 
\item The {\bf exponent matrix} $A \in \mathbb{Z}_{\geq 0}^{m \times n}$ whose $(i,j)$th entry is the stoichiometry of species $j$ in the source of the $i$th reaction. 
\item A vector of {\bf rate constants} $\k \in \mathbb{R}^m_+$.
\end{itemize}

The associated mass action ODE system can be written briefly as
\begin{equation}
\label{eqODE}
\dot x = g_\k(x):=\Gamma(\k \circ x^A)
\end{equation}
where $x \in \mathbb{R}^n_{\geq 0}$ is the vector of species concentrations. A standard computation shows that the Jacobian matrix of the right hand side of \eqref{eqODE} evaluated at $x\in \mathbb{R}^n_+$ is
\begin{equation}\label{eq:jac}
    J_{g_\k}(x) =  \Gamma\,  D_{\k\circ x^A}\,  A \,  D_{1/x}\, . 
\end{equation}

 A CRN is {\bf quadratic} if all its sources have molecularity two or less; equivalently, each polynomial on the right hand side of \eqref{eqODE} has degree two or less. It is {\bf bimolecular} if all sources {\em and products} have molecularity two or less.

The {\bf rank} $r$ of a CRN is defined as 
\[ r := \rank \Gamma \quad \leq \min\{n,m\}\, . \]  
If $r=n$, we say the network has {\bf full rank}. We refer to a CRN with $n$ species, $m$ reactions and rank $r$ as an {\bf $\boldsymbol{(n,m,r)}$ network}. 

The {\bf source rank} $s$ of a CRN   is the dimension of the affine hull of its source complexes, considered as points in $\mathbb{R}^n$. Equivalently, 
\[ s := \rank\, [A\,|\,\bm{1}]-1  \quad \leq \min\{n,m-1\}\,.\]

The intersections of cosets of $\mathrm{im}\,\Gamma$ with $\mathbb{R}^n_{\geq 0}$ (resp., $\mathbb{R}^n_+$) are termed {\bf stoichiometric classes} (resp., {\bf positive stoichiometric classes}) of the CRN. Positive stoichiometric classes are easily seen to be locally invariant under \eqref{eqODE}. A stoichiometric class is defined by  equations  
\begin{equation}\label{eq:Z1}
Z x -K=0 \, , \qquad  x\in \R^n_{\geq 0}\, , 
\end{equation}
where $K\in \R^{n-r}$ is a vector of real constants, called  {\bf total amounts}, and $Z\in \R^{(n-r) \times n}$ is a matrix whose rows form a basis of $\ker\Gamma^\top$ called a \textbf{matrix of conservation laws}. When referring to total amounts, we implicitly assume a choice of $Z$ has been made.

We refer to $\mC := \ker\Gamma \cap \mathbb{R}^m_+$ as the {\bf positive flux cone} of the CRN \cite{Clarke:1980tz}. A CRN is termed {\bf dynamically trivial} if $\mC=\emptyset$ and {\bf dynamically nontrivial} otherwise. Dynamically trivial CRNs admit no positive equilibria (and, in fact, no omega-limit sets intersecting the positive orthant). On the other hand, dynamically nontrivial mass action networks admit positive equilibria for at least some choices of rate constants \cite{banajiCRNcount}. Dynamically nontrivial CRNs are also sometimes called ``consistent''   (e.g., \cite{JoshiShiu2017}).

If $\mC\neq \emptyset$, then the closure $\overline{\mC}$ of $\mC$, namely $\ker\Gamma \cap \mathbb{R}^m_{\geq 0}$, is a closed, pointed, convex, polyhedral cone, satisfying $\mathrm{dim}\,\mC = \mathrm{dim}\,\overline{\mC} = m-r$. If $m-r \leq 2$, then $\overline{\mC}$ is simplicial, i.e., it has exactly $m-r$ extremal vectors (unique up to scaling); however, if $m-r \geq 3$, then $\overline{\mC}$ need not be simplicial. 

\smallskip
{\bf Positive, nondegenerate equilibria.} For fixed $\k \in \mathbb{R}^m_+$, we consider the set of positive equilibria of \eqref{eqODE}, namely,
\[\Ek := \{x \in \mathbb{R}^n_+ : \Gamma(\k \circ x^A) = 0 \}. \]
 It is clear that, for any fixed $\k \in \mathbb{R}^m_+$, $x \in \Ek$ if and only if
\begin{equation}\label{eqbasic}
\k \circ x^A \in \mC\,.
\end{equation}
We say that $x^*\in \mathcal{E}_\k$ is {\bf nondegenerate} if $J_{g_\k}(x^*)$ acts as a nonsingular transformation on $\mathrm{im}\,\Gamma$,  i.e., $\ker J_{g_\k}(x^*) \cap \mathrm{im}\,\Gamma = \{0\}$; and it is {\bf degenerate} otherwise \cite{banajipantea}. Equivalently, $x^*$ is nondegenerate if it is a nondegenerate equilibrium, in the usual sense, of the {\em restriction} of \eqref{eqODE} to the   stoichiometric class on which $x^*$ lies. Note that, by the implicit function theorem, a positive nondegenerate equilibrium is necessarily isolated on its stoichiometric class.

It is convenient also to define $x^*\in \mathcal{E}_\k$ to be {\bf strongly degenerate} if $\mathrm{rank}\,J_{g_\k}(x^*)<r$. Note that a strongly degenerate equilibrium is degenerate, but the converse is not necessarily true; however, for networks of full rank, ``strongly degenerate'' and ``degenerate'' are clearly equivalent.

A reaction network is termed {\bf nondegenerate} if it admits a positive nondegenerate equilibrium for some rate constants, and {\bf degenerate} otherwise. A network is {\bf strongly degenerate} if every positive equilibrium of the network is strongly degenerate. Dynamically trivial networks are included amongst the strongly degenerate (hence also the degenerate) networks.

From the general form \eqref{eq:jac} of the Jacobian matrix of a mass action network evaluated at a positive equilibrium, and using \eqref{eqbasic}, 
a mass action network is strongly degenerate if 
\begin{equation}\label{eq:strongdeg}
\mathrm{rank}(\Gamma\,  D_v\,  A) < r
\end{equation} 
for all $v\in \mC$; and is nondegenerate if and only if $\Gamma D_v A$ has an $r \times r$ principal minor which is not identically zero. The latter is easily inferred using the Cauchy-Binet formula (see e.g. \cite[Appendix A]{banajipantea}, or \cite[Prop. 5.3]{feliuwiufSIADS2013}). Equivalently, the network is nondegenerate if and only if the matrix 
\begin{equation}\label{eq:nondeg}
 \begin{pmatrix} \Gamma\, D_v\, A \, D_x \\ Z \end{pmatrix} \, 
\end{equation} 
has rank $n$ for some $x\in \R^n_+$ and $v\in \mC$,  with $Z$ being any matrix of conservation laws,
see e.g. \cite[Prop. 5.2]{feliuwiufSIADS2013}. As this condition holds for almost all values of $x,v$ if it holds for one, the condition can be easily verified by considering a random choice.
 
We recall the following results on dimension, nondegeneracy and emptiness of equilibrium sets from \cite[Thm 3.1, 3.4]{feliu:dimension}; see also \cite[Thm 3.7]{feliu:vertical}.

\begin{prop}
\label[proposition]{prop:nondeg}
Given a reaction network, let $\mathcal{Z} \subseteq \R^m_+$ be the set of rate constants $\k$ for which $\Ek\neq \emptyset$, and let $\mathcal{Z}_{sc} \subseteq \R^m_+ \times \R^{n-r}$ be the set of rate constants and total amounts $(\k,K)$ for which the corresponding stoichiometric class intersects $\Ek$. 
Then:
\begin{itemize}
 \item  The network is  strongly degenerate if and only if $\mathcal{Z}$ has empty Euclidean interior, equivalently $\Ek$ is generically empty. 
\item The network is degenerate if and only if $\mathcal{Z}_{sc}$ has empty Euclidean interior, equivalently for generic choices of $(\k,K)$, $\Ek$ does not intersect the stoichiometric class defined by $K$. 
 \end{itemize}
 
 Additionally:
 \begin{itemize}
  \item[(i)]  If the network is not strongly degenerate, then, for generic $\k$ in $\mathcal{Z}$, the dimension of $\Ek$ is $n-r$.

  \item[(ii)] If the network is nondegenerate, then (i) applies and also, for generic $(\k,K)$ in $\mathcal{Z}_{sc}$, the stoichiometric class defined by $K$ contains a finite number of positive equilibria in $\Ek$, and these are all nondegenerate. 
 \end{itemize}

\end{prop}

We emphasise that \Cref{prop:nondeg} tells us that the existence of a positive nondegenerate  equilibrium for some choice of rate constants is the only requirement for the conclusions of (ii) to hold. A degenerate network that is not strongly degenerate will have $\mZ_{sc}$ with empty Euclidean interior while the interior of $\mZ$ will be nonempty. In particular (i) will hold but (ii) will not.

In \Cref{secdegen} we consider a class of (strongly) degenerate networks, which we will refer to as ``overdetermined'', for which the dimension of $\Ek$ is always larger than $n-r$ when nonempty. The latter might not hold for an arbitrary strongly degenerate network, see  \cite[Example 3.10]{feliu:dimension}.

\smallskip
{\bf Naive B\'ezout and BKK bounds.} For any fixed $\k$, the right hand side of \eqref{eqODE} is a vector of $n$ polynomials $g_{\k,1}(x),\dots,g_{\k,n}(x)$ in the variables $x=(x_1, \ldots, x_n)$.  Given any $r \times m$ matrix $\widehat{\Gamma}$ whose rows form a basis of the row-span of $\Gamma$, let $\hat{g}_{\k}(x) := \widehat{\Gamma}(\k \circ x^A)$. Then positive equilibria on the stoichiometric class with total amount $K$ are precisely the positive solutions to the square system
\begin{equation} \label{eqbasic0}
 \begin{aligned}
\hat{g}_{\k,j}(x)&= 0 \qquad && \quad j = 1, \ldots, r,\\ 
Z\, x - K  &= 0 \, , 
 \end{aligned}
\end{equation}
where $Z\in \R^{(n-r)\times n}$ is a matrix of conservation laws.
By B\'ezout's theorem, the product of degrees of  $\hat{g}_{\k,j}$ ($j = 1, \ldots, r$) provides an upper bound on the number of nondegenerate solutions to \eqref{eqbasic0} in $\C^n$, hence on the number of positive nondegenerate  equilibria of \eqref{eqODE} on any stoichiometric class. We refer to any such bound as a {\bf naive B\'ezout bound}. In general, the choice of $\widehat{\Gamma}$ is not unique, and so naive B\'ezout bounds are not in general unique. 

We may also consider the BKK theorem, which establishes that the mixed volume of the polynomials on the left hand side of \eqref{eqbasic0} provides an upper bound on the number of nondegenerate solutions in $(\mathbb{C}^*)^n$, hence of positive  nondegenerate equilibria on any stoichiometric class. This bound, here called a {\bf naive BKK bound}, is no larger than the naive Bézout bound for the same choice of rows of $\Gamma$.

If we know the sources and the rank of a network, but not its products, then we know all the monomials that might appear in \eqref{eqbasic0} but not their coefficients. In particular, if $d$ is the maximum degree of the monomials in $x^A$, then all naive B\'ezout bounds are clearly bounded above by $d^r$; we refer to $d^r$ as the {\bf naive B\'ezout source bound} of the network.
Similarly, by letting $P$ be the convex hull of all sources of the network, and $\Delta$ the standard $n$-dimensional simplex with vertices the origin and the standard vectors of $\R^n$, the BKK theorem   tells us that the number of positive nondegenerate  equilibria is at most the mixed volume of $r$ copies of $P$ and $n-r$ copies of $\Delta$. We refer to this bound as the {\bf naive BKK source bound} which, by the monotonic properties of mixed volumes, is not lower than any naive BKK bound. 

\smallskip
{\bf Fewnomial bounds on positive equilibria.} 
The previous are bounds on the number of nondegenerate solutions to \eqref{eqbasic0} in $(\mathbb{C}^*)^n$  or $\mathbb{C}^n$. We give now, for later use, some bounds on the number of nondegenerate {\em positive} equilibria based on the number of monomials of \eqref{eqbasic0}. The main technique behind these bounds is the rewriting of the system of equations using \emph{Gale duality}. We will employ these same principles throughout the remainder of this work. 

\begin{prop}\label[proposition]{prop:fewnomials}
For a system of $n\geq 2$ real polynomial equations in $n$ variables and with $n+k+1$ monomials, the following holds: 
\begin{itemize}
\item[(i)] If $k<0$, then the system has no positive solution for generic choices of the coefficients. 
 \item[(ii)] If $k=0$, then the system has either no positive solution, one positive  nondegenerate solution, or infinitely many positive solutions.
 \item[(iii)] If $k=1$, then the system has at most $n+1$ positive nondegenerate solutions.
 \end{itemize}
\end{prop}

The first statement is a consequence of the BKK bound, while a proof of the second statement can be found e.g. in \cite[Prop. 3.3]{Bihan:positive_decorated}. Bihan gave the case $k=1$ in \cite[Prop. 2.1]{Bihan:circuit}, and sharper bounds  in the spirit of the univariate Descartes' rule of signs and which take into account the exponents and the coefficient matrix, are given in \cite{bihan:descartes:1,bihan:descartes:2}. Bounds for $k\geq 2$ are given in \cite{bihan:gale}, but these are too large to be of interest for the purposes of this work.

\section{Existence and number of positive equilibria: general theory}
\label{secmain}

In this section, we follow the principles of Gale duality (see \cite{bihan-sottile-complete-gale,sottile:book}) to provide a parameterisation of the set of positive equilibria of a mass action network, namely, Equation \eqref{eq:Ek_expression} in \Cref{thminj}(i). We then explore properties of this parameterisation, and use it to derive the alternative systems of equations in Equation~\eqref{gensol} in \Cref{thm:degensolvability}, which are used heavily in subsequent results on various special classes of networks. We use the language and exposition from \cite[Section 2]{BBH2024smallbif}, but note that closely related ideas have been exploited in several previous works about reaction networks \cite{helmer:gale,kaihnsa:connectivity}  and in the recent framework for systems of polynomial inequalities \cite{regensburger:gale}. We consider decompositions of the positive flux cone as in \cite{bryan:partition2,bryan:partition,feliu:toric}. In \cite{regensburger:gale}, decompositions of cones are also considered. 
 
While the construction we present below closely follows the template developed in \cite{banajiborosnonlinearity,BBH2024smallbif}, where special cases of this construction were considered, it can also be derived from \cite{bihan-sottile-complete-gale} when considering the trivial partition. It is also closely related to the construction in \cite[Theorem 5]{regensburger:gale}, and in \Cref{rem:gale}, we elaborate further on this relationship. Despite the existence of related results, we present the construction in full here, keeping the paper self-contained and allowing us to develop machinery and notation which are subsequently reused in various other claims, including about nondegeneracy of equilibria.

Throughout the section, we consider an arbitrary dynamically nontrivial $(n,m,r)$ mass action network $\mN$ with stoichiometric matrix $\Gamma$, exponent matrix $A$, and positive flux cone $\mC:=\ker \Gamma\cap \R^m_+$.

 \subsection{Partition subnetworks}
Let us assume that $\mC$, hence, by our earlier remarks on partitions, $\ker\Gamma$, factor over some $p$-partition $\mP = \{P_1, \ldots, P_p\}$ of $\{1, \ldots, m\}$. Note that this assumption is without loss of generality as we may always choose $\mP$ to be the trivial partition, in which case $p=1$. In this case, we say that $\mN$ {\bf factors over $\bm{\mP}$}. (Whenever this terminology is used, it will always be implicitly assumed that the network in question is dynamically nontrivial.) As $\dim \ker \Gamma = m-r$ and $\mC\neq \emptyset$, we have $m-r\geq p$. Each subnetwork $\mN_i$ having the same species as $\mN$ but only the reactions indexed by $P_i$ is called a {\bf partition subnetwork} of $\mN$. 

The partition $\mP$ induces a partition on the set of reactions, hence on the sets of rate constants, of columns of $\Gamma$, and of rows of $A$. However, the partition need not induce a partition on the set of sources, as the same source may occur in reactions in distinct blocks of $\mP$. 

Set $m_i := |P_i|$, and write $\Gamma^{(i)}\in \Z^{n\times m_i}$ for the submatrix of $\Gamma$ with columns indexed by $P_i$; $A^{(i)}\in \Z_{\geq 0}^{m_i\times n}$ for the submatrix of $A$ consisting of rows indexed by $P_i$; and $\k^{(i)}\in \R^{m_i}_+$ for the subvector of $\k$ with entries indexed by $P_i$.
 An important and easy observation is that $\ker\Gamma$ factors over $\mP$ if and only if
\begin{equation}\label{eq:Gammasum}
\mathrm{im}\,\Gamma = \mathrm{im}\,\Gamma^{(1)} \oplus  \cdots  \oplus \mathrm{im}\,\Gamma^{(p)}\,.
\end{equation}
Consequently the rank of a network is the sum of the ranks of its partition subnetworks. Equilibria of the network are solutions of $0 = \Gamma(\k\circ x^A) = \sum_{i=1}^p\Gamma^{(i)}(\k^{(i)}\circ x^{A^{(i)}})$. From \eqref{eq:Gammasum}, this is equivalent to
\[
\Gamma^{(i)}(\k^{(i)}\circ x^{A^{(i)}}) = 0, \,\qquad i = 1, \ldots, p\,,
\]
namely, equilibria of $\mN$ are  precisely the common equilibria of the partition subnetworks $\mN_i$. If, for example, some $\mN_i$ has no positive equilibria, the same holds for $\mN$. We will later make use of the following result.

\begin{lem}
  \label[lemma]{lem:degensubnet} 
 Assume that a dynamically nontrivial $(n,m,r)$ network $\mN$ factors over a $p$-partition $\mP$ of $\{1, \ldots, m\}$. Fix $\k \in \R^m_+$, and suppose that some positive equilibrium $x$ of $\mN$ is a strongly degenerate equilibrium of the partition subnetwork $\mN_j$ for some $j$. 
Then $x$ is also strongly degenerate as an equilibrium of $\mN$. Consequently, if, for some $\k$ and some $j$, all equilibria of $\mN_j$ are strongly degenerate, then this holds for $\mN$. If, for some $j$, $\mN_j$ is strongly degenerate, then so is $\mN$.
\end{lem}
\vspace{-\baselineskip}
\begin{proof}
  With the notation above we write $\Gamma = [\Gamma^{(1)}\,|\,\cdots\,|\,\Gamma^{(p)}]$. Fix $\k \in \R^m_+$ and a positive equilibrium $x$ of $\mN$. Recall that $x$ is also an equilibrium of $\mN_j$ for each $j$. Let $r_i=\rank \Gamma^{(i)}$ so that, by \eqref{eq:Gammasum}, we have $r=r_1+\dots+r_p$.  
  Using \eqref{eq:jac} we obtain 
\[
J := J_{g_\k}(x) = \Gamma \, D_{\k\circ x^A} \, A\,D_{1/x} = \sum_{i=1}^p J_i\,, \quad \mbox{where} \quad J_i:=\Gamma^{(i)}\,D_{\k^{(i)}\circ x^{A^{(i)}}} A^{(i)} D_{1/x}\,.
\]
But $\rank J_i\leq \rank \Gamma^{(i)} = r_i$ for all $i$, and if $x$ is a strongly degenerate equilibrium of $\mN_j$, then $\mathrm{rank}\, J_j < r_j$. As $\mathrm{im}\,J = \mathrm{im}\,J_1 \oplus \cdots \oplus \mathrm{im}\,J_p$, and hence $\mathrm{rank}\,J = \sum_i\mathrm{rank}\,J_i$, it follows that $\mathrm{rank}\,J < r$, and so $x$ is strongly degenerate as an equilibrium of $\mN$. The remaining claims now follow immediately.
\end{proof}

\subsection{The main equation for positive equilibria}\label{subsec:1}
We present an important equation satisfied by positive equilibria of a mass action network, a special case of which appeared as Equation~2 in \cite{BBH2024smallbif}. We assume that the network factors over a $p$-partition $\mP$. 
 
We begin by parameterising the positive flux cone $\mC$ in a natural way. Write $\mC_i:=\mC_{P_i}$ for tidiness, so that $\mC = \mC_1 \oplus \cdots \oplus \mC_{p}$. For each $i = 1, \ldots, p$, choose a flat cross section of $\mC_i$, say $\widehat{\mC}_i$, defined as the intersection of $\mC_i$ with any hyperplane in $\R^m$ which intersects each ray of $\mC_i$ exactly once. Write down an affine injective map 
\[h^{(i)} \colon Y_i \longrightarrow \widehat{\mC_i}\]
with domain $Y_i \subseteq \mathbb{R}^{\dim \mC_i-1}$ an open polyhedron. We can now write a general point in $\mC$ uniquely as $\sum_{i=1}^p\lambda_i\,h^{(i)}(\alpha^{(i)})$ with $(\lambda_i, \alpha^{(i)}) \in \mathbb{R}_+ \times Y_i$. Writing 
\[ Y := Y_1 \times \cdots \times Y_p \subseteq \mathbb{R}^{m-r-p},\quad  \quad \widehat{\mC}:= \widehat{\mC}_1 \oplus \cdots \oplus \widehat{\mC}_p \subseteq \R^m_+\, , \]
and 
$\alpha := (\alpha^{(1)}, \ldots, \alpha^{(p)}) = (\alpha_1, \ldots, \alpha_{m-r-p})$ for an element of $Y$,
we define the affine injective map
\begin{equation}\label{eq:h}
h:= h^{(1)} \oplus \cdots \oplus h^{(p)} \colon \quad Y \longrightarrow  \  \widehat{\mC}  \, \qquad 
(\alpha^{(1)}, \ldots, \alpha^{(p)}) \mapsto  \sum_{i=1}^p\,h^{(i)}(\alpha^{(i)}) \,. 
\end{equation}
We may naturally extend the domain of $h$ to all of $\R^{m-r-p}$ if needed. In the special case where $\mC_i$ is one dimensional, and so $\widehat{\mC_i}$ consists of a single point, say $v^{(i)}$, we may simply write $h^{(i)}(\alpha^{(i)}) = v^{(i)}$.

Write $\lambda := (\lambda_1, \ldots, \lambda_{p})^\top$. For given $\k \in \mathbb{R}^m_{+}$, \eqref{eqbasic} is satisfied by some $x \in \mathbb{R}^n_+$ if and only if there exist $(\lambda, \alpha) \in \mathbb{R}^p_+ \times Y$ such that
\begin{equation}
\label{eqbasicfull} 
\k \circ x^A = \sum_{i=1}^p\lambda_i\,h^{(i)}(\alpha^{(i)}) = (\bm{1}_{\mP}\lambda)\circ h(\alpha)\,.
\end{equation}
We remark that, by construction, each set of the form $\{(\bm{1}_{\mP}\lambda)\circ h(\alpha)\,:\,\lambda \in \R^p_+\}$ intersects $\widehat{\mC}$ exactly once, namely at $h(\alpha)$, and that these sets partition $\mC$.

\begin{rem}
\label[remark]{rem:Jacla}
Using \eqref{eq:jac} and \eqref{eqbasicfull}, the Jacobian matrix of $g_\k$ at $x \in \Ek$, can be written
\begin{equation}
\label{eq:jac1}
J_{g_\k}(x):= \Gamma \,D_{(\bm{1}_\mP \lambda) \circ h(\alpha)}\,A\,D_{1/x}\,,
\end{equation}
where $\lambda$ and $\alpha$ are (uniquely) determined by \eqref{eqbasicfull}. 
\end{rem}
 We can take logarithms and rearrange \eqref{eqbasicfull} to get
\begin{equation}
\label{eqbasiclog1}
\ln\k + [A\,|\,\bm{1}_\mP]\,\left(\begin{array}{r}\ln x\\-\ln\lambda\end{array}\right) = \ln h(\alpha)\,.
\end{equation}
Equation \eqref{eqbasiclog1} is the crucial equation satisfied by positive equilibria of a mass action network. The special case of Equation \eqref{eqbasiclog1} where $\mP$ is trivial, namely $\bm{1}_\mP = \bm{1}$, appeared in \cite{BBH2024smallbif}. 
A similar equation appeared in \cite[Eq. (12a)-(12l)]{conradi-flockerzi}, where the authors consider a parameterisation of the positive flux cone, instead of restricting to a section.

\begin{rem}[The geometry of Equation~\eqref{eqbasiclog1}]
 As $(\k, x,\lambda)$ varies in $\R^m_+ \times \R^n_+ \times \R^p_+$, the left hand side of \eqref{eqbasiclog1} defines a family of cosets of the hyperplane $\mH_{\mP}:= \mathrm{im}\,[A\,|\,\bm{1}_\mP]$ in $\mathbb{R}^m$. As $\ln\,h(\cdot)$ is a smooth, injective, immersion of $Y$ into $\R^m$, the right hand side of \eqref{eqbasiclog1} defines an $(m-r-p)$-dimensional manifold
\begin{equation}\label{eq:Y}
\mathcal{Y} := \big\{\ln h(\alpha)\,:\,\alpha \in Y\big\}\, \subseteq \mathbb{R}^m\,.
\end{equation}
For any fixed $\k \in \mathbb{R}^m_+$, points of intersection of $\mathcal{Y}$ and $\ln\k + \mH_{\mP}$ correspond to solutions of \eqref{eqbasiclog1} and hence to positive equilibria of \eqref{eqODE}. This correspondence is one-to-one if and only if $\ker\,[A\,|\,\bm{1}_\mP]$ is trivial.
\end{rem}

\begin{rem}[Decompositions of Equation~\eqref{eqbasiclog1}]
\label[remark]{remsubsyst}
Observe that solutions to \eqref{eqbasiclog1} are precisely the common solutions to the corresponding set of equations for each block, namely, for $i = 1, \ldots, p$,
\begin{equation}
\label{eqbasiclog1blocks}
\ln\k^{(i)} + [A^{(i)}\,|\,\bm{1}]\,\left(\begin{array}{r}\ln x\\-\ln\lambda_i\end{array}\right) = \ln\,h_i(\alpha^{(i)})\,.  
\end{equation}
where $h_i(\alpha^{(i)})$ refers to the subvector of $h^{(i)}(\alpha^{(i)})$ indexed by $P_i$. 
More precisely, $(x, \lambda, \alpha, \k)$ together solve \eqref{eqbasiclog1}, if and only if $(x, \lambda_i, \alpha^{(i)}, \k^{(i)})$ together solve \eqref{eqbasiclog1blocks} for each $i = 1, \ldots, p$. Of course, the $i$th equation is precisely the equation we would derive if our goal was to find positive equilibria of the $i$th partition subnetwork (using the trivial partition). 
\end{rem}

The next lemma tells us that associated with the parameterisation $h$ are natural decompositions of $\ker \Gamma$. 
Recall that $h\colon Y \subseteq \R^{m-r-p} \to \R^{m}_+$ is an affine injective map with constant Jacobian matrix, say $J_h$. Define $\mH$ to be the unique $(m-r-p)$-dimensional linear subspace of $\R^m$ parallel to $\mathrm{im}\,h$, i.e., $\mH:=\mathrm{im}\,J_h$ when $m>r+p$ and $\mH = \{0\}$ when $m=r+p$.

\begin{lem}
\label[lemma]{lem:kerG}
Consider a dynamically nontrivial $(n,m,r)$ network  which factors over a $p$-partition $\mP$. Given $c \in \widehat{\mC}$, define $\mS_c :=  \{c \circ (\bm{1}_\mP\lambda): \lambda \in \R^p\}$.
\begin{itemize}
\item[(i)] $\mH \subseteq \ker \Gamma$ and $\mH$ factors over $\mP$.

\item[(ii)] Given $c \in \widehat{\mC}$, $\mH$ and $\mS_c$ have trivial intersection.

\item[(iii)] For any $c \in \widehat{\mC}$, we have $\ker\Gamma = \mH \oplus \mS_c$, so that, for any fixed $\alpha \in Y$,
\begin{equation}
  \label{kerGdecomp}
\ker\Gamma = \{J_hz + (\bm{1}_\mP\lambda)\circ h(\alpha): (z,\lambda) \in \R^{m-r-p} \times \R^p\}\,.
\end{equation}

\end{itemize}

\end{lem}

\vspace{-\baselineskip}\begin{proof}
In the case when $m=r+p$, $\widehat{\mC}$ is a singleton, $\mH = \{0\}$,  and all the claims become immediate; so assume $m>r+p$.

(i) That $\mH \subseteq \ker \Gamma$ is obvious. To see that $\mH$ factors over $\mP$, we observe that $J_h$ has block diagonal structure, namely $\Pi_{P_i} J_h = J_h \Pi_i$, where for each $i$, $\Pi_{P_i}$ is the natural projection of $\R^m$ onto $\R^m_{P_i}$, and $\Pi_i$ is the natural projection of $\mathrm{span}\,Y$ onto $\mathrm{span}\,Y_i$, regarded as a subspace of $\mathrm{span}\,Y$. 

(ii) Suppose the claim is false, and $(\bm{1}_\mP\hat{\lambda}) \circ h(\alpha) = J_hz$ for some $\alpha \in Y$, $\hat{\lambda} \neq 0$ and $z \neq 0$. Choose $\varepsilon > 0$ small enough that (a) $\beta:= \alpha + \varepsilon z \in Y$ and (b) $\bm{1} + \varepsilon \hat{\lambda} \in \R^p_+$. Define $\lambda := \varepsilon \hat{\lambda}$. Then $(\bm{1}_\mP\lambda) \circ h(\alpha) = J_h(\varepsilon z) = h(\beta) - h(\alpha)$; consequently $h(\beta) = (\bm{1}_\mP(\bm{1}+\lambda)) \circ h(\alpha)$. The fact that $h(\alpha), (\bm{1}_{\mP}(\bm{1}+\lambda))\circ h(\alpha) \in \widehat{\mC}$ with $\bm{1}_{\mP}(\bm{1}+\lambda) \neq \bm{1}$ contradicts the definition of $\widehat{\mC}$. 

(iii) Fix $c \in \widehat{\mC}$. By (i), $\mH\subseteq \ker\Gamma$ is an $(m-r-p)$-dimensional linear subspace, while $\mS_c\subseteq \ker\Gamma$ is a $p$-dimensional linear subspace. Moreover, $\mH$ and $\mS_c$ have trivial intersection from (ii). 
\end{proof}

\subsection{Four matrices associated with a general mass action network}\label{subsec:matrices}
Consider a dynamically nontrivial $(n,m,r)$ network that factors over a $p$-partition $\mP$ of $\{1,\dots,m\}$. Define
\[ s_{\mP}:=\mathrm{rank}\,[A\,|\,\bm{1}_\mP] - p \ \geq 0\, .
\] 
In the case $p=1$, $s_{\mP}$ is the source rank of the network. Clearly $s_{\mP} \leq \min\{n,m-p\}$. 

We will be interested in the following matrices $W$, $Q$, $Z$ and $G$ which are not, in general, uniquely defined. Moreover, $W$, $Q$ and $G$ depend on the partition $\mP$; when we wish to emphasise this, we may write $W_\mP$, and so forth. 
\begin{itemize}
\item An $(m-p-s_{\mP})\times m$ full rank matrix $W$, termed a {\bf solvability matrix} for reasons to become clear below, whose rows are a basis of $\ker\,[A\,|\,\bm{1}_\mP]^\top$, i.e., such that
\[W\, [A\,|\,\bm{1}_\mP] = 0\, .\]

\smallskip
\item An $(n+p) \times (n-s_{\mP})$ full rank matrix $Q$ whose columns are a basis of $\ker\,[A\,|\,\bm{1}_\mP]$, i.e., 
\[ [A\,|\,\bm{1}_\mP]\, Q = 0\, . \] 
Also define $\wQ$ to be the $n \times (n-s_{\mP})$ matrix consisting of the first $n$ rows of $Q$, which we will refer to as a {\bf toricity matrix} for reasons to become clear below.  It is easily confirmed that $\ker \wQ = \{0\}$, hence $\rank \wQ = \rank Q = n-s_\mP$.

\smallskip
\item An $(n+p) \times m$ matrix $G$ which is a \textbf{generalised inverse} of $[A\,|\,\bm{1}_\mP]$ \cite{James1978TheGI}, i.e., such that 
\[ [A\,|\,\bm{1}_\mP]\,G\,[A\,|\,\bm{1}_\mP] = [A\,|\,\bm{1}_\mP]\, . \] 
Also define $\wG$ to be the $n \times m$ matrix consisting of the first $n$ rows of $G$. 

\smallskip
\item Recall the $(n-r) \times n$ full rank matrix $Z$  of  \textbf{conservation laws}, whose  rows form a basis of $\ker \Gamma^\top$, i.e., such that \[Z\, \Gamma = 0 \, .\] 

\end{itemize}

\begin{rem}[Empty matrices]
  \label[remark]{rem:emptymat}
Some or all of $W$, $Q$ and $Z$ may be empty: 
\begin{itemize}
 \item $Z$ is empty if and only if the network has full rank, i.e. $r=n$. We explore this special case in some detail in the sections to follow. 
 \item $W$ is empty if and only if $[A\,|\,\bm{1}_\mP]$ is a surjective matrix, i.e. $m=\mathrm{rank}\,[A\,|\,\bm{1}_\mP] = s_\mP+p$. In this case we will say that the sources are \textbf{$\mP$-independent}.  In the case $p=1$, $\mP$-independence is precisely affine independence of the sources. 
 \item $Q$ is empty if and only if $[A\,|\,\bm{1}_\mP]$ is an injective matrix, i.e. $n=s_{\mP}$. In this case, $G$ is a left inverse of $[A\,|\,\bm{1}_\mP]$. 
 \item If $W$ and $Z$ are empty, then so is $Q$: $m=s_\mP +p$ and  $r=n$ together give $s_\mP = m-p \geq r =n$, hence $s_\mP = n$.
 \item If $W$ and $Q$ are empty, then $m=n+p$, consequently $[A\,|\,\bm{1}_\mP]$ is an invertible square matrix, and $G$ is its ordinary inverse.
 \end{itemize}
 \end{rem}

In what follows we will assume that $W, Q$ and $Z$ have been chosen to be integer matrices whenever they are nonempty. 
On the other hand, $G$ has rational entries, and it may or may not be possible to choose it to be an integer matrix. 
 
\begin{rem}\label[remark]{rem:trivial}
Only $Z$ requires knowledge of reaction vectors of the network. The matrices $W$, $Q$ and $G$ are associated only with the multiset of sources of the network and the partition $\mP$. This observation will allow us often to make simultaneous claims for all networks with a given multiset of sources. 
\end{rem}

\begin{rem}[Geometrical interpretation of the toricity matrix $\wQ$]
Geometrically, each column of $\wQ$ is orthogonal to the affine hull of the sources of the reactions in each block of $\mP$. In particular, the columns of $\wQ$ form a basis of the intersection of the orthogonal complements of the directions of each of these affine hulls. Note that the latter may be strictly larger than the orthogonal complement of the direction of the affine hull of the sources.
\end{rem}

\begin{ex}[When $p=1$, $W$ encodes affine dependencies amongst the sources]
Consider the trivial partition and a $2$-species and $5$-reaction network with sources, $[A\,|\,\bm{1}]$, and choice of $W$ as shown:
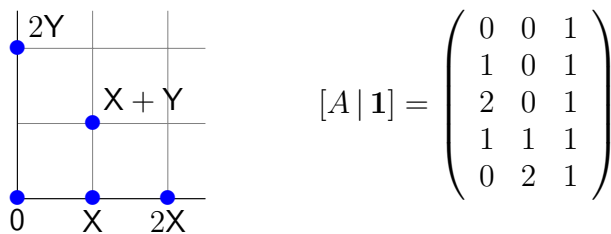
\begin{center}
\begin{tikzpicture}

\draw [step=1, gray, very thin] (0,0) grid (2.5,2.5);
\draw [ -, black] (0,0)--(2.5,0);
\draw [ -, black] (0,0)--(0,2.5);

\node[inner sep=0,outer sep=1] (P1) at (0,0) {\large \textcolor{blue}{$\bullet$}};
\node[inner sep=0,outer sep=1] (P2) at (1,0) {\large \textcolor{blue}{$\bullet$}};
\node[inner sep=0,outer sep=1] (P3) at (2,0) {\large \textcolor{blue}{$\bullet$}};
\node[inner sep=0,outer sep=1] (P4) at (1,1) {\large \textcolor{blue}{$\bullet$}};
\node[inner sep=0,outer sep=1] (P5) at (0,2) {\large \textcolor{blue}{$\bullet$}};

\node [below] at (P1) {$\mathsf{0}$};
\node [below] at (P2) {$\fX$};
\node [below] at (P3) {$2\fX$};
\node [above right] at (P4) {$\fX+\fY$};
\node [above right] at (P5) {$2\fY$};

\node at (6,1.25) {$[A\,|\,\bm{1}] = \left(\begin{array}{ccc}0&0&1\\1&0&1\\2&0&1\\1&1&1\\0&2&1\end{array}\right)$};

\node at (12,1.25) {$W = \left(\begin{array}{rrrrr}-1&2&-1&0&0\\0&0&1&-2&1\end{array}\right)$.};

\end{tikzpicture}
\end{center}
The first row of $W$ tells us about the affine dependency amongst $\mathsf{0}$, $\fX$ and $2\fX$, while the second tells us about the affine dependency amongst $2\fX$, $\fX+\fY$ and $2\fY$. These choices for rows of $W$ are minimal in the sense that $\{\mathsf{0}, \fX, 2\fX\}$ and $\{2\fX, \fX+\fY,2\fY\}$ are the smallest affinely dependent subsets of the sources. 
\end{ex}

Affine independence of the sources implies $\mP$-independence, but if $p > 1$, the converse does not necessarily hold as we see in the following example. 

\medskip
\begin{ex}[Affinely dependent but $\mP$-independent sources]
Consider the following dynamically nontrivial $(2,4,2)$ network
\[
\mathsf{0} \ce{<=>[\k_1][\k_2]} 2\fX, \qquad \fX \ce{<=>[\k_3][\k_4]} \fX + \fY\, .
\]
Obviously the four sources in $\mathbb{Z}^2$ are affinely dependent. However, we can easily check that  the network factors over $\mP = \{\{1,2\},\{3,4\}\}$. We have
\[ \Gamma = \begin{pmatrix}
 2 & - 2 & 0 & 0 \\ 0 & 0 & 1 & -1 
\end{pmatrix}, \quad [A\,|\,\bm{1}_\mP] = \left(\begin{array}{cccc}0&0&1&0\\2&0&1&0 \\ 1&0&0&1\\ 1&1&0&1\end{array}\right)\, . \]
As $[A\,|\,\bm{1}_\mP]$ has rank $4$, $W$ is empty, i.e., the sources are $\mP$-independent.
\end{ex}

The following lemma on the choice of $W$ will be used frequently. Given an $(n,m,r)$ network which factors over a partition $\mP$, let us say that indices $i$ and $j$ are \textbf{\bm{$\mP$}-equivalent} if they belong to the same block of $\mP$ and correspond to reactions with the same source. 
\begin{lem}[A useful choice of $W$]
\label[lemma]{lem:Wstruct}
  Consider an $(n,m,r)$ network $\mN$ which factors over a $p$-partition $\mP$. Let $m_i:=|P_i|$,  $\ell_i$ be the number of distinct sources in $P_i$, and $\ell := \sum_i \ell_i$. Define $s_i$ to be the source rank of the partition subnetwork $\mN_i$, namely, $s_i:=\mathrm{rank}[A^{(i)}\,|\,\bm{1}] - 1$, and let $\widehat{s}:= \sum_i s_i$. Then there exists a matrix $W$ whose rows form a basis $\mB = \mB_1 \sqcup \mB_2 \sqcup \mB_3$ of $\ker\,[A\,|\,\bm{1}_\mP]^\top$ such that:
\begin{itemize}
\item $\mB_1$ consists of $m-\ell$ ``trivial'' vectors, each with two nonzero entries, $1$ and $-1$, supported on a pair of $\mP$-equivalent indices. 
\item $\mB_2$ consists of $\ell-\widehat{s}-p$ vectors, each supported on exactly one block of $\mP$ and with support not including a pair of $\mP$-equivalent indices.
\item $\mB_3$ consists of $\widehat{s} - s_\mP$ vectors, each supported on more than one block of $\mP$, and with support not including a pair of $\mP$-equivalent indices. 
\end{itemize}
\end{lem}
\vspace{-\baselineskip}\begin{proof}
  For each $i$, we can clearly choose $m_i-\ell_i$ linearly independent trivial vectors in $\ker\,[A\,|\,\bm{1}_\mP]^\top$, each supported on $P_i$. Together, these give the $m-\ell$ vectors in $\mB_1$.

  For each $i$, choose any additional $\ell_i-s_i-1$ vectors in $\ker\,[A\,|\,\bm{1}_\mP]^\top$ supported on $P_i$ such that the total of $m_i-s_i-1$ vectors supported on $P_i$ are a linearly independent set. Let $\mB_2'$ be the set formed by the union of these vectors. 
  Choose any additional $\widehat{s}-s_\mP$ vectors forming a set $\mB'_3$, such that $\mB_1 \cup \mB'_2 \cup \mB'_3$ is a basis of $\ker\,[A\,|\,\bm{1}_\mP]^\top$. 
  
 By adding linear combinations of trivial vectors in $\mB_1$ to the vectors in $\mB'_2$ (resp., $\mB'_3$) if necessary, we obtain sets of vectors $\mB_2$ (resp., $\mB_3$) whose support does not include $\mP$-equivalent indices. (Note that the support of a vector in $\mB_3$ may still include indices of two reactions with the same source belonging to different blocks of $\mP$.)
\end{proof}

The following technical lemma simplifies several proofs. 
\begin{lem}
\label[lemma]{lemGsolve}
Consider an $(n,m,r)$ network which factors over a partition $\mP$.
\begin{itemize}
\item[(i)] Given $u \in \mathrm{im}\,[A\,|\,\bm{1}_\mP]$, $v \in \R^{n-s_{\mP}}$, and $z \in \R^n$, satisfying $\wG u + \wQ v = z$, there exists $t \in \R^p$ such that $u = A z  + \bm{1}_\mP t$. 
\item[(ii)] Given $(z,t) \in \R^n \times \R^p$, there exists a unique $v \in \R^{n-s_{\mP}}$ such that $\wG(A z  + \bm{1}_\mP t) + \wQ v = z$.
\end{itemize}
\end{lem}
\vspace{-\baselineskip}\begin{proof}
(i) We may ``augment'' the equation $\wG u + \wQ v = z$ by adding to $\wG$ and $\wQ$ the final $p$ rows of $G$ (resp., $Q$), to obtain
\[
G u + Q v  = \left(\begin{array}{c} z \\t\end{array}\right)\,,
\]
where this equation defines $t \in \R^p$. Multiplying on the left by $[A\,|\,\bm{1}_\mP]$, and observing that $[A\,|\,\bm{1}_\mP]Q=0$, and that $G$ is a generalised inverse of $[A\,|\,\bm{1}_\mP]$, now gives the first result. (ii) The second claim follows by setting $u = A z  + \bm{1}_\mP t$ and considering again the augmented equation
\[
G[A\,|\,\bm{1}_\mP]\left(\begin{array}{c} z \\t\end{array}\right) - \left(\begin{array}{c} z \\t\end{array}\right) = -Qv\,.
\]
The existence of $v$ satisfying this equation follows as the left hand side lies in $\mathrm{im}\,Q = \ker[A\,|\,\bm{1}_\mP]$, by the definition of $G$. The uniqueness of $v$ follows as $\wQ$ is an injective matrix.
\end{proof}

\subsection{Parameterising the equilibrium set}\label{subsec:parameterisation}
With a $p$-partition $\mP$ and the matrices $W,Q$ and $G$ in place, we revisit the main equation \eqref{eqbasiclog1} describing the set of positive equilibria. 
For given $(\alpha, \k) \in Y \times \R^m_+$, \eqref{eqbasiclog1} has a solution $(x, \lambda)\in \R^n_+ \times \R^p_+$ if and only if 
$\ln h(\alpha) - \ln \k \in \im [A \,|\,\bm{1}_\mP]$, hence if and only if the (Fredholm) solvability condition $W(\ln h(\alpha) - \ln \k) = 0$, equivalently,
\begin{equation} \label{eqsolv}
h(\alpha)^W = \k^W 
\end{equation}
is satisfied. We refer to system \eqref{eqsolv} as the {\bf solvability system} of the network, and equations of this system as {\bf solvability equations}. 

We consider \eqref{eqsolv} with $\k \in \R^m_+$ fixed and $\alpha$ varying: it is then a system of $m-s_{\mP}-p$ equations for $m-r-p$ unknowns. Define $\Ak$ to be the set of solutions to \eqref{eqsolv}, i.e.,
\[
\Ak := \{\alpha \in Y\,:\, h(\alpha)^W = \k^W\}\,. 
\]
As $h(\alpha)$ and $\k$ are positive vectors, we may clear denominators in $h(\alpha)^W = \k^W$ to obtain a system of {\em polynomial} equations. Thus $\Ak$ is semialgebraic, being the intersection of an algebraic variety in $\R^{m-s_\mP-p}$ with the open polyhedron $Y \subseteq \R^{m-s_\mP-p}$. Three special cases of importance are: 
\begin{enumerate}[label=(\roman*)]
\item if $s_{\mP}=m-p$ (sources are $\mP$-independent), so that $W$ is empty and solvability of \eqref{eqbasiclog1} is automatic, we have $\Ak = Y$; 
\item if $s_{\mP} < r$, then we will see later that $\Ak$ is generically empty and we will refer to the network as {\bf $\mP$-overdetermined}; 
\item  if $s_{\mP} = r$, we will see later that $\Ak$ is generically finite, in which case we will refer to the network as {\bf $\bm{\mP}$-toric}.
\end{enumerate}

\begin{rem}[Subsystems of the solvability system]
 \label[remark]{remsolvsubsystems}
Let us choose $W$ as in \Cref{lem:Wstruct}, and define $W^{(i)}$ to be the submatrix of $W$ consisting of columns indexed by $P_i$ and rows supported only on $P_i$. With $h_i$ defined as in \Cref{remsubsyst}, \eqref{eqsolv} consists of a set of subsystems (some or all possibly empty) of the form
\begin{equation}
  \label{eqsolvsub}
  h_i(\alpha^{(i)})^{W^{(i)}} = (\k^{(i)})^{W^{(i)}}, \quad i = 1, \ldots, p
\end{equation}
along with further equations $h(\alpha)^{\widetilde{W}} = \k^{\widetilde{W}}$, where $\widetilde{W}$ consists of rows of $W$ supported on more than one block of $\mP$. For any $i \in \{1, \ldots, p\}$, the $i$th equation of \eqref{eqsolvsub} is precisely the solvability system for the $i$th partition subnetwork.
\end{rem}
 
 In what follows, we allow functions and matrices to be empty, and adopt the natural convention that any system of equations which is empty is trivially satisfied. In particular, given a $0 \times m$ matrix $M$, the equation $Mz=0$ is satisfied by all $z \in \R^m$. And if $n=s_{\mP}$, so that $Q$ is an $(n+p) \times 0$ matrix, we interpret $\mu^{\wQ}$ as $\bm{1}$. 

For fixed $\k \in \R^m_+$, define $T_\k\colon Y\times \R^{n-s_{\mP}}_+ \to \R^{m-s_{\mP}-p}$ and $F_\k\colon Y \times \R^{n-s_\mP}_+ \to \R^n$ by
\[
  T_\k(\alpha, \mu) = h(\alpha)^W-\k^W\,, \qquad F_\k(\alpha,\mu) = (h(\alpha)/\k)^{\wG} \circ \mu^{\wQ}\,,
\]
so that $T_\k^{-1}(0) = \Ak \times \R^{n-s_{\mP}}_+$. The next result, which tells us that there is a smooth bijection between $\Ak \times \R^{n-s_{\mP}}_+$ and the positive equilibrium set of \eqref{eqODE}, is key to the results to follow. Special cases of Equation~\eqref{eq:Ek_expression} in \Cref{thminj}(i), derived in a closely related way, appear in \cite[Section 3.1]{banajiborosnonlinearity}, and \cite[Section 2]{BBH2024smallbif}. Equation \eqref{eq:Ek_expression} is structurally equivalent to the system of equations presented in \cite[Theorem 5]{regensburger:gale}, and can be deduced from \cite{bihan-sottile-complete-gale} for the trivial partition. 

\begin{thm}\label{thminj}
Consider a dynamically nontrivial $(n,m,r)$ network which factors over a $p$-partition $\mP$, and fix $\k \in \mathbb{R}^m_+$. Then: 
\begin{enumerate}[label=(\roman*)]
\item The map $F_\k$ is smooth on $Y \times \mathbb{R}^{n-s_{\mP}}_+ $, and an injective, semialgebraic, mapping on $\Ak \times \R^{n-s_{\mP}}_+$. The set $\mathcal{E}_\k$ of positive equilibria of \eqref{eqODE} is $F_\k(\Ak \times \R^{n-s_{\mP}}_+)$, namely, 
\begin{equation}\label{eq:Ek_expression} \mathcal{E}_\k = \{(h(\alpha)/\k)^{\wG}\circ \mu^{\wQ} : \alpha \in \Ak,\,\, \mu \in \mathbb{R}^{n-s_{\mP}}_+\}  \,\, = \bigsqcup_{\alpha \in \Ak } (h(\alpha)/\k)^{\wG}\circ \{ \mu^{\wQ} : \mu \in \mathbb{R}^{n-s_{\mP}}_+ \}\,.
\end{equation}

\item The map $\widetilde{F}_\k\colon Y \times \R^{n-s_\mP}_+ \to \R^{m-s_{\mP}-p} \times \R^{n}$, $(\alpha, \mu) \mapsto (T_\k(\alpha,\mu), F_\k(\alpha, \mu))$, is a smooth injective map such that $J_{\widetilde{F}_k}(\alpha, \mu)$ has trivial kernel for all $(\alpha, \mu) \in Y \times \R^{n-s_\mP}_+$.

\item Given $(\alpha,\mu) \in \Ak \times \R^{n-s_{\mP}}_+$, $\ker J_{T_\k}(\alpha,\mu) \cong \ker J_{g_\k}(x)$, where $x := F_\k(\alpha,\mu)$. In particular, $J_{F_\k}(\alpha,\mu)$ takes $\ker J_{T_\k}(\alpha,\mu)$ isomorphically to $\ker J_{g_\k}(x)$.
 \end{enumerate}
 \end{thm}
\vspace{-\baselineskip}\begin{proof}
  (i) Smoothness of $F_\k$ is obvious. Given $\alpha \in \Ak$, we have $\ln (h(\alpha) / \k) \in \mathrm{im}\,[A\,|\,\bm{1}_\mP]$, hence 
\begin{equation}\label{eq:z}[A\,|\,\bm{1}_\mP]\, G \, \ln (h(\alpha) / \k) = \ln (h(\alpha) / \k)\,.
\end{equation}
Suppose that $F_\k$ fails to be injective on $\Ak \times \R^{n-s_{\mP}}_+$, namely, there exist $(\alpha_1, \mu_1), (\alpha_2,\mu_2) \in \Ak \times \R^{n-s_{\mP}}_+$ such that $(h(\alpha_1)/\k)^{\wG}\circ \mu_1^{\wQ} = (h(\alpha_2)/\k)^{\wG}\circ \mu_2^{\wQ}$. Taking logarithms we obtain 
\begin{equation}\label{eq:aux}\wG\big(\ln h(\alpha_1) - \ln h(\alpha_2)\big) + \wQ\ln(\mu_1/\mu_2) = 0\,.  
\end{equation} 
By \eqref{eq:z}, 
$\ln h(\alpha_1) - \ln h(\alpha_2) \in \mathrm{im}\,[A\,|\,\bm{1}_\mP]$, hence, \eqref{eq:aux} gives by \Cref{lemGsolve}(i) that $\ln h(\alpha_1) - \ln h(\alpha_2) = \bm{1}_\mP t$ for some $t \in \R^p$, namely $h(\alpha_1) = (\bm{1}_\mP e^t)\circ h(\alpha_2)$ implying that $h(\alpha_1) = h(\alpha_2)$, hence $\alpha_1 = \alpha_2$. As  $\widehat{Q}$ is an injective matrix it immediately follows from \eqref{eq:aux} that $\mu_1 = \mu_2$, proving that $F_\k$ is injective on $\Ak \times \R^{n-s_{\mP}}_+$. 

If $x=(h(\alpha)/\k)^{\wG}\circ \mu^{\wQ}$ for $(\alpha, \mu)  \in \Ak \times \mathbb{R}^{n-s_{\mP}}_+$, then, taking logarithms and observing that by \eqref{eq:z} $\ln (h(\alpha) / \k) \in \mathrm{im}\,[A\,|\,\bm{1}_\mP]$, \Cref{lemGsolve}(i) gives that $A(\ln x)= \ln (h(\alpha) / \k) - \bm{1}_\mP t$ for some $t$. With $\lambda=e^{-t}$,  $(x, \lambda, \alpha, \k)$ together satisfy \eqref{eqbasiclog1}, and thus $x \in \Ek$. On the other hand, if $x \in \Ek$, namely, there exists $(\alpha, \lambda) \in Y \times \R^p_+$ such that $(x, \lambda, \alpha, \k)$ together satisfy \eqref{eqbasiclog1}. Then $\alpha \in \Ak$ by the solvability condition and, using \eqref{eq:z}, there exists $\mu\in \R^{n-s_{\mP}}_+$ such that 
\begin{equation}\label{eq:Ek1}
  \left(\begin{array}{r}\ln x\\-\ln \lambda\end{array}\right) - G \ln (h(\alpha) / \k) = Q \ln \mu\,.
\end{equation}
Exponentiating gives that $x=(h(\alpha)/\k)^{\wG}\circ \mu^{\wQ}$. Clearly $\Ak \times \R^{n-s_{\mP}}_+$ and $\Ek = F_\k(\Ak \times \R^{n-s_{\mP}}_+)$ are, when nonempty, both semialgebraic sets. That the graph of $\left.F_\k\right|_{\Ak \times \R^{n-s_{\mP}}_+}$ is a semialgebraic set now follows by raising both sides of the defining equation $y = (h(\alpha)/\k)^{\wG} \circ \mu^{\wQ}$ to the lowest common multiple of the denominators appearing in $\wG$ and $\wQ$ and clearing denominators in the resulting rational system to get a system of polynomial equations with the same positive zero set.

(ii) Smoothness is obvious. To see injectivity, suppose that $\widetilde{F}_\k(\alpha_1, \mu_1) = \widetilde{F}_\k(\alpha_2, \mu_2)$. Then it follows from $T_\k(\alpha_1,\mu_1) = T_\k(\alpha_2,\mu_2)$ that $\ln h(\alpha_1)- \ln h(\alpha_2) \in \mathrm{im}\,[A\,|\,\bm{1}_\mP]$. With this fact in place, $F_\k(\alpha_1, \mu_1) = F_\k(\alpha_2, \mu_2)$ implies that $(\alpha_1, \mu_1) = (\alpha_2, \mu_2)$ exactly as in (i). It remains to prove that $\widetilde{F}_\k$ is an immersion. Writing $x=(h(\alpha)/\k)^{\wG} \circ \mu^{\wQ}$, we have 
$J_{T_\k}(\alpha,\mu) = (D_{h(\alpha)^W}WD_{1/h(\alpha)}J_h,\,0)$ and $J_{F_\k}(\alpha,\mu) = (D_{x}\wG D_{1/h(\alpha)}J_h,\, D_{x}\wQ D_{1/\mu})$. Hence 
the Jacobian matrix of $\widetilde{F}_\k$ at $(\alpha,\mu) \in Y \times \R^{n-s_\mP}_+$  is 
\[
J_{\widetilde{F}_\k}:= \left(\begin{array}{cc}D_{h(\alpha)^W}WD_{1/h(\alpha)}J_h&0\\D_{x}\wG D_{1/h(\alpha)}J_h& D_{x}\wQ D_{1/\mu} \end{array}\right)\,.
\]
This matrix is singular if and only if there exists $(z_1,z_2) \neq (0,0)$ satisfying the pair of equations
\begin{equation}
  \label{eqsolvJac}
WD_{1/h(\alpha)}J_hz_1=0, \qquad \wG D_{1/h(\alpha)}J_h z_1 + \wQ D_{1/\mu}z_2=0\,,
\end{equation}
where we have used the fact that $D_{h(\alpha)^W}$ and $D_{x}$ are nonsingular. The first equation of \eqref{eqsolvJac} tells us that $D_{1/h(\alpha)}J_hz_1 \in \ker W = \mathrm{im}\,[A\,|\,\bm{1}_\mP]$. Using \Cref{lemGsolve}(i), the second equation of \eqref{eqsolvJac} gives that $D_{1/h(\alpha)}J_hz_1 = \bm{1}_{\mP} t$ for some $t$. By \Cref{lem:kerG}(ii), $t=0$, hence $z_1=0$ (as $D_{1/h(\alpha)}J_h$ is injective). It follows that $z_2 = 0$ (as $\wQ D_{1/\mu}$ is injective). Thus $J_{\widetilde{F}_\k}$ has trivial kernel. 

(iii) For notational convenience write $J_{F_\k} = J_{F_\k}(\alpha, \mu)$, $J_{T_\k} = J_{T_\k}(\alpha,\mu)$ and $J_{g_\k} = J_{g_\k}(x)$. 
By part (ii), $J_{F_\k}$ is injective when restricted to $\ker J_{T_\k}$. It thus suffices to show that $J_{F_\k}(\ker J_{T_\k}) = \ker J_{g_\k}$.

For a vector $z=(z_1,z_2) \in  \ker J_{T_\k} \subseteq \R^{m-r-p}\times \R^{n-s_{\mP}}$,  write $u:=D_{1/h(\alpha)}J_h z_1$. 
As above  $z \in \ker J_{T_\k}$ if and only if  $u  \in \im\, [A \,|\, \bm{1}_\mP]$. Using  $J_{F_\k}z = D_{x}\big( \wG u + \wQ D_{1/\mu}z_2\big)$, we obtain
\begin{equation}\label{eq:JFk}
J_{F_\k}(\ker J_{T_\k}) =  \Big\{D_{x}\big( \wG u + \wQ D_{1/\mu}z_2\big) : z_2\in \R^{n-s_{\mP}}, u  \in \im\, [A \,|\, \bm{1}_\mP] \cap \im D_{1/h(\alpha)}J_h \Big\} . 
\end{equation}

Using formula \eqref{eq:jac1} in \Cref{rem:Jacla}, and the fact that $\ker\Gamma$ factors over $\mP$,  $z^*\in \ker J_{g_\k}$ if and only if $D_{h(\alpha)} AD_{1/x}z^*  \in \ker\Gamma$. By \Cref{lem:kerG}(iii) this is in turn equivalent to the existence of $z',t$ such that 
\[ D_{h(\alpha)} AD_{1/x}z^*  =   J_h z' - (\bm{1}_\mP t) \circ h(\alpha),  \quad\text{
equivalently} \quad  AD_{1/x}z^*  =  D_{1/h(\alpha)} J_h z' - \bm{1}_\mP t\, . \]
Hence
\begin{equation}\label{eq:Jgk} 
\ker J_{g_\k} = \big\{ z^*\in \R^n : AD_{1/x}z^* \in \im D_{1/h(\alpha)}J_h \oplus \im \bm{1}_\mP  \big\}. 
\end{equation}
We show that $ J_{F_\k}(\ker J_{T_\k})=\ker J_{g_\k} $ using \eqref{eq:JFk} and \eqref{eq:Jgk}. 
The inclusion $\subseteq$ follows from \Cref{lemGsolve}(i) which gives  that $A ( \wG u + \wQ D_{1/\mu}z_2)= u- \bm{1}_\mP t$ for some $t$. 

To show $\supseteq$, for $z^*\in \ker J_{g_\k}$, write $ A D_{1/x} z^*  =u - \bm{1}_\mP t$ with $u \in \im D_{1/h(\alpha)} J_h$. As $\mu$ is positive, from \Cref{lemGsolve}(ii), there exists a (unique) $z_2 \in \R^{n-s_\mP}$ which solves $\wG u + \wQ D_{1/\mu} z_2 = D_{1/x} z^*$. Multiplying both sides by $D_x$ and noting that $u  \in \im\, [A \,|\, \bm{1}_\mP] \cap \im D_{1/h(\alpha)}J_h$ we see from \eqref{eq:JFk} that $z^*\in  J_{F_\k}(\ker J_{T_\k})$. This completes the proof.
\end{proof}

\begin{rem}[Relationship with results in \cite{regensburger:gale}]
\label[remark]{rem:gale}
As indicated above, equation \eqref{eq:Ek_expression} in \Cref{thminj}(i) follows from recent work \cite[Theorem 5]{regensburger:gale}, after some minor adjustments. In terms of the notation developed here, the authors of \cite{regensburger:gale} consider the image of $h$ (here referred to as $\widehat{\mC}$, and referred to in \cite{regensburger:gale} as the \emph{coefficient set}) as the base space for $\mathcal{A}_\k$, instead of the domain $Y$ of $h$ used here. In addition, the authors impose the constraint that the entries of $h^{(i)}(Y_i)$ add up to one, effectively fixing a particular choice of $\widehat{\mC}$. They obtain the matrix $\widehat{G}$ (called $E$ in that work) via a slightly different construction: however, both approaches lead to the same matrix, in the sense that $E$, like $\wG$, is formed of the top rows of a generalized inverse of $[A\,|\,\bm{1}_\mP]$. Our choice to use $Y$, rather than $\widehat{\mC}$, as the base space for $\mathcal{A}_\k$ is better suited for the purposes of this work, allowing us to write down equations explicitly, and obtain bounds on the number of positive nondegenerate equilibria in various classes of networks. 
\end{rem}

We note that  \Cref{thminj}(i) expresses $\mathcal{E}_\k$ as a  disjoint union of \textit{toric sets}, that is, sets admitting a monomial parameterisation in the positive orthant. These sets are parameterised by $\mu$ and have dimension equal to $\rank \wQ$, namely $n-s_\mP$, explaining the terminology ``toricity matrix'' for $\wQ$. We will get back to toricity in \Cref{sec:triangular}. 

\begin{rem}[Dimension of $\Ek$]
\label[remark]{Ekdim}
By \Cref{thminj}(i), $\left.F_\k\right|_{\Ak \times \R^{n-s_\mP}}$ is an injective semialgebraic map. Consequently $\dim \Ek = \dim (\Ak \times \R^{n-s_\mP}_+) = \dim\Ak + n - s_{\mP}$, where the first equality follows from the dimension-preserving property of an injective semialgebraic map, and the second equality follows as the dimension of the Cartesian product of semialgebraic sets is the sum of their dimensions (see \cite[p56ff]{costesemialgebraic}.)
\end{rem}

We will frequently manipulate systems of equations on some domain to obtain polynomial systems with the same zero sets. For example, we often clear denominators in the solvability system $h(\alpha)^W = \k^W$, equivalent to multiplying both sides of each equation by a positive (polynomial) function. We observe in the following lemma that such manipulations, applied to a square system, preserve degeneracy and nondegeneracy of solutions. Henceforth we will refer to systems of equations obtained via any sequence of manipulations which preserve degeneracy or nondegeneracy of solutions, such as those in \Cref{lem:equiveq}, as {\bf equivalent}.

\begin{lem}[Equivalent systems of equations]
  \label[lemma]{lem:equiveq}
Let $Y \subseteq \R^k$ be open, let $d_1, \ldots, d_k \in \R_+$, and let $f_i, g_i, h_i\colon Y \to \R_+$, $(i=1, \ldots, k)$ be smooth functions. Let $F_1, F_2, F_3 \colon Y \to \R^k$ be defined by
\begin{align*}
  F_1(y) &:=  (f_1(y) - g_1(y), \ \ldots, \ f_k(y)-g_k(y))\\
  F_2(y) &:=  (h_1(y)(f_1(y) - g_1(y)), \ \ldots, \  h_k(y)(f_k(y)-g_k(y)))\\
  F_3(y) &:=  (f_1(y)^{d_1} - g_1(y)^{d_1}, \ \ldots,  \ f_k(y)^{d_k}-g_k(y)^{d_k}) \, .
\end{align*}
Let $F_1(y^*)=0$, equivalently, $F_2(y^*)=0$, equivalently, $F_3(y^*)=0$. Then $\ker J_{F_1}(y^*) \cong \ker J_{F_2}(y^*) \cong \ker J_{F_3}(y^*)$. Hence $y^*$ is a degenerate zero of $F_1$ if and only if it is a degenerate zero of $F_2$ if and only if it is a degenerate zero of $F_3$. 
\end{lem}
\vspace{-\baselineskip}\begin{proof}
We check that $J_{F_2}(y^*) = D_{h(y^*)}J_{F_1}(y^*)$ and $J_{F_3}(y^*) = \widehat{D} J_{F_1}(y^*)$, where $\widehat{D}$ is a diagonal matrix with $\widehat{D}_{ii} = d_if_i^{d_i-1}(y^*) =d_ig_i^{d_i-1}(y^*)$. As $D_{h(y^*)}$ and $\widehat{D}$ are positive diagonal matrices, the claim is immediate. 
\end{proof}

While \Cref{thminj} provides a smooth injective parameterisation of the equilibrium set, it has the drawback that a factor of the parameter space, namely $\Ak$, is implicitly defined (via some linear inequalities and polynomial equations). Nevertheless, as we will see in examples and when we examine special cases, we can fairly often pass to an explicit, perhaps piecewise, parameterisation of the equilibrium set. Cases (i) and (iii) of the next corollary provide instances where this holds.  
  
\begin{cor}\label[corollary]{cor:special} We have the following special cases:
 \begin{itemize}
 \item[(i)] ($\mP$-independence: $W$ is empty.) If $m=s_{\mP}+p$, then $W$ is empty, in which case $\Ak=Y$ and
 \[\mathcal{E}_\k = \bigsqcup_{\alpha \in Y} (h(\alpha)/\k)^{\wG}\circ \{ \mu^{\wQ} : \mu \in \mathbb{R}^{n-s_{\mP}}_+ \} \,. \]
 \item[(ii)] ($\wQ$ is empty.) If $n=s_{\mP}$, then $Q$ is empty and
 \[\mathcal{E}_\k = \bigsqcup_{\alpha \in \mathcal{A}_\k} (h(\alpha)/\k)^{\wG}\,. \]
 \item[(iii)] ($Y$ is a singleton.) If $r=m-p$, then $Y$ is a singleton, say $Y = \{\alpha\}$. If, further, $s_{\mP} \geq r$ (i.e., the network is not $\mP$-overdetermined), then $W$ is also empty, and
 \[ \mathcal{E}_\k = (h(\alpha)/\k)^{\wG}\circ \{ \mu^{\wQ} : \mu \in \mathbb{R}^{n-s_{\mP}}_+ \}\, . \]
 In particular $\mathcal{E}_\k$ is toric, $s_\mP=r$, and the sources are $\mP$-independent.

 \smallskip
 \item[(iv)] (Conditions (i)--(iii) hold simultaneously.) In this case, $n=s_{\mP}=r=m-p$, and $\mathcal{E}_\k$ consists of exactly one equilibrium for all $\k\in \R^m_+$, namely $(h(\alpha)/\k)^{\wG}$, where $\wG$ is now a submatrix of $[A\,|\bm{1}_\mP]^{-1}$. Moreover this unique equilibrium is nondegenerate. 
 \end{itemize}
\end{cor}
\vspace{-\baselineskip}\begin{proof}
  (i) and (ii) are merely applications of \Cref{thminj}(i). In case (iii), we observe that $m=r+p$ implies that $Y$ is a singleton, say $Y = \{\alpha\}$, while $m=r+p$ and $s_{\mP} \geq r$ together imply $s_\mP=r$ ($\mP$-independence), i.e., that $W$ is empty and hence $\mathcal{A}_\k = \{\alpha\}$. (iv) When $n=s_{\mP}=r=m-p$, $[A\,|\,\bm{1}_\mP]$ is square and invertible, and we can solve (\ref{eqbasiclog1}) to get the unique positive equilibrium $(h(\alpha)/\k)^{\wG}$ for each $\k$. Nondegeneracy follows immediately from \Cref{thminj}(iii) (see also \cite{banajiborosnonlinearity} for the direct calculations in the special case where $\mP$ is trivial.)
\end{proof}

\begin{rem}[Full rank networks with $\mP$-independent sources]
 \label[remark]{remfullPindependent}
 A full rank network with $\mP$-independent sources satisfies condition (iv) of \Cref{cor:special}, as
\[
r \leq m-p = s_{\mP} \leq n = r
\]
where the first equality follows from $\mP$-independence. Hence any such network has a unique, positive, nondegenerate, equilibrium for all rate constants. The case of $p=1$ was covered in \cite{banajiborosnonlinearity}.
\end{rem}

\begin{rem}[Posynomials]
 \label[remark]{remposynomials}
In general, $(h(\alpha)/\k)^{\wG}$ is a vector of {\em posynomial} expressions \cite{boyd2004convex} in the variables $h_i$ with rational (not necessarily positive) exponents. Consequently, $(h(\alpha)/\k)^{\wG}$ is a vector of {\em generalised posynomials} \cite[Exercise~4.37]{boyd2004convex} in $\alpha$. When solving for equilibria, we can obtain corresponding systems of {\em polynomial} equations in various ways as we will illustrate by example.
\end{rem}

\subsection{Positive equilibria in stoichiometric classes}

Recall that for networks not of full rank, i.e., such that $r < n$, $Z$ is nonempty and each stoichiometric class is defined by a linear system of the form $Zx = K$ for some constant $K \in \mathbb{R}^{n-r}$. For the following theorem,  recall the notation and definitions surrounding \Cref{thminj}. We adopt the natural convention that a (constant) map on a zero dimensional vector space is nondegenerate.
\begin{thm}
\label{thm:degensolvability}
Consider a dynamically nontrivial $(n,m,r)$ network which factors over a $p$-partition $\mP$. For fixed $\k \in \R^m_+$, $K \in \R^{n-r}$, define $H_{\k,K}\colon Y \times \R^{n-s_\mP} \to \R^{m-s_\mP-p} \times \R^{n-r}$ by
\[
H_{\k,K}(\alpha, \mu) := \left(\begin{array}{c}T_\k(\alpha, \mu)\\[4pt]
ZF_{\k}(\alpha, \mu) - K\end{array}\right) = \left(\begin{array}{c}h(\alpha)^W - \k^W\\Z[(h(\alpha)/\k)^{\wG}\circ \mu^{\wQ}] - K\end{array}\right)\,.
\]
Let $\mS_K \subseteq \R^n_+$ be the positive stoichiometric class defined by $Zx = K$. Then 
\begin{enumerate}[align=left,leftmargin=*, label=(\roman*)]
    \item $F_\k$ is a smooth injective map from the zero set of $H_{\k,K}$ on $Y \times \mathbb{R}_+^{n-s_{\mP}}$, namely from solutions to
\begin{equation}
\label{gensol}
h(\alpha)^W = \k^W, \quad Z[(h(\alpha)/\k)^{\wG}\circ \mu^{\wQ}] = K\, \qquad (\alpha, \mu) \in Y \times \mathbb{R}_+^{n-s_{\mP}}\,,
\end{equation} 
onto $\Ek \cap \mS_K$, the set of positive equilibria on $\mS_K$. 
\item Fix $\k \in \R^m_+$ and $(\alpha,\mu) \in \Ak \times \R^{n-s_{\mP}}_+$, and define $x := F_\k(\alpha, \mu) = (h(\alpha)/\k)^{\wG} \circ \mu^{\wQ}$ and $K:=Zx$. Then 
\[\ker J_{H_{\k,K}}(\alpha, \mu) \cong \ker J_{g_\k}(x) \cap \mathrm{im}\,\Gamma\, ,\] 
with $J_{F_\k}(\alpha, \mu)$ providing the isomorphism. Consequently, $x$ is degenerate as an equilibrium of \eqref{eqODE} if and only if $(\alpha, \mu)$ is degenerate as a solution of \eqref{gensol}.
\end{enumerate}
\end{thm}

\vspace{-\baselineskip}\begin{proof}
The first claim is immediate from \Cref{thminj}(i). In order to show the second claim, for notational convenience write $J_{F_\k} = J_{F_\k}(\alpha, \mu)$, $J_{T_\k} = J_{T_\k}(\alpha,\mu)$, $J_{H_{\k}} = J_{H_{\k,K}}(\alpha,\mu)$ and $J_{g_\k} = J_{g_\k}(x)$ (observe that $J_{H_{\k,K}}(\alpha,\mu)$ is independent of $K$). Define $\mT := \ker J_{H_{\k}} = \ker J_{T_\k} \cap \ker (ZJ_{F_\k})$, and observe that, as $\mathrm{im}\,\Gamma = \ker Z$, $\mT$ is the subspace of $\ker J_{T_\k}$ mapped by $J_{F_\k}$ into $\mathrm{im}\,\Gamma$. Recalling, from \Cref{thminj}(iii), that $J_{F_\k}$ maps $\ker J_{T_\k}$ isomorphically onto $\ker J_{g_\k}$, clearly $\left. J_{F_\k}\right|_{\mT}\colon \mT \to \ker J_{g_\k} \cap \mathrm{im}\,\Gamma$ is an isomorphism. But now the following four statements are equivalent: (a) $(\alpha, \mu)$ is a degenerate solution of \eqref{gensol}; (b) $\mT \neq \{0\}$; (c) $\ker J_{g_\k} \cap \mathrm{im}\,\Gamma \neq \{0\}$; (d) $x$ is a degenerate equilibrium of \eqref{eqODE}. The equivalences (a) $\Leftrightarrow$ (b) and (c) $\Leftrightarrow$ (d) are by definition, while (b) $\Leftrightarrow$ (c) uses the isomorphism between $\mT$ and $\ker J_{g_\k} \cap \mathrm{im}\,\Gamma$. 
\end{proof}

We refer to \eqref{gensol} as the {\bf alternative system of equations} for equilibria on the positive stoichiometric class defined by $Zx=K$. 
\begin{rem}[The alternative system of equations]
Note that \eqref{gensol} is a square system, consisting of $(m-s_{\mP}-p) + (n-r)$ equations in the same number of unknowns $(\alpha, \mu)$. Recall that if $W$ is empty we take the first equation of (\ref{gensol}) to be trivially satisfied, and if $Q$ is empty we define $\mu^{\wQ} = \bm{1}$. Similarly, if $Z$ is empty, we take the second equation of (\ref{gensol}) to be satisfied. \Cref{cor:special}(iv) covers the case where both $W$ and $Z$ are empty. \Cref{thm:degensolvability}(ii) tells us that the level of degeneracy of a solution of the alternative system is precisely the level of degeneracy of the corresponding equilibrium. This remains true if we manipulate the alternative system of equations to obtain an equivalent system as in \Cref{lem:equiveq}.
\end{rem}

\begin{rem}[Noninteger entries in $G$]
 Even when $G$ has noninteger entries (see \Cref{remposynomials}), we can always write down a system of polynomial equations, perhaps with additional variables and equations, whose solutions (if necessary projected onto $Y \times \mathbb{R}_+^{n-s_{\mP}}$) include the solutions of system \eqref{gensol}. The most natural or efficient way to do this may not be clear. In the results and examples to follow, all such manipulations lead to equivalent systems as in \Cref{lem:equiveq}, hence preserve degeneracy and nondegeneracy of solutions. 
\end{rem}

\subsection{Alternative B\'ezout and BKK bounds}
Whenever solutions of \eqref{gensol} are also solutions of some equivalent polynomial system, we obtain a new B\'ezout bound on the number of positive nondegenerate equilibria on any stoichiometric class. We refer to any such bound as an {\bf alternative B\'ezout bound}. If the polynomial system is defined on a domain not intersecting the coordinate hyperplanes, we can apply the BKK theorem to improve on the alternative B\'ezout bound: such a bound is termed an {\bf alternative BKK bound}. Any such bound derived using knowledge of the sources and rank of a CRN alone is termed a {\bf source bound}.

\begin{rem}[Nonuniqueness of the matrices $Z$, $W$, $Q$ and $G$]
Different choices of the various matrices lead to different alternative systems of equations hence, potentially, different alternative bounds. Obviously the choices leading to the lowest alternative bounds are desired. In various special cases, the best choice is obvious, but this is not always true. 
\end{rem}

The following examples illustrate the construction and use of the alternative systems of equations, and alternative bounds. 

\begin{ex}
\label[example]{ex353}
Consider any dynamically nontrivial $(3,5,3)$ network with sources
\[
\mathsf{0},\quad 2\fX, \quad 2\fY, \quad \fX+\fZ ,\quad \fY+\fZ\,.
\]
The naive B\'ezout source bound   is $8$ and the naive BKK source bound is $6$ (we may check the latter claim, for example, using the \texttt{Mathematica} function {\tt BKKRootBound} with the {\tt Toric} option).
We have
\[
[A\,|\,\bm{1}] = \left(\begin{array}{rrrr}0&0&0&1\\2&0&0&1\\0&2&0&1\\1&0&1&1\\0&1&1&1\end{array}\right)\,, \qquad W=\left(\begin{array}{ccccc} 0 &  1 &   -1 &   -2 &  2 \end{array}\right)   \, . 
\]
 The solvability system $h(\alpha)^W = \k^W$ is thus equivalent to the single polynomial equation
\[
h_2(\alpha)h_5(\alpha)^2 - \k^Wh_3(\alpha)h_4(\alpha)^2 = 0\,,
\]
which is at-most-cubic in $\alpha$, giving an alternative B\'ezout source bound of $3$, a considerable improvement on the naive BKK source bound. Moreover, for any particular network with these sources, identifying parameter regions for multistationarity, and parameter sets for bifurcations such as fold and cusp bifurcations, reduces to examining a single univariate, at-most-cubic equation, considerably simplifying the analysis.
\end{ex}

\begin{ex}
 \label[example]{ex352a}
Consider the following $(3,5,2)$ network
\[
\mathsf{0} \ce{<=>[\k_1][\k_2]} 2\fX\,, \qquad 2\fY \ce{<=>[\k_3][\k_4]} 2\fZ\,, \qquad   \fZ \ce{->[\k_5]} \fX + \fY\,.
\]
The best possible naive B\'ezout bound is $4$, as is the naive BKK bound.
We have
\[
\Gamma = \left(\begin{array}{rrrrr}2&-2&0&0&1\\0&0&-2&2&1\\0&0&2&-2&-1\end{array}\right) \quad \mbox{and} \quad A = \left(\begin{array}{cccc}0&0&0\\2&0&0\\0&2&0\\0&0&2\\0&0&1\end{array}\right)\,.
\]
In this case $\ker\Gamma$ is spanned by  the vectors $\{(1,1,0,0,0),(0,0,1,1,0),(0,1,1,0,2)\}$, and $\ker \Gamma$, hence $\mC$, factor only over the trivial partition. We set
\[
h(\alpha_1, \alpha_2) = (\alpha_1,\,  1-\alpha_2, \,  1-\alpha_1,\, \alpha_2, \, 2(1-\alpha_1-\alpha_2))^\top\,,
\]
defined on the triangle $Y:= \{(\alpha_1,\alpha_2) \in\mathbb{R}^2_+\,:\, \alpha_1+\alpha_2 < 1\}$. As $[A\,|\,\bm{1}]$ is an injective matrix, $Q$ is empty. Setting $W= (1,  0, 0, 1, -2)$, the solvability system rearranges to
\[
h_1(\alpha)h_4(\alpha) = \hat{\k}\, h_5(\alpha)^2, \quad \mbox{i.e.,} \quad \alpha_1\alpha_2 = 4\hat{\k}\, (1-\alpha_1-\alpha_2)^2
\]
where $\hat{\k} = \k_1\k_4/\k_5^2$. So, in this case,
\[
\Ak = \{(\alpha_1, \alpha_2) \in Y : \alpha_1\alpha_2 = 4\hat{\k}\, (1-\alpha_1-\alpha_2)^2\}\,.
\]
For the following choice of generalised inverse of $[A\,|\,\bm{1}]$, the set of positive equilibria $\Ek$ is parameterised by
\[
\left(\begin{array}{c}x\\y\\z\end{array}\right) = \left(\frac{h(\alpha)}{\k}\right)^{\wG} = \k^{-\wG}\circ \left(\begin{array}{c}\alpha_1^{-\nicefrac{1}{2}}(1-\alpha_2)^{\nicefrac{1}{2}}\\\alpha_1^{-\nicefrac{1}{2}}(1-\alpha_1)^{\nicefrac{1}{2}}\\\alpha_1^{-\nicefrac{1}{2}}\alpha_2^{\nicefrac{1}{2}}\end{array}\right)\,, 
\quad\text{where}\quad 
G = \left(\begin{array}{ccccc}-\nicefrac{1}{2} & \nicefrac{1}{2} &0&0&0\\
-\nicefrac{1}{2}&0&\nicefrac{1}{2}&0&0\\
-\nicefrac{1}{2}&0&0&\nicefrac{1}{2}&0\\1&0&0&0&0\end{array}\right)\,,
\]
for $(\alpha_1,\alpha_2) \in \Ak$. Choosing $Z = (0,1,1)$, positive equilibria on the stoichiometric class $y+z = K$ correspond to the common solutions, in $Y$, to
\[
\theta_2\alpha_1^{-\nicefrac{1}{2}}(1-\alpha_1)^{\nicefrac{1}{2}} + \theta_3\alpha_1^{-\nicefrac{1}{2}}\alpha_2^{\nicefrac{1}{2}} = K, \qquad \alpha_1\alpha_2(1-\alpha_1-\alpha_2)^{-2} = 4\hat{\k}\,,
\]
where $\theta = \k^{-\wG}$.  These are equivalent to a pair of quadratic equations in $(\alpha_1,\alpha_2)$, and so we obtain an alternative B\'ezout bound of $4$, which is precisely the naive bound in this case. The alternative BKK bound corresponding to this pair of quadratic equations is also $4$. 
\end{ex}

In the next example, the source rank and rank of the networks considered are equal, namely the networks are $\mP$-toric with the trivial partition. ($\mP$-toric networks are treated in more generality in \Cref{sec:triangular}.) 
\begin{ex}
 \label[example]{ex342a}
Consider any $(3,4,2)$ network with sources
\[
\fX+\fY,\quad 2\fY,\quad 2\fX,\quad \fY+\fZ\,.
\]
The naive B\'ezout source bound in this case is $4$, and the naive BKK source bound can be checked to be $2$. 
Corresponding to the trivial partition we have $W = (-2,1,1,0)$ and $Q = (1,1,1,-2)^\top$. The solvability system $h(\alpha)^W = \k^W$ is thus equivalent to
\begin{equation}
\label{eq342solve}
h_2(\alpha)h_3(\alpha) - \k^Wh_1(\alpha)^2 = 0\,,
\end{equation}
which is at-most-quadratic in $\alpha$. For each $\alpha^* \in Y$ which solves \eqref{eq342solve}, we have a ray of equilibria
\[
\left(\begin{array}{c}x\\y\\z\end{array}\right) = \left(\frac{h(\alpha^*)}{\k}\right)^{\wG}\circ \mu^{\wQ} = \left(\frac{h(\alpha^*)}{\k}\right)^{\wG}\circ \left(\begin{array}{c}\mu\\\mu\\\mu\end{array}\right)\,.
\]
The equation for positive equilibria on a stoichiometric class, corresponding to $\alpha = \alpha^*$, takes the form $Z[\left(h(\alpha^*)/\k\right)^{\wG}\circ \mu^{\wQ}] = K$, which is clearly linear in $\mu$. We obtain an alternative B\'ezout source bound  of $2$, and the mass action network admits at most two positive nondegenerate equilibria on any stoichiometric class. (Observe that we did not actually need to compute $G$ in order to arrive at this bound, see also \Cref{remtoricBezout} below.) 

The parameter region for multistationarity is easily determined in any particular $(3,4,2)$ network with these sources. For example, consider the CRN
\[
\fX+\fY \ce{<=>[\k_1][\k_2]} 2\fY,\quad 2\fX \ce{->[\k_3]} \fX + \fZ,\quad \fY+\fZ \ce{->[\k_4]} 2\fX\,.
\]
In this case, 
we may set $Z = (1,1,1)$ and $h(\alpha) = (1,\,\alpha,\,1-\alpha,\,1-\alpha)^\top$ with $\alpha \in (0,1)$. The solvability system is equivalent to
\[
\alpha^2-\alpha + \k^W = 0, \quad \mbox{so that} \quad \Ak = \{\alpha \in (0,1)\,:\, \alpha^2-\alpha + \k^W = 0\}\,.
\]
Clearly, $\Ak$ consists of two points if and only if $\k^W = \k_2\k_3/\k_1^2 < 1/4$. The parameter region for multistationarity is thus $\{\k \in \mathbb{R}^4_+\,:\, 4\k_2\k_3<\k_1^2\}$, and we can see via \Cref{thm:degensolvability} that in fact {\em every} stoichiometric class contains two positive nondegenerate equilibria whenever $4\k_2\k_3<\k_1^2$, a unique positive degenerate equilibrium whenever $4\k_2\k_3=\k_1^2$, and no positive equilibria otherwise. The equation $4\k_2\k_3=\k_1^2$ thus defines the parameter set for potential fold bifurcations on each stoichiometric class. We remark that despite the simplicity of this network, these conclusions are not entirely trivial if we work directly with the mass action and conservation equations.
\end{ex}

\subsection{The finest partition} \label{sec:finest} 
In the constructions above, the matrices $W$, $Q$ and $G$ depend on choosing a partition over which a  network factors. Different choices thus give rise to different parameterisations of the set of positive equilibria. We now discuss further the role of the partition, and highlight the finest partition as an advantageous choice. \cite[Section 2]{regensburger:gale} already hints briefly at this fact.

As we will use several partitions in this subsection, we indicate the choice by subscripting the relevant matrices. We start with a basic observation. 

\begin{lem}\label[lemma]{lem:finer}
 Let $\mP' \leq \mP$ be two partitions over which a network factors, of sizes $p'$ and $p$ respectively. Then $s_{\mP} \leq s_{\mP'}$, and further 
 \[ \rank \wQ_\mP \geq \rank \wQ_{\mP'}, \qquad \rank \wQ_\mP - \rank W_\mP = \rank \wQ_{\mP'}
 - \rank W_{\mP'} + (p-p') \, . \]
\end{lem}
\vspace{-\baselineskip}\begin{proof}
The matrix $[A\, |\, \mathbf{1}_\mP]$ has $p-p'$ more columns than
$[A\, |\, \mathbf{1}_{\mP'}]$ and $\im [A\, |\, \mathbf{1}_{\mP'}] \subseteq \im [A\, |\, \mathbf{1}_\mP]$. 
Hence for some $0\leq \ell \leq p-p'$ we have
\[ \rank\, [A\, |\, \mathbf{1}_\mP] =\rank \, [A\, |\, \mathbf{1}_{\mP'}] + \ell. \]
It follows immediately that $s_{\mP} \leq s_{\mP'}$, and that 
\[ \rank \wQ_\mP = \rank \wQ_{\mP'} + (p-p') - \ell, \qquad \rank W_\mP = \rank W_{\mP'} - \ell \]
 from where the remaining statements are immediate.
\end{proof}

\begin{rem}[The effect of choosing a finer partition on the parameterisation of $\Ek$]
  \label[remark]{remfinerP}
  For a finer partition, the dimension of $Y$, namely, the number of variables in the solvability system \eqref{eqsolv}, decreases. \Cref{lem:finer} tells us   that a finer partition may provide a description of the set of positive equilibria where the toricity matrix $\wQ$ has larger rank and hence a larger part of the equilibrium set is explained by a toric set. If this does not occur, then the rank of $W$ must decrease, and hence the solvability system has fewer equations. In either case, choosing the finest possible partition can be advantageous (see \Cref{exmoretoric,exnomoretoric} below).
\end{rem}

 The \textbf{finest partition} of the network, namely the unique finest partition over which a dynamically nontrivial network factors,  is termed the \textit{matroid partition} in \cite{feliu:toric} where partitions are considered in the general setting of vertically parameterised systems. It has previously appeared in the context of reaction networks, see \cite[Remark 3.4]{bryan:partition}. We denote the finest partition by $\mP_f$ and its cardinality by $p_f$.

Given a dynamically nontrivial network with stoichiometric matrix $\Gamma$, by our remarks on partitions, the positive flux cone of the network, being a relatively open subset of $\ker\Gamma$, factors over some partition $\mP$ if and only if $\ker\Gamma$ factors over $\mP$, equivalently, $(\ker\Gamma)^\perp$, namely the row space of $\Gamma$, factors over $\mP$. The finest partition of a network is thus precisely the finest partition over which $\ker\Gamma$ factors, and can be computed as follows. Let $M$ be a row-reduced matrix whose rows form a basis of $\ker\Gamma$, or of the row span of $\Gamma$. 
Then $M$ is, up to permutation of its rows and columns, a block-diagonal matrix, and two indices $i,j\in \{1,\dots,m\}$ are in the same block of $\mP_f$ if and only if $i$ and $j$ are both in the support of some row of $M$. It is shown in \cite[Lemma 4.6]{feliu:toric} that the partition obtained in this way does not depend on the choice of $M$. 
The following lemma gathers together some of the observations above.

\begin{lem}\label[lemma]{lem:refine}
The positive flux cone $\mC$ of a dynamically nontrivial network factors over a $p$-partition $\mP$ if and only if the same holds for $\ker \Gamma$. There is a unique partition, denoted $\mP_f$ and termed the finest partition of the network, over which $\ker\Gamma$ and $\mC$ factor, and which refines any other partition $\mP$ over which $\ker\Gamma$ and $\mC$ factor. 
\end{lem}

We conclude with two examples illustrating the effects of the choice of partition on the resulting parameterisation, in particular, the two scenarios discussed in \Cref{remfinerP}.

\begin{ex}[Finest partition, increased toricity]
  \label[example]{exmoretoric}
Consider the following $(3,4,2)$ network
\[
2\fX \ce{->[\k_1]} \fX+\fY\, , \qquad \fY \ce{->[\k_2]} \fX\, , \qquad 2\fY \ce{->[\k_3]} \fY+\fZ\, , \qquad \fY+\fZ \ce{->[\k_4]} 2\fY\,.
\]
A basis of $\ker \Gamma$ is $\{(1,1,0,0),(0,0,1,1)\}$, and so there are two partitions to consider: the finest partition $\mP_f=\{\{1,2\},\{3,4\}\}$ and the trivial partition. We can set $Z=(1,1,1)$. 

For the trivial partition, say $\mP$, $W_\mP$ and $Q_\mP$ are empty and $G = G_\mP = [A\,|\,\bm{1}]^{-1}$. Setting $h(\alpha) = (\alpha, \alpha, 1-\alpha, 1-\alpha)$ the equilibrium set is 
\[ \mathcal{E}_\k=\{(h(\alpha)/\k)^{\wG} :\alpha \in (0,1)\} =
\Big\{
(\tfrac{\sqrt{1-\alpha} \, \k_2 }{ \sqrt{\alpha} \sqrt{\k_1\, \k_3}}, 
\tfrac{(1-\alpha) \, \k_2}{\alpha \, \k_3 }, 
\tfrac{(1-\alpha) \, \k_2}{ \alpha\, \k_4}
)
:\alpha \in (0,1)
\Big\} . \]
Defining $q\colon  (0,1) \to \R_+$ by $q(\alpha):=\sqrt{(1-\alpha)/\alpha}$, the alternative system \eqref{gensol} takes the form $a \, q(\alpha)^2 + b \, q(\alpha) - c=0$ with $a, b, c > 0$; 
as the quadratic $at^2+bt-c$ has one positive (nondegenerate) zero, and $q$ is a smooth, strictly decreasing function, clearly the alternative system has exactly one nondegenerate solution. The network thus has exactly one positive equilibrium on each stoichiometric class for any choice of rate constants, and this equilibrium is nondegenerate.

For the finest partition $\mP_f$, $Y$ is a singleton and we readily realise that $\mathcal{E}_\k$ is toric (see \Cref{cor:special}(iii)). As we may choose $Q_{\mP_f} = (1, 2, 2, -2)^\top$, without computing $\wG_{\mP_f}$ we see that system \eqref{gensol} takes the form $a \mu^2 + b \mu - K = 0$ with $a,b,K>0$, which clearly always has exactly one positive nondegenerate solution (see also \Cref{thmnnn1}(iii) below). 
\end{ex}

\begin{ex}[Finest partition, no increase in toricity]
  \label[example]{exnomoretoric}
 Consider the following $(2,4,2)$ network:
\[
\mathsf{0} \ce{->[\k_1]} \fY\,, \qquad \fX + \fY \ce{->[\k_2]}   \fX\,, \qquad \fX \ce{<=>[\k_3][\k_4]} 2\fY\,.
\]
A basis of $\ker \Gamma$ is $\{(1,1,0,0),(0,0,1,1)\}$ and so   the finest partition is $\mP_f= \{\{1,2\},\{3,4\}\}$. 
For the trivial partition, say $\mP$, $Q_\mP$ is empty and hence solutions to the solvability system correspond to positive equilibria. Choosing $W_\mP= (1,2,-2,-1)$, gives an alternative B\'ezout source bound of $3$. However, in this case, the solvability system takes the form
\[
\frac{\k_1\k_2^2}{\k_4\k_3^2} = \left(\frac{\alpha}{1-\alpha}\right)^3\,, \quad \alpha \in (0,1)\, ,
\]
which has a unique nondegenerate solution for any choice of rate constants. Thus the network has a unique positive nondegenerate equilibrium for all $\k \in \R^4_+$. On the other hand, $[A\,|\,\bm{1}_{\mP_f}]$ has full rank, hence $W_{\mP_f}$ and $Q_{\mP_f}$ are empty, \Cref{cor:special}(iv) applies, and we conclude without further analysis that there is a unique positive nondegenerate equilibrium for all $\k \in \R^4_+$.
\end{ex}

\section{Special cases}
With the mathematical framework in place, we derive further results for several particular classes of networks, using \Cref{thminj,,thm:degensolvability}. 

\subsection{$\mP$-overdetermined and source deficient networks}\label{secdegen}

 Let us fix an $(n,m,r)$ mass action network and a partition $\mP$. If $s_{\mP}<r$, then the solvability system \eqref{eqsolv} has more equations than variables: recall that in this case we refer to the network as $\mP$-overdetermined. We show that positive equilibria of a $\mP$-overdetermined mass action network exist only for an exceptional set of rate constants and, when they do exist, are always strongly degenerate (i.e., the network is strongly degenerate). Moreover, for parameters at which positive equilibria exist, the positive equilibrium set has dimension greater than $n-r$. We remark that a network being $\mP$-overdetermined is sufficient, but not necessary, for these conclusions (see \Cref{exdegen1} below).

 In the special case where $\mP$ is the trivial partition, we refer to a $\mP$-overdetermined network as {\bf source deficient}. Importantly, to decide if a network is source deficient, we require knowledge only of its sources, and not its flux cone.
 Source deficient networks are ones with too many affine dependencies amongst their sources, i.e., whose sources span an affine space of dimension less than the dimension of the span of the reaction vectors. For example, a rank-$2$ network is source deficient if its sources lie on a line. 

We remark that in the case of $(n,m,m-1)$ networks, source deficiency is equivalent to affine dependence amongst the sources of the network (see also \cite[Remark 3.2]{banajiborosnonlinearity} for the special case of $(n,n+1,n)$ networks.)

The main result on $\mP$-overdetermined networks, extending Lemma~6 in \cite{BBH2024smallbif}, is \Cref{thmdegen}. We start with a well-known lemma. Hausdorff dimension is denoted $\mathrm{dim}_\mathrm{H}$. 
\begin{lem}
\label[lemma]{lemH}
Let $U \subseteq \mathbb{R}^k$ be open and let $f\colon U \to \mathbb{R}^m$ be a $C^1$ diffeomorphism (i.e., a $C^1$ map with $C^1$ inverse) of $U$ onto $f(U)$. Then given any $A \subseteq U$, $\mathrm{dim}_\mathrm{H}f(A) = \mathrm{dim}_\mathrm{H}A$.
\end{lem}
\vspace{-\baselineskip}\begin{proof}
Note first that, restricted to any compact subset of $U$, $f$ is clearly bi-Lipschitz onto its image. Let $\{U_i\}_{i=1}^\infty$ be any sequence of compact $k$-dimensional subsets of $\mathbb{R}^{k}$ such that $\cup_{i}U_i = U$. Let $A_i = A \cap U_i$ and $\mathcal{A}_i := f(A_i)$ so that $\mathcal{A}:=f(A) = \cup_i \mathcal{A}_i$. Then, by the bi-Lipschitz stability of Hausdorff dimension \cite[Prop 3.3(b)]{falconerfractal}, $\mathrm{dim}_\mathrm{H}\mathcal{A}_i = \mathrm{dim}_\mathrm{H}A_i$ for each $i$, and hence
\[
\mathrm{dim}_\mathrm{H}\mathcal{A} = \mathrm{dim}_\mathrm{H}\cup_i\mathcal{A}_i = \sup\nolimits_i\left\{\mathrm{dim}_\mathrm{H}\mathcal{A}_i\right\} = \sup\nolimits_i\left\{\mathrm{dim}_\mathrm{H}A_i\right\} = \mathrm{dim}_\mathrm{H}\cup_i A_i = \mathrm{dim}_\mathrm{H}A\,,
\]
where the second and fourth equalities follow from the countable stability of Hausdorff dimension \cite[p49]{falconerfractal}. 
\end{proof}

We now prove \Cref{thmdegen}, the main theorem on $\mP$-overdetermined networks. We note that Claim (iv), which follows easily from \Cref{thminj}(i), already implies Claims (i)-(iii)  by \cite[Thm 3.1]{feliu:dimension}. However, we present a self-contained proof.
 
\begin{thm}
\label{thmdegen}
Consider a $\mP$-overdetermined $(n,m,r)$ mass action network, i.e., a network which factors over a partition $\mP$ of $\{1, \ldots, m\}$, and such that $s_{\mP} < r$. Then
\begin{enumerate}[label=(\roman*)]
\item The set of $\k \in \mathbb{R}^m_+$ for which $\Ek \neq\emptyset$ has $m$-dimensional Lebesgue measure $0$.
\item For generic $\k \in \mathbb{R}^m_+$, $\Ek = \emptyset$. 
\item The network is strongly degenerate. 
\item For $\k \in \mathbb{R}^m_+$ such that $\Ek \neq \emptyset$, $\Ek$ has dimension at least $n-r+1$.
\end{enumerate}

\end{thm}

\vspace{-\baselineskip}\begin{proof}
(i) Let $p:=|\mP|$. Observe that $\mH_{\mP}:= \mathrm{im}\,[A\,|\,\bm{1}_\mP] = \ker W$ has dimension $s_{\mP}+p < r+p$, so its orthogonal complement $\mH_{\mP}^\perp$ has dimension $q:=m-s_{\mP}-p$. For brevity, let $k:= m-r-p$, so that by assumption $k<q$. Recall that $Y \subseteq \mathbb{R}^{k}$, the domain of $h$, is open and note that clearly $\hat{h}\colon Y \to \mathbb{R}^m$ defined by $\hat{h}(y) = \ln h(y)$ is a $C^1$ diffeomorphism onto its image. Then, by \Cref{lemH}, $\mathrm{dim}_\mathrm{H}\mathcal{Y} = k$, where recall $\mathcal{Y} = \hat{h}(Y)$. Let $\Pi\colon \R^m \to \mH_{\mP}^\perp$ be the orthogonal projection onto $\mH_{\mP}^\perp$. As $\Pi$ is clearly Lipschitz, and $\mathrm{dim}_\mathrm{H}$ is nonincreasing under Lipschitz maps \cite[Prop 3.3(a)]{falconerfractal}, $\mathrm{dim}_\mathrm{H}\Pi(\mathcal{Y}) \leq k < q$.

  Now note that the set of rate constants for which solutions to \eqref{eqbasiclog1} exist is
  \[
\{\k \in \mathbb{R}^m_{+} : (\ln \k+ \mH_{\mP}) \cap \mathcal{Y} \neq \emptyset\} = \{\k \in \mathbb{R}^m_{+}: (\ln \k+ \mH_{\mP}) \cap \mH_{\mP}^\perp \in \Pi(\mathcal{Y})\} = \mathrm{exp}(\Pi(\mathcal{Y}) \times \mH_{\mP})\,.
\]
We identify $\mH_{\mP}$ with $\mathbb{R}^{m-q}$ and note that $\mathrm{dim}_\mathrm{H}\left(\Pi(\mathcal{Y}) \times \mathbb{R}^{m-q}\right) = \mathrm{dim}_H\Pi(\mathcal{Y}) + m-q$ (see, for example, \cite{Tricot_1982}). 
As a consequence,
\[
\mathrm{dim}_\mathrm{H}\left(\mathrm{exp}(\Pi(\mathcal{Y}) \times \mathbb{R}^{m-q})\right) = \mathrm{dim}_\mathrm{H}\left(\Pi(\mathcal{Y}) \times \mathbb{R}^{m-q}\right) = \mathrm{dim}_H\Pi(\mathcal{Y}) + m-q<m\,,
\]
where the first equality follows from \Cref{lemH}, and the final inequality from above. Hence $\mathrm{exp}(\Pi(\mathcal{Y}) \oplus \mH_{\mP})$ has $m$-dimensional Hausdorff measure $0$, and consequently $m$-dimensional Lebesgue measure $0$. This confirms that the system admits positive equilibria for only an exceptional set of rate constants.

(ii)   As the set $\{(x,\k) \in \R^n_+ \times \R^m_+:\Gamma(\k \circ x^A) = 0\}$ is semialgebraic, the Tarski-Seidenberg theorem  \cite[Prop. 5.2.1]{RAG} tells us that so is its projection onto $\R^m_+$; namely, the set $\mathcal{Z}$ of rate constants for which positive equilibria exist is semialgebraic. As $\mathcal{Z}$ has zero Lebesgue measure by (i), it does not contain an open ball and consequently its semialgebraic dimension is strictly smaller than $m$. The same holds for its closure in $\R^m_+$ in the Zariski topology \cite[Prop. 2.8.2]{RAG}. Therefore $\mathcal{Z}$ is contained in a proper algebraic subvariety of $\R^m_+$ giving (ii). 

(iii)   The claim that all positive equilibria are degenerate would be an easy corollary of (ii) via the implicit function theorem. To show strong degeneracy, recall from \eqref{eq:strongdeg} that it is enough to show that $J':= \Gamma \,D_{v} \,A$ has rank strictly smaller than $r$ for each fixed $v\in \mC$. The condition $\mathrm{rank}\,[A\,|\,\bm{1}_\mP]<r+p$, holds if and only if either (1) $\mathrm{rank}\,A < r$; or (2) $\mathrm{rank}\,A := \hat{r} \geq r$ and $\mathrm{dim}(\mathrm{im}\bm{1}_\mP \cap \mathrm{im}\,A) > \hat{r} - r$. We consider each case:
\begin{enumerate}
\item[(1)] If $\mathrm{rank}\,A<r$, then clearly $\mathrm{rank}\,J'<r$, since $J'$ includes $A$ as a factor. 
\item[(2)] Assume that $\mathrm{rank}\,A:= \hat{r} \geq r$ and there exist linearly independent $\{w_1,\ldots,w_{\hat{r}-r+1}\} \subseteq \mathrm{im}\bm{1}_\mP \cap \mathrm{im}\,A$.
The first assumption gives the existence of linearly independent vectors $v_1,\dots,v_{n-\hat{r}}\in \ker A\subseteq \ker J'$. For each $i \in \{1, \ldots, \hat{r}-r+1\}$, choose $\hat{v}_i$ such that $w_i=A\hat{v}_i$. Clearly $\{\hat{v}_1,\ldots,\hat{v}_{\hat{r}-r+1}\}$ are linearly independent. The fact that $w_i \in \mathrm{im}\bm{1}_\mP$ and $\ker \Gamma$ factors over $\mP$ implies that $D_{v}w_i\in \ker \Gamma$ for each $i$. Thus $\hat{v}_i \in \ker J'\backslash\ker A$ for each $i$, and so $v_1,\dots,v_{n-\hat{r}}, \hat{v}_1, \ldots, \hat{v}_{\hat{r}-r+1}$ are linearly independent vectors in $\ker J'$. This implies that $\mathrm{dim}\,\ker J' \geq n-r+1$, and hence $\rank J' < r$ as claimed. 
\end{enumerate}
 
(iv) Let $\k$ be a rate constant such that positive equilibria exist, namely, \eqref{eqbasiclog1} admits a solution $(x_0, \lambda_0)$ for some $\alpha \in Y$. By \Cref{thminj}(i), the equilibrium set then includes the branch
\[
\{x_0 \circ \mu^{\wQ}\,:\,\mu \in \mathbb{R}^{n-s_{\mP}}_+\} \, ,
\]
which has dimension $n-s_{\mP} > n-r$ by \Cref{lemH} as the map $\mu \mapsto x_0 \circ \mu^{\wQ}$ is a smooth $C^1$ diffeomorphism onto its image. 
\end{proof}

\begin{rem}[Dimension of the equilibrium set in $\mP$-overdetermined networks]
\Cref{thmdegen}(iv) tells us that the positive equilibrium set $\Ek$ of a $\mP$-overdetermined network has dimension greater than $n-r$ when it is nonempty. In the special case $r=n$ (networks of full rank), this means that, when nonempty, $\Ek$ has positive dimension, and it immediately follows that positive equilibria are not isolated and hence {\em must} be degenerate, equivalently, strongly degenerate. In the case where $r < n$, when $\Ek$ is nonempty, $\dim \Ek + r > n$, and so we expect intersections of $\Ek$ and stoichiometric classes generally to have positive dimension when nonempty. \Cref{ex:degenisolated} demonstrates that a zero dimensional intersection is possible, though \Cref{thmdegen} tells us that even in this case equilibria on this class are strongly degenerate.
\end{rem}

\begin{ex}[A source deficient network with an isolated equilibrium on a stoichiometric class]
 \label[example]{ex:degenisolated}
 The following dynamically nontrivial $(3,4,2)$ network is source deficient, having rank $2$, but only two distinct sources:
\begin{center}
\begin{tikzpicture}[scale=1.2]
 \node at (0,0) {$\fX + 2\fY + \fZ$};
 \node[rotate=30] at (1.5,0.6) {$\ce{->[\k_1]}$};
 \node at (1.5,0.1) {$\ce{->[\k_2]}$};
 \node[rotate=-30] at (1.5,-0.4) {$\ce{->[\k_3]}$};
 \node at (2.8,0.8) {$2\fX + 2\fY$};
 \node at (3.0,0) {$2\fX + \fY + \fZ$};
 \node at (2.8,-0.8) {$2\fY + 2\fZ$};
 \node at (7,0.05) {$2\fX + 3\fY + 2\fZ \ce{->[\k_4]}  2\fX + 4\fY + \fZ\,.$};
\end{tikzpicture}

\end{center}
Stoichiometric classes have equation $x+y+z = K$ for $K\in \R_+$. Setting $\k_1=\k_2=\k_4=1$, $\k_3=2$ the reader may confirm that positive equilibria are the positive solutions to $x\, y\, z=1$. Hence there is a unique equilibrium on the stoichiometric class defined by $x+y+z=3$. By \Cref{thmdegen}, this equilibrium is strongly degenerate.
\end{ex}

It is not necessary for a network to be $\mP$-overdetermined for it to be strongly degenerate, as illustrated by the next example.
\begin{ex}[A strongly degenerate network which is not $\mP$-overdetermined]
 \label[example]{exdegen1}
Consider the dynamically nontrivial $(2,4,2)$ network
\[
2\fX \ce{->}  3\fX,\quad 2\fX \ce{->} \fX+\fY,\quad 0 \ce{->} \fY,\quad 2\fY \ce{->} \fY\,.
\]
The reader may immediately check that $\mP_f$ is the trivial partition  and the network is not source deficient. However, the network is degenerate; equivalently, as the network has full rank, it is strongly degenerate.
\end{ex}

\begin{rem}[$\mP$-overdetermined networks are $\mP_f$-overdetermined]
 \label[remark]{overdet}
By \Cref{lem:finer}, $s_{\mP_f}$ is minimal amongst $\{s_{\mP}\}$ for all partitions $\mP$ over which the network factors, as $\mP_f$ refines them all. It follows   that if a network is $\mP$-overdetermined for any partition $\mP$, then it also is $\mP_f$-overdetermined. The converse is not true as illustrated by the following example. 
\end{rem}

\begin{ex}[A $\mP_f$-overdetermined network that is not source deficient]
 \label[example]{ex:Pfoverdet}
Consider the following dynamically nontrivial $(2,4,2)$ network:
\[
\fX \ce{->} \fX + \fY, \quad 2\fX \ce{->} \fX + \fY, \quad 2\fY \ce{->} \fX + \fY, \quad \fY \ce{->} \mathsf{0}\,.
\]
The reader may easily confirm that the finest partition is $\mP_f=\{\{1,4\},\{2,3\}\}$ and the network is $\mP_f$-overdetermined as $s_{\mP_f} = 1 < 2=r$. However, the network is not source deficient as the source rank is $s  = 2 = r$. 
\end{ex}

Even if a network is not $\mP$-overdetermined, it may still occur that some subnetwork is overdetermined, giving rise to conclusions similar to those of \Cref{thmdegen}. Recall from \Cref{lem:degensubnet} that a network is strongly degenerate if it has a strongly partition subnetwork. This occurs, in particular, if some partition subnetwork is source deficient.

\begin{thm}\label{thm:linear_few}
 Assume that an $(n,m,r)$ network $\mN$ factors over a $p$-partition $\mP$ of $\{1, \ldots, m\}$. Choose $W$ as in \Cref{lem:Wstruct}, fix $i$, and let $w_i$ denote the number of rows of $W$ supported on $P_i$. If, for some $i$, $\dim \mC_i \leq w_i$, then $\mN$ does not admit positive equilibria for generic $\k$ in $\R^m_+$, and is strongly degenerate.
\end{thm}
\vspace{-\baselineskip}\begin{proof}
  Let $m_i = |P_i|$, and consider  the $(n,m_i, m_i-\dim \mC_i)$ partition subnetwork  $\mN_i$ corresponding to $P_i$. The source rank of $\mN_i$ satisfies
  \[ s_i:=\rank \, [ A^{(i)} | \bm{1}] - 1 \leq m_i-w_i - 1\, ,  \]
as, by assumption, $\ker\,[ A^{(i)} \, |\,  \bm{1}]^\top$ contains at least $w_i$ linearly independent vectors. As the rank of $\mN_i$ is $r_i:=m_i-\dim \mC_i$, we have 
  \[ s_i \leq m_i-w_i-1 \leq  m_i-\dim \mC_i-1 < r_i \, , \]
 i.e., $\mN_i$ is source deficient. As a consequence, by \Cref{thmdegen}(ii), $\mN_i$ does not admit positive equilibria for generic $\kappa^{(i)}$, and the same holds for $\mN$ for generic $\kappa$. By \Cref{thmdegen}(iii), $\mN_i$ is strongly degenerate, and the same holds for $\mN$ by \Cref{lem:degensubnet}. 
\end{proof}

We have the following immediate corollary of \Cref{thm:linear_few} which will be needed later. 

\begin{cor}\label[corollary]{cor:linear_few}
 Assume that an $(n,m,r)$ network factors over a $p$-partition $\mP$ of $\{1, \ldots, m\}$. If, for some $i\in \{1,\dots,p\}$, $\dim \mC_i \leq m_i-\ell_i$, where $\ell_i$ is the number of distinct sources among the reactions in $P_i$, then the network does not admit positive equilibria for generic $\k$ in $\R^m_+$, and is strongly degenerate. 
\end{cor}

\subsection{$\mP$-toric networks}\label{sec:triangular} 
If $s_{\mP}=r$, then the solvability system $h(\alpha)^W = \k^W$ has the same number of equations as unknowns. Recall that we refer to such networks as $\mP$-toric. This terminology is explained by the fact that if the solvability system has a finite number of solutions, then by \Cref{thminj}(i) $\mathcal{E}_\k$ is a finite union of toric sets, a notion referred to as {\bf local toricity} in \cite{feliu:toric}. 

\begin{rem}[$\mP$-toric networks automatically give alternative B\'ezout bounds]
  \label[remark]{remtoricBezout}
For any fixed $\alpha$, the system $Z[(h(\alpha)/\k)^{\wG}\circ \mu^{\wQ}] = K$ can be written as a system of $n-r$ {\em polynomial} equations in the $n-r$ parameters $\mu$. Thus the equations $h(\alpha)^W = \k^W$ and $Z[(h(\alpha)/\k)^{\wG}\circ \mu^{\wQ}] = K$ give an alternative B\'ezout bound regardless of whether $G$ is an integer matrix or not: this is the product of the bounds obtained from $h(\alpha)^W = \k^W$ considered as a (square) polynomial system in $\alpha$, and $Z(h(\alpha)/\k)^{\wG}\circ \mu^{\wQ} = K$ considered as a (square) polynomial system in $\mu$ and treating $\alpha$ as constant. The alternative B\'ezout bound in \Cref{ex342a} was of this form.
\end{rem}

Recall that $\mathcal{A}_\k$, the set of $\alpha \in Y$ solving $h(\alpha)^W = \k^W$, depends, via $W$, on the chosen partition $\mP$. The following result confirms that under the assumption of nondegeneracy, $\mathcal{A}_\k$ cannot be generically finite unless the network is $\mP$-toric. See \cite[Thm 5.3]{feliu:toric} for a more algebraic approach where toricity properties are extended to $\mathbb{C}$.

\begin{thm}[Finiteness of $\Ak$]
 \label{thm:finite}
 Consider an $(n,m,r)$ network  which factors over a partition $\mP$.
 \begin{enumerate}[label=(\roman*)]
     \item If $s_{\mP} \leq r$, then $\# \mathcal{A}_\k<+\infty$ generically for $\k \in \R^m_+$.
     \item If $s_{\mP} > r$, then $\# \mathcal{A}_\k<+\infty$ generically for $\k \in \R^m_+$ if and only if the network is strongly degenerate.
 \end{enumerate}
 In particular, if the network is not strongly degenerate, then $\#\mathcal{A}_\k<+\infty$ generically for $\k \in \R^m_+$ if and only if $s_{\mP}=r$. 
 \end{thm}
\vspace{-\baselineskip}\begin{proof}
We start with three observations: (1) If the network is strongly degenerate, then by \Cref{prop:nondeg} $\Ek=\emptyset$ generically for $\k \in \R^m_+$ and hence $\# \mathcal{A}_\k=0<+\infty$ generically; (2) $\dim \Ek=\rank(\wQ)=n-s_\mP$ if and only if $\Ak$ is finite and nonempty: this follows as $\dim \Ek = \dim \Ak  + n-s_\mP$ (see \Cref{Ekdim}). (3) If the network is not strongly degenerate, then the set $\mZ$ of rate constants for which $\Ek\neq \emptyset$ has nonempty interior, and $\dim \Ek=n-r$ for generic $\k\in \mZ$ by \Cref{prop:nondeg}. 

With this in place, (1) covers: the $\mP$-overdetermined case, namely, $s_{\mP} < r$ (recall from \Cref{thmdegen} that $\mP$-overdetermined networks are strongly degenerate); the case $s_\mP= r$ with the network strongly degenerate; and the reverse implication in (ii). 

If $s_\mP= r$ and the network is not strongly degenerate, then (3) tells us that   $\dim \Ek=n-s_\mP$ for generic $\k$ in $\mZ$, and hence  generically $\#\Ak<+\infty$ by (2). This completes the proof of (i). 

If $s_\mP> r$ and the network is not strongly degenerate then $\dim\Ek=n-r>n-s_\mP$ generically for $\k\in \mZ$ by (3), and hence by (2), $\Ak$ is infinite generically in $\mZ$. This completes proof of (ii).

The last statement is immediate from (i), (ii) and \Cref{thmdegen}.
\end{proof}

We say that a network \textbf{admits a monomial parameterisation} if there exists a nonnegative integer $\ell$ and a matrix $B\in \Z^{n \times \ell}$ such that for all $\k\in \R^m_+$, either $\Ek=\emptyset$ or $\Ek=\{ x^*_\k \circ \mu^B : \mu \in \R^{\ell}_+\}$ for some $x^*_\k \in \R^n_+$. We say the network \textbf{admits generically a monomial parameterisation}, if this holds for $\k$ in a nonempty Zariski open subset of $\R^m_+$. Observe that by definition, a strongly degenerate network is considered here to admit generically a monomial parameterisation.

Networks which admit a monomial parameterisation include deficiency zero networks \cite{Craciun2009toric,Feinberg1972complex}, networks satisfying the hypotheses of the deficiency one theorem \cite{Feinberg1995class}, networks with toric steady states \cite{Millan2012toricsteadystates}, and the networks studied in \cite{conradi2019total}. By \cite{Millan2012toricsteadystates}, see also \cite{Muller2015injectivity,sadeghimanesh:multi}, deciding upon multistationarity (the existence of multiple positive equilibria in a stoichiometric class), reduces to a simple check involving only the exponent matrix $B$ and $\Gamma$. 

Clearly, if $\# \mathcal{A}_\k\leq 1$ for all $\k \in \R^m_+$, then the network admits a monomial parameterisation with exponent matrix being the toricity matrix $\wQ$ by \Cref{thminj}(i). The next result shows that for a network that is not strongly degenerate, and working with the finest partition, the existence of a monomial parameterisation in fact implies that $\# \mathcal{A}_\k\leq 1$ for all $\k \in \R^m_+$.

\begin{thm}\label{thm:toric}
A $(n,m,r)$ network that is not strongly degenerate admits  a monomial parameterisation if and only if, with the finest partition $\mP_f$, $\# \mathcal{A}_\k\leq 1$ for all $\k \in \R^m_+$. Additionally, in this case the exponent matrix is any choice of $\wQ_{\mP_f}$.
\end{thm}
\vspace{-\baselineskip}\begin{proof}
The reverse implication follows from \Cref{thminj}. For the forward implication, as the network is not strongly degenerate, \cite[Sec 4]{feliu:toric} gives that if the network admits a monomial parameterisation, then necessarily the exponent matrix $B$ consists of the first $n$ rows of any matrix whose columns form a basis of $\ker\,[ A \, | \, \bm{1}_{\mP_f}]$. Hence the result follows again by \Cref{thminj}. 
\end{proof}

\begin{rem}
When $\mP_f$ is the trivial partition, \cite[Theorem 15]{muller_deshpande} gives conditions for when $\Ak$ has exactly one element for all $\k$. By \Cref{thm:toric}, under these conditions, the network admits a monomial parameterisation for all $\k$. In \cite[Section 3.1]{regensburger:gale}, 
it is shown that for the trivial case where $Y$ consists  of one point and $W$ is empty, that is, $m=s_{\mP} + p = r+p$, explicit binomial equations for $\Ek$ can be easily obtained.
\end{rem}

\begin{thm}\label{thm:toric2}
Consider a $\mP_f$-toric, $(n,m,r)$ network with exactly $r+p_f$ distinct sources. Suppose further that each pair of repeated sources occurs for reactions in the same block of $\mP_f$. Then the solvability system is equivalent to a (square) linear system, and hence the network admits generically a monomial parameterisation. 
\end{thm}
\vspace{-\baselineskip}\begin{proof}
As the network is $\mP_f$-toric, $\mathrm{rank}\,W = m-p_f-r$. The assumptions of the theorem imply that $W$, chosen as in \Cref{lem:Wstruct}, has rows each corresponding to a pair of repeated sources, and hence with two nonzero entries ($1$ and $-1$). After clearing denominators, the solvability system $h(\alpha)^W = \k^W$ is a linear system. As $\mathcal{A}_{\k}$ is generically finite by \Cref{thm:finite}, it must, generically, have cardinality $0$ or $1$. The result now follows from \Cref{thminj}(i).
\end{proof}

In the context of \Cref{thm:toric2}, deciding whether a network admits a monomial parameterisation reduces to investigating the parametric determinant of the coefficient matrix of the solvability system. \Cref{thm:toric,thm:toric2} are illustrated in the following example. 
 
\begin{ex}[Dual phosphorylation--dephosphorylation networks]
\label[example]{ex:2site}
 We consider first the irreversible dual phosphorylation network (adapted e.g. from \cite{Wang.2008aa}), which is a $(9,8,6)$ network:
\begin{align*}
\fX_1+\fX_3 \ce{->[\k_1]} \fX_6 \ce{->[\k_2]} \fX_1+ \fX_4 & \ce{->[\k_5]} \fX_8 \ce{->[\k_6]} \fX_1+ \fX_5 \\
\fX_2+\fX_5 \ce{->[\k_7]} \fX_9 \ce{->[\k_8]} \fX_2+\fX_4 &\ce{->[\k_3]} \fX_7 \ce{->[\k_4]} \fX_2+ \fX_3 \, . 
\end{align*}
We compute that $\ker \Gamma$ has the basis $\{(1,1,1,1,0,0,0,0),(0,0,0,0,1,1,1,1)\}$, and hence we deduce that $\mP_f= \{ \{1,2,3,4\} , \{5,6,7,8\}\}$, so that $m=r+p_f$. 
Simple computations give $\rank [A\,|\,\bm{1}_{\mP_f}]=8$, hence $s_{\mP_f}=6=r$. A computation using the criterion in \eqref{eq:nondeg} shows that the network is nondegenerate and hence \Cref{thm:finite} tells us that $\mathcal{A}_\k$ is generically finite. In fact, \Cref{cor:special}(iii) applies, $Y$ consists of one point and, by \Cref{thm:toric}, the network admits a monomial parameterisation. 
 
To find it, we need to find a matrix $\wQ$ and $v = \mathrm{im}\,h$. We can take $v=\{(1,1,1,1,1,1,1,1)\}$. We choose
{\small \[G = \begin{pmatrix}
 0 & 0 & 0 & 0 & 0 & 0 & 0 & 0 
\\
 0 & 0 & 0 & 0 & 0 & 0 & 0 & 0 
\\
 0 & 0 & 0 & 0 & 0 & 0 & 0 & 0 
\\
 -1 & 0 & 1 & 0 & 0 & 0 & 0 & 0 
\\
 -1 & 0 & 1 & 0 & -1 & 0 & 1 & 0 
\\
 -1 & 1 & 0 & 0 & 0 & 0 & 0 & 0 
\\
 -1 & 0 & 0 & 1 & 0 & 0 & 0 & 0 
\\
 -1 & 0 & 1 & 0 & -1 & 1 & 0 & 0 
\\
 -1 & 0 & 1 & 0 & -1 & 0 & 0 & 1 
\\
 1 & 0 & 0 & 0 & 0 & 0 & 0 & 0 
\\
 1 & 0 & -1 & 0 & 1 & 0 & 0 & 0 
\end{pmatrix}, \qquad Q= 
\begin{pmatrix}
1 & 0 & 0 
\\
 0 & 1 & 0 
\\
 0 & 0 & 1 
\\
 1 & -1 & 1 
\\
 2 & -2 & 1 
\\
 1 & 0 & 1 
\\
 1 & 0 & 1 
\\
 2 & -1 & 1 
\\
 2 & -1 & 1 
\\
 -1 & 0 & -1 
\\
 -2 & 1 & -1 
\end{pmatrix}.
\]}%
Removing the last two rows of $G$ and $Q$ gives $\wG$ and $\wQ$, and in particular 
\begin{align*}
  (v / \k)^{\wG} &= \left( 1, 1, 1, \tfrac{\k_{1}}{\k_{3}}, \tfrac{\k_{1} \k_{5}}{\k_{3} \k_{7}}, 
\tfrac{\k_{1}}{\k_{2}}, \tfrac{\k_{1}}{\k_{4}}, 
\tfrac{\k_{1} \k_{5}}{\k_{3} \k_{6}}, \tfrac{\k_{1} \k_{5}}{\k_{3} \k_{8}}\right)\,, 
\\ 
\mu^{\wQ} &= \left(\mu_{1}, \mu_{2}, \mu_{3}, 
\tfrac{ \mu_{1} \mu_{3}}{ \mu_{2}}, 
\tfrac{ \mu_{1}^{2} \mu_{3}}{ \mu_{2}^{2}}, 
 \mu_{1} \mu_{3}, 
 \mu_{1} \mu_{3}, 
\tfrac{ \mu_{1}^{2} \mu_{3}}{ \mu_{2}}, 
\tfrac{ \mu_{1}^{2} \mu_{3}}{ \mu_{2}}\right).
\end{align*} 
Hence, 
\[\mathcal{E}_\k = \left\{ \left(\mu_{1}, \mu_{2}, \mu_{3}, 
\frac{\k_{1} \mu_{1} \mu_{3}}{\k_{3} \mu_{2}}, 
\frac{\k_{1} \k_{5} \mu_{1}^{2} \mu_{3}}{\k_{3} \k_{7} \mu_{2}^{2}}, 
\frac{\k_{1} \mu_{1} \mu_{3}}{\k_{2}}, 
\frac{\k_{1} \mu_{1} \mu_{3}}{\k_{4}}, 
\frac{\k_{1} \k_{5} \mu_{1}^{2} \mu_{3}}{\k_{3} \k_{6} \mu_{2}}, 
\frac{\k_{1} \k_{5} \mu_{1}^{2} \mu_{3}}{\k_{3} \k_{8} \mu_{2}}\right) : \mu \in \R^3_{+} \right\}, \]
 which is the usual parameterisation of the set of positive equilibria for this network.

 \medskip 
 Consider now the more common version of the network with reversible reactions as follows:
 \begin{align*}
\fX_1+\fX_3 \ce{<=>[\k_1][\k_9]} \fX_6 \ce{->[\k_2]} \fX_1+ \fX_4 & \ce{<=>[\k_5][\k_{10}]} \fX_8 \ce{->[\k_6]} \fX_1+ \fX_5 \\
\fX_2+\fX_5 \ce{<=>[\k_7][\k_{11}]} \fX_9 \ce{->[\k_8]} \fX_2+\fX_4 &\ce{<=>[\k_3][\k_{12}]} \fX_7 \ce{->[\k_4]} \fX_2+ \fX_3 \, . 
\end{align*}
We still have that the network is nondegenerate
and $s_{\mP_f}=6=r$. We no longer have that $Y$ is a point, but there are $r+p_f=8$ distinct sources, and reactions with the same source belong to the same block of $\mP_f$. 
This can be seen as 
\begin{align*}
 & \{(1,1,1,1,0,0,0,0,0,0,0,0),(1,0,0,0,0,0,0,0,1,0,0,0),(0,0,1,0,0,0,0,0,0,0,0,1),\\ & \qquad (0,0,0,0,1,1,1,1,0,0,0,0),(0,0,0,0,1,0,0,0,0,1,0,0),(0,0,0,0,0,0,1,0,0,0,1,0)\}
\end{align*} is a basis of $\ker \Gamma$ yielding  $\mP_f= \{ \{1,2,3,4,9,12\} , \{5,6,7,8,10,11\}\}$. 
Hence, \Cref{thm:toric2} applies and the network admits generically a monomial parameterisation. 

As in the proof of \Cref{thm:toric2}, the solvability system gives rise to a system of linear equations, and hence to confirm that the network admits a monomial parameterisation, we can simply inspect the linear system. Specifically, we can choose
\[
Y:= \{(\alpha_1,\alpha_2,\alpha_3,\alpha_4) \in\mathbb{R}^4_+\,:\, \alpha_1+\alpha_2 <1,\, \alpha_3+\alpha_4 < 1\}
\]
and the solvability system is then
\[ \tfrac{\alpha_1}{\alpha_2} = \tfrac{\k_2}{\k_9}, \qquad\tfrac{\alpha_1}{1-\alpha_1-\alpha_2} = \tfrac{\k_4}{\k_{12}}, \qquad \tfrac{\alpha_3}{\alpha_4} = \tfrac{\k_6}{\k_{10}}, \qquad \tfrac{\alpha_3}{1-\alpha_3-\alpha_4} = \tfrac{\k_8}{\k_{11}}\,.\]
This is easily checked to have exactly one (nondegenerate) solution in $Y$ for all $\k \in \R^{12}_+$, i.e., $\#\Ak = 1$ for all $\k \in \R^{12}_+$.
\end{ex}

\begin{rem}
\label[remark]{remodebase}
 Of the $31$ nondegenerate networks in the database ODEbase \cite{odebase2022} reported in \cite{feliu:toric} to satisfy $s_{\mP_f}=r$  (i.e., to be $\mP_f$-toric in the terminology of this paper), $14$ satisfy that, with the finest partition, $W$ is empty and hence \Cref{thm:toric} applies, giving the existence of a monomial parameterisation; and $13$ fall into the setting of \Cref{thm:toric2} and admit generically a monomial parameterisation. Only $4$ of the networks are therefore not covered by this simple analysis. In \cite{feliu:toric} it is shown they also admit generically a monomial parameterisation. 
\end{rem}

\begin{rem}
In \cite[Section 6]{feliu:toric}, once it has been established that a network is nondegenerate and $s_\mP=r$, the generic existence of a monomial parameterisation is investigated via the \emph{coset counting system} $\Gamma(\k\circ x^A)=0$, $\wQ x -b=0 $ for any fixed $b\in \wQ(\R^{n-s_\mP}_+)$. This provides an alternative approach to studying the cardinality of $\Ak$. For some networks, one approach might be easier to pursue than the other. For example, \Cref{thm:toric,thm:toric2} are straightforward to apply, while the coset counting system admits simple analyses such as the injectivity test \cite{Muller2015injectivity}. 
\end{rem}

\begin{rem}[Bounds from the toricity matrix alone for networks admitting a monomial parameterisation]
 \label[remark]{remmonopar}
When $\mathcal{A}_\k$ is a singleton $\{\alpha_\k\}$ whenever nonempty, finding the positive equilibria in the stoichiometric class corresponding to total amount $K$ reduces to solving 
\begin{equation}\label{eq:Z}
 Z[ (h(\alpha_\k)/\k)^{\wG}\circ \mu^{\wQ}] = K\,, \qquad \mu \in \mathbb{R}_+^{n-s_\mP}.
\end{equation} 
Recall from \Cref{thm:degensolvability}, that solutions to \eqref{eq:Z} are nondegenerate if and only if they correspond to nondegenerate equilibria. If the network is nondegenerate (hence $r=s_\mP$ by \Cref{thm:finite}), \eqref{eq:Z} is a square system that will have a finite number of solutions for $(\k, K)$ in a set with nonempty Euclidean interior. In this case, we get alternative BKK and Bézout bounds depending only on $Z$ and $\wQ$. In particular, the normalised volume of the convex hull of the rows of $\wQ$ provides a new upper bound on the cardinality of positive, nondegenerate equilibria on any stoichiometric class.
\end{rem}

We build on \Cref{remmonopar} in the following theorem. 
\begin{thm}\label{thm:bounds_toric}
Consider an $(n,m,r)$ network which factors over a partition $\mP$, and suppose that $\# \mathcal{A}_\k\leq 1$ for all $\k$. 
Let $I\subseteq \{1,\dots,n\}$ be the set of indices that are in the support of at least one row of $Z$, and let $q$ be the cardinality of $\{ \wQ_i : i\in I \}\cup \{0\}$, where $\wQ_i$ is the $i$th row of $\wQ$. Then:
\begin{enumerate}[label=(\roman*)]
\item The normalised volume of the convex hull of $\{ \wQ_i : i\in I \}\cup \{0\}$ is an upper bound on the number of 
positive nondegenerate equilibria on any stoichiometric class. 
 \item If $q= n-r+1$, then for each $\k$, and on each stoichiometric class, the network has 
 at most one positive nondegenerate  equilibrium. 
 \item If $q= n-r+2$, then for each $\k$, and on each stoichiometric class, the network has at most $n-r+1$ 
  positive  nondegenerate equilibria. 
\end{enumerate}
\end{thm}
\vspace{-\baselineskip}\begin{proof}
As $\#\Ak\leq 1$ for all $\k$, positive equilibria in stoichiometric classes are given by \eqref{eq:Z}. A (Laurent) monomial in $\mu^{\wQ}$ appears in \eqref{eq:Z} if it corresponds to a nonzero column of $Z$ and this fact is independent of the choice of $Z$.

If the network is strongly degenerate, all claims hold trivially. If the network is not strongly degenerate, then by \Cref{thm:finite}, $s_\mP=r$ and hence \eqref{eq:Z} is a square system.  The first claim now follows from the BKK theorem, and the other two claims follow from \Cref{prop:fewnomials}. In all three cases we use \Cref{thm:degensolvability} as in \Cref{remmonopar} to guarantee that nondegenerate (resp., degenerate) solutions of \eqref{eq:Z} correspond to nondegenerate (resp., degenerate) equilibria.
\end{proof}

Observe that any network such that $r=m-p$, and which is not $\mP$-overdetermined, falls into the setting of \Cref{thm:bounds_toric} by \Cref{cor:special}(iii). In particular, an $(n,m,m-1)$ network can only factor over the trivial partition, and so such a network which is not source deficient falls into this setting. Also in this setting are $(n,n+p-1,n-1)$ networks which are not $\mP$-overdetermined. These are inspected next.

\begin{thm}
 \label{thmnnn1}
Consider an $(n,m,n-1)$ network which is not $\mP_f$-overdetermined and satisfies $n\leq m\leq n+p_f-1$. Then $p_f=m-n+1$, and the network is $\mP_f$-toric and has $\mP_f$-independent sources. Further:
\begin{itemize}
\item[(i)] The network admits no more than $n$ positive nondegenerate equilibria on any stoichiometric class.
\item[(ii)] If there exists $\omega \in \ker \Gamma^\top\cap \mathbb{R}^n_{\geq 0}\backslash\{0\}$ (in particular, if stoichiometric classes are bounded), the network admits no more than two positive nondegenerate equilibria on any stoichiometric class.
\item[(iii)] If there exists $\omega \in \ker \Gamma^\top\cap \mathbb{R}^n_{\geq 0}\backslash\{0\}$, and we can choose $\wQ_{\mP_f}$ (which is a column vector) such that whenever $\omega_i\neq 0$, $(\wQ_{\mP_f})_i\geq 0$, then the network admits at most one positive nondegenerate equilibrium on any stoichiometric class. 
\end{itemize}
\end{thm}
\vspace{-\baselineskip}\begin{proof}
Note first that
\[
n+p_f-1 \leq s_{\mP_f}+p_f = \mathrm{rank}\,[A\,|\,\bm{1}_{\mP_f}] \leq m \leq n + p_f - 1\,,
\]
where the first inequality follows as the network is not $\mP_f$-overdetermined, and the final inequality is by assumption. It follows that $p_f=m-n+1$. We are now in the setting of \Cref{cor:special}(iii), namely the sources are $\mP_f$-independent, and the network is $\mP_f$-toric; consequently $Y$ is a singleton, say $Y = \{\alpha\}$.

Additionally, $\wQ_{\mP_f}$ is a column vector and $Z$ is a row-vector; hence, setting $v=h(\alpha)$, by \Cref{thm:degensolvability} positive nondegenerate equilibria on any stoichiometric class are in smooth one-to-one correspondence with the positive nondegenerate solutions $\mu$ of the equation
\[
Z[(v/\k)^{\wG}\circ \mu^{\wQ}] - K =0 \,.
\]
After clearing denominators if necessary, this is a univariate polynomial equation in $\mu$, say $F(\mu)=0$, with at most $n+1$ terms. Parts (i)-(iii) now follow by Descartes' rule of signs. For (ii) and (iii) note that as all entries of $Z$ are nonnegative, $K$ must be positive in positive stoichiometric classes and hence $F(\mu)$ has exactly one negative coefficient. For (iii), this negative coefficient is the constant coefficient. In all cases, we use \Cref{thm:degensolvability} to transfer claims about solutions of the alternative system to positive equilibria of the network.
\end{proof}

The networks in \Cref{thmnnn1} are the only nondegenerate networks such that there is one conservation law, $Y$ is a point, and $\wQ$ has one column. Setting $\mP$ to be the trivial partition, we get an immediate corollary of \Cref{thmnnn1}.
\begin{cor}
\label[corollary]{cornnn1}
(i) $(n,n,n-1)$ networks admit no more than $n$ positive nondegenerate equilibria on any stoichiometric class. (ii) $(n,n,n-1)$ networks with bounded stoichiometric classes admit no more than two positive nondegenerate equilibria on any stoichiometric class.
\end{cor}
 
\begin{ex}[Illustrating the claim in \Cref{thmnnn1}(iii)]
\label[example]{exBIOMD}
We consider the following bimolecular $(5,6,4)$ network, which is a modification of the network \texttt{BIOMD0000000405} from the \texttt{Biomodels} database \cite{biomodels}, where outflow reactions $\fX_1\ce{->} 0$ and $\fX_2\ce{->}0$ are disregarded: 
\begin{align*}
\fX_1+\fX_5 & \ce{->[\k_1]} \fX_3 \ce{->[\k_5]} \fX_5 & 
\fX_2+\fX_5 & \ce{->[\k_2]} \fX_4 \ce{->[\k_6]} \fX_5 & 
0 & \ce{->[\k_3]} \fX_1 & 0 & \ce{->[\k_4]} \fX_2 \, .
\end{align*}
We find $\mP_f = \{\{1,2,5\},\{3,4,6\}\}$, and choose $Z=(0, 0, 1, 1, 1)$, and $\wQ_{\mP_f} = ({-1}, {-1}, 0, 0, 1)^\top$, from where it is easily seen that the hypotheses of \Cref{thmnnn1}(iii) hold. Hence, 
the network has at most one positive nondegenerate equilibrium on each stoichiometric class for all parameter choices. As the network is nondegenerate, by \Cref{prop:nondeg}, it will hold generically for parameters in $\mZ_{sc}$ that the corresponding positive stoichiometric class has exactly one equilibrium, and this equilibrium is nondegenerate. 

The reader can verify that the original network with the two additional outflow reactions admits three positive nondegenerate equilibria on some stoichiometric classes.  
\end{ex}

Refinements of \Cref{thmnnn1} and \Cref{cornnn1} will be given in \Cref{sec:quadraticnnn1} for quadratic networks. 

\begin{rem}[Sharpness of the bounds in \Cref{thmnnn1}: open questions]
 \label[remark]{remnnn1}
 It is somewhat remarkable that the bounds in \Cref{thmnnn1} hold regardless of the reactant and product molecularities. In the case $n=2$, both bounds in \Cref{cornnn1} are achieved by the $(2,2,1)$ network
\[
\fX+2\fY \longrightarrow 3\fY\,, \qquad \fY \longrightarrow \fX\,,
\]
which has bounded stoichiometric classes and admits two positive, nondegenerate equilibria for some choices of rate constants (see \cite{banajipanteaMPNE} for the calculations). We remark, however, that {\em quadratic} $(2,2,1)$ networks do not admit multiple positive nondegenerate equilibria; this follows from \cite[Lemma 20]{BBH2024smallbif}, or by various easy calculations using the techniques in this paper. We leave open whether the bound on positive nondegenerate equilibria in \Cref{cornnn1}(i) is achieved for general $n$ and, if so, the minimum source molecularity of a network needed to achieve the bound. 
\end{rem}

\subsection{Networks of full rank with few sources}
When $r=n$, i.e., the network has full rank, then if it is not $\mP$-overdetermined, we must have $s_{\mP}=n$ (namely, the network is $\mP$-toric) with $\Ak$ described by a system of $m-n - p$ polynomial equations in as many variables. Moreover, $\mathrm{rank}\,[A\,|\,\bm{1}_\mP] = n+p$, hence $Q$ is empty, see \Cref{cor:special}. By \Cref{thm:finite}, we are guaranteed that $\Ak$ is generically finite. In the special case where $m=n+p$, part (iv) of \Cref{cor:special} applies (see \Cref{remfullPindependent}), $Y$ is a singleton, say $Y=\{\alpha_\k\}$, and for each $\k \in \R^m_+$ there is a unique, positive, nondegenerate equilibrium $(h(\alpha_\k)/\k)^{\wG}$, where $G = [A\,|\,\bm{1}_\mP]^{-1}$.

In the next result we show that full rank networks with sufficiently few sources admit no more than one positive nondegenerate equilibrium, regardless of source molecularity. \Cref{thmfewsources}, which expands on the claim in Theorem~12 in \cite{BBH2024smallbif}, treats the case of full rank networks with up to $n+p$ distinct sources. The consequences, which would be immediate from \Cref{prop:fewnomials}(ii) in the case of the trivial partition, are shown to hold for nontrivial partitions, provided repeated sources always occur in the same block of the partition. 
 
\begin{thm}
\label{thmfewsources}
Consider a dynamically nontrivial $(n,m,n)$ network which factors over a $p$-partition $\mP$. Assume that the network has no more than $n+p$ distinct sources and that each pair of reactions with the same source belongs to the same block of $\mP$. Then for each choice of rate constants the network has either (i) no positive equilibria; (ii) a unique positive nondegenerate equilibrium; or (iii) a continuum of positive degenerate  equilibria. 

Additionally, if the network is nondegenerate, then (ii) holds generically in the set of rate constants $\k$ for which $\mathcal{E}_\k\neq \emptyset$. 
\end{thm}
\vspace{-\baselineskip}\begin{proof}
 The result follows from \Cref{thmdegen} if the network is $\mP$-overdetermined, so we assume it is not. Then $s_\mP=n$, i.e., the network is $\mP$-toric, and consequently $\mathrm{rank}\,[A\,|\,\bm{1}_\mP] = n+p$, implying  
that there must be exactly $n+p$ distinct sources. If $m=n+p$, then $W$ is empty and, by \Cref{cor:special}, the network has exactly one positive, nondegenerate equilibrium for all choices of rate constants.

 So let us now assume that $m > n+p$. Then $W$ has $m-n-p$ rows and,  since repeated sources occur in the same block of $\mP$, each of these can be chosen to have one entry $1$, one entry $-1$, and the remaining entries equal to zero, corresponding to $m-n-p$ pairs of identical rows in $[A\,|\,\bm{1}_\mP]$ (c.f. \Cref{lem:Wstruct}). For any fixed $\k$, positive equilibria are in one-to-one correspondence with solutions $\alpha \in Y$ to the solvability system, $h(\alpha)^W = \kappa^W$ which, clearing denominators, is a system of $m-n-p$ {\em linear} equations in the same number of unknowns. We immediately get an alternative B\'ezout bound of $1$. These equations thus have, in $\R^{m-n-p}$, either no solution, infinitely many degenerate solutions, or exactly one nondegenerate solution. As $Y$ is open, the same conclusions hold over $Y$. By \Cref{thm:degensolvability}, degenerate (resp., nondegenerate) solutions of the solvability system correspond to degenerate (resp., nondegenerate) positive equilibria, and the first claim follows.

The last statement follows from \cite[Thm 3.4]{feliu:dimension}, as summarised in \Cref{sec:background}. 
\end{proof}

Choosing the trivial partition in \Cref{thmfewsources} we have the following immediate corollary. 
\begin{cor}
\label[corollary]{corfewsources}
If an $(n,m,n)$ network has no more than $n+1$ distinct sources, then for each choice of rate constants it admits either (i) no positive equilibria; (ii) a unique positive nondegenerate equilibrium; or (iii) a continuum of positive degenerate  equilibria. 
\end{cor}

\begin{ex}[Sources shared across different blocks of the partition]
\label[example]{exsrouceblock}
The conclusion of \Cref{thmfewsources} does not necessarily hold if reactions with the same source occur in different blocks of the partition. To see this, consider the $(2,5,2)$ network
\[
2\fY \ce{->[\k_1]} \fX+\fY \, , \quad 
\fX+\fY \ce{->[\k_2]} 2\fY  \, , \quad 
 2\fX+2\fY \ce{->[\k_3]} \fX +4\fY \, , \quad
 \fX+\fY \ce{->[\k_4]} 3\fY \, , \quad
 \fX+2\fY \ce{->[\k_5]} 2\fX  \,.
\] 
In this case, $\ker \Gamma$ is spanned by $\{(1,1,0,0,0),\,(0,0,1,1,2)\}$, and so $\mP_f=\{\{1,2\},\{3,4,5\}\}$. 
 On the other hand, the second and fourth reactions share a source. 
 The solvability equation takes the form $\alpha(1-\alpha) = \k_1\k_2^{-1}\k_3\k_4\k_5^{-2}$ for $\alpha \in (0,1)$. Consequently, the network admits two positive nondegenerate equilibria whenever $\k_1\k_2^{-1}\k_3\k_4\k_5^{-2} < 1/4$, and the parameter set for potential fold bifurcations is defined by $\k_1\k_2^{-1}\k_3\k_4\k_5^{-2} = 1/4$.
\end{ex}

\begin{rem}
\Cref{thmfewsources} shows that $(n, m, n)$ networks factoring over a $p$-partition, with exactly $n+p$ distinct sources, and no repeated sources appearing in different blocks of the partition, have similarities to $(n, n+1, n)$ networks, treated in \cite{banajiborosnonlinearity}. From \Cref{exsrouceblock} we see that the assumption that sources are not shared between reactions in different blocks is crucial.
\end{rem}

\Cref{prop:fewnomials} implies that an $(n,m,n)$ network with at most $n+2$ distinct sources can have at most $n+1$ positive nondegenerate equilibria. \Cref{thm:n_plus_2} below gives an extension to networks with at most $n+p+1$ distinct sources. We start by proving a more general and technical result, which builds on the proof of the Descartes' rule of signs for circuits in \cite{bihan:descartes:2}.

\begin{lem}\label[lemma]{lem:extend_circuits}
Consider a dynamically nontrivial $(n,m,n)$ network which factors over a $p$-partition $\mP$ and such that reactions with the same source belong to the same block of $\mP$. For each $i=1,\dots,p$, let $\ell_i$ be the number of distinct sources for the reactions in $P_i$ and assume that $\dim \mC_i\leq (m_i-\ell_i)+ 2$, where $m_i:=|P_i|$. 
Let $W$ be chosen as in \Cref{lem:Wstruct}, and suppose that the support of each row of $W$ intersects at most one block $P_i$ with $\dim \mC_i= (m_i-\ell_i) +2$. 
 
Then for each choice of rate constants the network admits at most 
$\prod_{i: \dim \mC_i =(m_i-\ell_i)+ 2} (\ell_i - 1)$ positive nondegenerate equilibria. 
\end{lem}
\vspace{-\baselineskip}\begin{proof}
If the network is (strongly) degenerate, then the result is trivially true. So we assume that the network is nondegenerate, hence not $\mP$-overdetermined. This implies $\rank\, [A\, |\, \bm{1}_\mP] = n+p$, hence $\rank W = m-n-p$, and  $\dim \mC=m-n$. By \Cref{lem:degensubnet}, for each $i$, the partition subnetwork $\mN_i$ is not strongly degenerate, 
hence not source deficient. If $\k \in \R^m_+$ is such that, for some $i$, either (i) the positive equilibrium set of $\mN_i$ is empty; or (ii) all positive equilibria of $\mN_i$ are strongly degenerate, then the same conclusions hold for $\mN$ (we use \Cref{lem:degensubnet} in the latter case), and the conclusions of the theorem are automatic at such $\k$. So henceforth we fix $\k\in \mathcal{K}$, where $\mathcal{K} \subseteq \R^m_+$ consists of all rate constants such that neither (i) nor (ii) hold for any $i$.

Let $\ell:=\sum_i\ell_i$ be the number of distinct sources of $\mathcal{N}$. Then $W$ has $m - \ell$ ``trivial'' rows in the sense that their nonzero entries are $1$ and $-1$ for a pair of identical sources. These give rise to linear equations in the solvability system. The remaining $\ell-n-p$ rows of $W$ do not include pairs of reactions with common sources in their support.  

Let us refer to the block $P_i$ as ``simple'' if $\dim \mC_i\leq (m_i-\ell_i)+ 1$. In that case, as $\mN_i$ is not strongly degenerate, 
\Cref{cor:linear_few} tells us that $\dim \mC_i = (m_i-\ell_i)+ 1$ and hence $\mathrm{dim}\,Y_i = m_i-\ell_i$. 
As $\mN_i$ is not source deficient, using that the rank of $\mN_i$ is $m_i-\dim \mC_i = \ell_i-1$, we conclude that the rank of the matrix  $W^{(i)}$ from \Cref{remsolvsubsystems} is at most $m_i-\ell_i$. As $W^{(i)}$ has at least $m_i-\ell_i$ trivial rows, we conclude that the solvability equations for $\mN_i$ in the variables $\alpha^{(i)}$ can be written as a square linear system. As $\k \in \mathcal{K}$, solutions to this system must consist of exactly one nondegenerate solution.

For blocks $i$ which are not simple, we have $\dim \mC_i = (m_i-\ell_i)+ 2$. A simple counting argument tells us that there must be $\ell-n-p$ of these blocks, which is equal to the number of nontrivial rows of $W$.
The support of each nontrivial row of $W$ must intersect some nonsimple block, for otherwise the subnetwork corresponding to the union of the simple blocks, with the corresponding natural partition, say $\mP'$, inherited from $\mP$, is $\mP'$-overdetermined. Moreover, by assumption, the support of each nontrivial row of $W$ intersects at most one nonsimple block $P_i$. By the pigeonhole principle, the support of each nontrivial row of $W$ intersects exactly one nonsimple block.
  
To conclude the proof, as $\k \in \mathcal{K}$, each subsystem of the solvability system corresponding to a simple block has exactly one nondegenerate solution; we substitute these solutions into the remaining solvability equations.  Having fixed the values of $\alpha_i$ corresponding to simple blocks, the solvability equations with support intersecting a nonsimple block $P_i$ now involve only the $m_i-\ell_i+1$ variables $\alpha^{(i)}$. These consist of $m_i-\ell_i$ independent linear equations corresponding to repeated sources and one further equation corresponding to a nontrivial row of $W$. We can use the linear equations to write all variables in this subsystem in terms of one variable, say $\hat{\alpha}_i$, from the variables $\alpha^{(i)}$. This subsystem is then a single equation of the form
  \begin{equation}
    \label{eqsolvpi}
    \beta \, \prod_j v_j(\hat{\alpha}_i)^{\omega_j} =1,
  \end{equation}
where $\beta>0$, each $v_j(\cdot)$ is affine, $j$ ranges over a subset of $P_i$ of size at most $\ell_i$, and $\omega_j$ are the corresponding subset of the entries in some row of $W$. Hence the left hand side of \eqref{eqsolvpi} includes a product of at most $\ell_i$ nonconstant factors. The arguments in the proof of \cite[Thm 2.4]{bihan:descartes:2} now imply that \eqref{eqsolvpi} has at most $\ell_i-1$ nondegenerate solutions. 

The maximum number of positive nondegenerate solutions of the solvability system is then bounded by the product of the number of possible nondegenerate solutions for each variable $\hat{\alpha}_{i}$ over all blocks for which $\dim \mC_i=(m_i- \ell_i)+ 2$, and the claim follows by \Cref{thm:degensolvability}. 
\end{proof}

\begin{rem}\label[remark]{rem:max}
For the assumptions of \Cref{lem:extend_circuits} to hold, the network must have at most $n+2p_f$ distinct sources. Indeed, if for each $i$ we have $\ell_i \leq m_i - \dim \mC_i + 2$ and the network has full rank, then $\sum_i \ell_i \leq m-(m-n)+2p \leq n+2p_f$.
\end{rem}

\begin{thm}\label{thm:n_plus_2}
Consider a dynamically nontrivial $(n,m,n)$ network which factors over a $p$-partition $\mP$. Assume that the network has no more than $n+p+1$ distinct sources and that reactions with the same source belong to the same block of $\mP$. Then for each choice of rate constants the network admits at most $n+p$ positive nondegenerate  equilibria. 

In particular, considering the trivial partition, any full rank network with $n+2$ distinct sources admits at most $n+1$ positive nondegenerate equilibria.
\end{thm}
\vspace{-\baselineskip}\begin{proof}
If the network is $\mP$-overdetermined, then we are done. If the network has at most $n+p$ distinct sources, then the statement follows from \Cref{thmfewsources}. So we are left to consider the case with $n+p+1$ distinct sources and $\rank\, [A\, |\, \bm{1}_\mP] = n+p$. 

As in \Cref{lem:extend_circuits}, let $\ell_i$ be the number of distinct sources in each block, so that $\ell := \sum_{i=1}^p \ell_i = n+p+1$. Setting $m_i := |P_i|$, we have
\begin{equation}\label{sum:Ci}
 \sum_{i=1}^p \big( \dim \mC_i - (m_i-\ell_i) \big) = (m-n)- m + \ell= p+1. 
\end{equation}
By \Cref{cor:linear_few}, if   $\dim \mC_i \leq m_i-\ell_i$ for some $i$, then the network does not have positive nondegenerate equilibria. Assuming this does not hold, 
\eqref{sum:Ci} implies that there is exactly one $i$ such that $\dim \mC_i - (m_i-\ell_i)=2$, with $\dim \mC_j - (m_j-\ell_j)=1$ for all $j \neq i$. The result now follows from \Cref{lem:extend_circuits}. 
\end{proof} 

\begin{rem}
In \Cref{sec:quadratic_full}, we will show that for a {\em quadratic} $(n,m,n)$ network with no more than $n+2$ distinct sources, the bound $n+1$ in \Cref{thm:n_plus_2} is also the worst-case scenario for the alternative Bézout source bound. Moreover, the combinatorial techniques we use to prove this claim automatically allow us to infer sharper B\'ezout source bounds for quadratic networks with various particular source combinations.
\end{rem}

The bound in \Cref{thm:n_plus_2} applies, in particular, to $(n,m,n)$ networks with at most $n+2$ distinct sources. The following result from \cite{bihan:descartes:2} tells us that for networks in this class the alternative B\'ezout source bound always beats the naive BKK source bound.

\begin{prop}\label[proposition]{remaltBezout}
    For an $(n,m,n)$ network with at most $n+2$ distinct sources, the alternative Bézout source bound is never larger than the naive BKK source bound. 
\end{prop}
\vspace{-\baselineskip}\begin{proof} See \cite[Remark 2.8]{bihan:descartes:2}. 
\end{proof}

The relation in \Cref{remaltBezout} does not necessarily hold for full rank networks with more than $n+2$ distinct sources: the next example, and \Cref{ex363} below, illustrate that for $(3,6,3)$ networks the alternative B\'ezout source bound may or may not improve on the naive BKK source bound.

\begin{ex}[The naive BKK source bound can beat the alternative B\'ezout source bound]
\label[example]{BKKbeatsBezout}
Consider a $(3,6,3)$ network with sources 
\[0\, , \quad 2\fX\, ,\quad \fY\, , \quad \fZ\, , \quad \fX+\fZ\,  ,\quad \fY+\fZ\, . \]
for which the naive BKK source bound is computed to be $5$. Choosing the trivial partition, the solvability matrix $W$ can be chosen to have rows $(1,0,-1,-1,0,1)$ and $(1,-1,0,-2,2,0)$, yielding an alternative B\'ezout source bound of $6$. It is easily checked that no choice of $W$ leads to a better alternative Bézout source bound.
\end{ex}

\subsection{Networks with $\mP$-independent sources}\label{sec:independent}
For a $p$-partition $\mP$, recall that $\mP$-independence of the sources of a network means that $s_{\mP}=m-p$; equivalently, $[A\,|\,\bm{1}_\mP]$ is a surjective matrix, hence, $W$ is empty. Consequently, the solvability of \eqref{eqbasiclog1} is automatic, \Cref{cor:special}(i) holds, and we have an explicit parameterisation of $\Ek$ with parameter space $Y \times \mathbb{R}^{n-m+p}_+$. Recall also that when we use the trivial partition, $\mP$-independence just means affine independence.

Let us now consider dynamically nontrivial networks with $\mP$-independent sources. Clearly such networks cannot be $\mP$-overdetermined.
\begin{enumerate}
\item The case $r=n$ has already been discussed above (see \Cref{remfullPindependent}), leading to a unique, positive, nondegenerate equilibrium for every choice of rate constants. 
\item If $r < n$, then the alternative system $Z[(h(\alpha)/\k)^{\wG}\circ \mu^{\wQ}] = K$ is a system of $n-r$ equations in $n-r$ unknowns $(\alpha,\mu)\in Y\times \R^{n-m+p}_+$, whose solutions are in smooth one-to-one correspondence with positive equilibria on the stoichiometric class defined by $Zx = K$. 
\end{enumerate}

 In the following three examples we demonstrate how we may use the alternative systems of equations defined in \Cref{thm:degensolvability} to study networks with $\mP$-independent sources: in particular, to describe positive equilibria on each stoichiometric class; characterise parameter regions for multistationarity; and study bifurcations.

\begin{ex}
 \label[example]{ex342b}
Consider any dynamically nontrivial $(3,4,2)$ network with sources
\[
\fX\, , \quad \fY\,  , \quad 2\fX \, ,\quad \fY+\fZ\,.
\]
The naive B\'ezout source bound in this case is $4$ and the naive BKK source bound is $3$. 
We next show that the latter is also the alternative B\'ezout source bound. In this case, $[A\,|\,\bm{1}]$ is a $4 \times 4$ matrix of rank $4$, and so $W$ and $Q$ are empty, and $G=[A\,|\,\bm{1}]^{-1}$. 
We have the parameterisation of $\Ek$,
\[
\left(\begin{array}{c}x\\y\\z\end{array}\right) = \left(\frac{h(\alpha)}{\k}\right)^{\wG} = \k^{-\wG}\circ\left(\begin{array}{c}  h_3(\alpha)h_1(\alpha)^{-1} \\ h_2(\alpha)\, h_3(\alpha) h_1(\alpha)^{-2}\\ h_4(\alpha)h_2(\alpha)^{-1}\end{array}\right)\,,
\]
where each $h_i$ is affine in the real variable $\alpha$. The equation $Z(h(\alpha)/\k)^{\wG} = K$ is equivalent to a single at-most-cubic equation in $\alpha$, and so the alternative B\'ezout source bound is $3$. 

This bound can be achieved and the particularly simple form of the alternative equations for equilibria yields other benefits. To see this, consider the bimolecular network:
\[
\fX \ce{->[\k_1]} \fY\,,\quad \fY \ce{->[\k_2]} \fZ,\quad 2\fX \ce{->[\k_3]} 2\fZ\,,\quad \fY+\fZ \ce{->[\k_4]} 2\fX\,,
\]
where we have $Z = (1,1,1)$ and $\ker \Gamma$ is spanned by $\{(2,1,0,1),(2,0,1,2)\}$. 
We set $h(\alpha) = (2, \alpha, 1-\alpha, 2-\alpha)^\top$ for $\alpha \in (0,1)$ and obtain the parameterisation of $\Ek$: 
\[
\left(\begin{array}{c}x\\y\\z\end{array}\right) = 
\left(\frac{h(\alpha)}{\k}\right)^{\wG} = 
\k^{-\wG} \circ \left(\begin{array}{c}
(1-\alpha)/2\\\alpha(1-\alpha)/4\\\alpha^{-1}(2-\alpha)\end{array}\right)\,.
\]
The equation $Z\left(h(\alpha)/\k\right)^{\wG} = K$ can be rearranged to give the equivalent polynomial equation
\begin{equation}
  \label{eqpoly2}
f(\alpha; \theta_1, \theta_2, \theta_3, K):=2\theta_1\alpha(1-\alpha)+\theta_2\alpha^{2}(1-\alpha) + 4\theta_3(2-\alpha) -4K\alpha = 0\,,
\end{equation}
where $\theta := \k^{-\wG}$. We can find (for example, using the \texttt{Mathematica} function {\tt FindInstance}), $\theta_1, \theta_2, \theta_3, K \in (0, \infty)$, such that $f(\alpha; \theta_1, \theta_2, \theta_3, K)$ has three distinct roots in $(0,1)$. Given that $\theta := \k^{-\wG}$, we can easily obtain the corresponding rate constants at which the network admits three positive nondegenerate equilibria on some stoichiometric class. But we can go much further, if desired: we can use Sturm's theorem to get a system of polynomial inequalities in the rate constants, which together characterise the parameter region at which three equilibria occur on any given stoichiometric class. The lengthy expressions are omitted here; but it is somewhat remarkable that we can obtain explicitly the semialgebraic set of parameters for which this network admits three positive equilibria on a given stoichiometric class.

Considering bifurcations, we can confirm also that the cubic above admits a cusp point in $(0,1)$, namely, there exist $\alpha \in (0,1)$ and $\theta_1, \theta_2, \theta_3, K \in (0, \infty)$ such that $f(\alpha; \theta_1, \theta_2, \theta_3, K)=f'(\alpha; \theta_1, \theta_2, \theta_3, K)=f''(\alpha; \theta_1, \theta_2, \theta_3, K)=0$ where $'$ denotes differentiation w.r.t. the first argument. It is easily seen that wherever this holds for \eqref{eqpoly2}, obtained after clearing denominators, it also holds for the original alternative equations. Consequently, for any fixed $K$, we can obtain a semialgebraic description of the parameter set for potential cusp bifurcations \cite{kuznetsov:2023} on the corresponding stoichiometric class. For example, we can easily eliminate $\alpha$ from the equations $f=f'=f''=0$ to obtain a set of two polynomial equations and some polynomial inequalities describing the potential cusp bifurcation set in $\theta$-space. We can then substitute $\theta = \k^{-\wG}$ to get the corresponding equations and inequalities in $\k$-space. We can also obtain an explicit parameterisation of the potential cusp bifurcation set in $\theta$-space using, for example, cylindrical algebraic decomposition. 
\end{ex}

\begin{ex}
\label[example]{ex453a}
Consider any dynamically nontrivial $(4,5,3)$ network with sources
\[
\mathsf{0}\, ,\quad 2\mathsf{W}\, ,\quad \fY+\fZ\,  ,\quad \fY\, ,\quad \fX+\fZ\,.
\]
The equilibrium equations on any stoichiometric class take the form of one linear equation and three at-most-quadratic equations in four variables which, {\em a priori} give a naive B\'ezout source bound of $8$ positive nondegenerate equilibria on any stoichiometric class. We can also compute the naive BKK source bound to be $6$.
The alternative equations do not give us improvements on the latter bound in full generality; but they greatly simplify the calculations for any particular network. We observe that $[A\,|\,\bm{1}]$ is a $5 \times 5$ matrix of full rank, so that $W$ and $Q$ are empty, and setting $G = [A\,|\,\bm{1}]^{-1}$
yields the parameterisation of equilibria
\[
\left(\begin{array}{c}w\\x\\y\\z\end{array}\right) = 
\left(\frac{h(\alpha)}{\k}\right)^{\wG} = \k^{-\wG}\circ\left(
\begin{array}{c}
h_2(\alpha)^{\nicefrac{1}{2}} h_1(\alpha)^{-\nicefrac{1}{2}}\\h_4(\alpha)h_5(\alpha)(h_1(\alpha) h_3(\alpha))^{-1}\\
h_4(\alpha)h_1(\alpha)^{-1} \\
h_3(\alpha)h_4(\alpha)^{-1}
\end{array}\right)
\]
where each $h_i$ is affine in the single variable $\alpha$. Indeed, $Z\left(h(\alpha)/\k\right)^{\wG} = K$ reduces to an equivalent polynomial in $\alpha$ of degree no more than $6$, giving an alternative B\'ezout source bound which is the same as the naive BKK source bound.
But more importantly, in any specific instance, we have reduced the problem of counting positive equilibria on some stoichiometric class to counting the roots in $(0,1)$ of a univariate polynomial. 

To illustrate, consider the following particular $(4,5,3)$ network with the above sources:
\[
\mathsf{0} \ce{->[\k_1]} \mathsf{W},\quad 2\mathsf{W} \ce{->[\k_2]} \fZ \,  ,\quad \fY+\fZ \ce{->[\k_3]} \fY \ce{->[\k_4]} \fX+\fZ \ce{->[\k_5]} \fY+\fZ\, .
\]
In this case, the naive BKK bound is $4$
but using the alternative equations we can quickly find that in fact the network has at most one positive nondegenerate equilibrium. For this network we have, $Z = (0,1,1,0)$, and $\ker \Gamma ={\rm span}\{(2,1,1,0,0),(0,0,1,1,1)\}$.
We set $h(\alpha) = (2\alpha, \alpha, 1, 1-\alpha, 1-\alpha)^\top$ for $\alpha \in (0,1)$, yielding the parameterisation of $\Ek$:
\[
\left(\begin{array}{c}w\\x\\y\\z\end{array}\right) = \left(\frac{h(\alpha)}{\k}\right)^{\wG} = \theta \circ \left(\begin{array}{c} 2^{-\nicefrac{1}{2}} \\(2\alpha)^{-1}(1-\alpha)^2\\(2\alpha)^{-1}(1-\alpha)\\(1-\alpha)^{-1}\end{array}\right)
\]
where $\theta = \k^{-\wG}$. The alternative system on the stoichiometric class defined by $x+y = K$ reduces to the equivalent quadratic
\[
\theta_2\alpha^2 - (2\theta_2+\theta_3+2K)\alpha + (\theta_2+\theta_3)=0\,\,,
\]
and we obtain the alternative B\'ezout bound of $2$. In fact, since $K > 0$, the above quadratic equation always has one solution greater than $1$, and so the network can have no more than one positive equilibrium on any stoichiometric class, and when it exists, this equilibrium is always nondegenerate. 
\end{ex}

In the final example of this section, we analyse a network with two conservation laws but where we can easily reduce the alternative system of equations to a single polynomial equation.

\begin{ex}
\label[example]{ex664}
Consider the following dynamically nontrivial bimolecular $(6,6,4)$ network 
\[
\fX_1 \ce{->[\k_1]} \fX_2 \ce{->[\k_2]} \fX_3 \ce{->[\k_3]} \fX_4, \quad \fX_3+\fX_5 \ce{->[\k_4]} \fX_1 + \fX_6, \quad \fX_4+\fX_5 \ce{->[\k_5]} \fX_2 + \fX_6, \quad \fX_6 \ce{->[\k_6]} \fX_5\,,
\]
corresponding to a phosphorylation mechanism  studied in \cite{Kothamachu2015multistability}.
In this case, the finest partition is the trivial partition, say $\mP$, and the network has $\mP$-independent (namely, affinely independent) sources. 
We may take $\{(1,1,0,1,0,1), (0,1,1,0,1,1)\}$ as a basis of $\ker\Gamma$, so that $h(\alpha) = (\alpha, 1, 1-\alpha, \alpha, 1-\alpha, 1)^\top$ for $\alpha \in (0,1)$, and set 
\[
\widehat{G} = \left(\begin{array}{ccrrcc}1&0&0&0&0&0\\0&1&0&0&0&0\\0&0&1&0&0&0\\0&0&1&-1&1&0\\0&0&-1&1&0&0\\0&0&0&0&0&1\end{array}\right), \quad Z = \left(\begin{array}{cccccc}0&0&0&0&1&1\\1&1&1&1&0&0\end{array}\right), \quad \wQ = (1,1,1,1,0,1)^\top\,.
\]
Setting $\t = \k^{-\wG}$, we get from \Cref{thminj} 
\[
\mathcal{E}_\k = \left\{\left(\t_1\alpha\mu, \t_2\mu, \t_3(1-\alpha)\mu, \frac{\t_4(1-\alpha)^2\mu}{\alpha}, \frac{\t_5\alpha}{1-\alpha}, \t_6\mu\right)\,:\,(\alpha, \mu) \in (0,1) \times \R_+\right\}\,.
\]
The alternative equations corresponding to the stoichiometric class defined by total amounts $(K_1, K_2)$ reduce to the equivalent polynomial system
\[
\begin{array}{rcl}
 \t_5\alpha + \t_6(1-\alpha)\mu & = & K_1(1-\alpha)\,,\\[5pt]
 \t_1\alpha^2\mu + \t_2\alpha\mu + \t_3\alpha(1-\alpha)\mu + \t_4(1-\alpha)^2\mu & = & K_2\alpha\,.
\end{array}
\]
where $K_1,K_2$ are positive constants. Eliminating $\mu$ gives a cubic in $\alpha$ whose solutions in $(0,1)$ are in one-to-one correspondence with positive equilibria on the corresponding stoichiometric class. The problem of finding the parameter sets for multistationarity and for potential fold and cusp bifurcations is thus reduced to examination of a single parameterised family of cubic equations. 
\end{ex}

\section{Quadratic networks}

In this section we present some further results on quadratic networks, namely, networks with source molecularity of no more than $2$. These are often considered more (bio)chemically plausible than networks with higher source molecularity, while the particularly simple structure of their exponent matrices allows stronger claims than in the general case. 

\subsection{Combinatorial preliminaries: kernel vectors as graph labellings}
Given an exponent matrix $A\in \Z_{\geq 0}^{m\times n}$, we define the {\bf SR graph} $\mG_A$ (adapted from \cite{banajicraciun}) to be the bipartite graph with $n$ {\bf S-vertices} corresponding to columns of $A$ (i.e., species), and $m$ {\bf R-vertices} corresponding to rows of $A$ (i.e., source complexes, not necessarily distinct). 
 We include $\ell$ parallel edges between the $i$th S-vertex and the $j$th R-vertex if $A_{ji} =\ell \geq 1$. In particular,  an isolated R-vertex  corresponds to a zero row in $A$, and an {\bf R-leaf}, namely an R-vertex of degree $1$, corresponds to a row of $A$ with only one nonzero entry, which is equal to $1$. 

\begin{ex}
\label[example]{exASR}
An example of an exponent matrix $A$ for a (quadratic) $3$-species, $5$-reaction network with sources $2\fX,\, \fX+\fY,\, \fY,\, \fY+\fZ$ and $\fZ$, and the corresponding SR graph $\mG_A$ are shown. 

\begin{center}
\begin{tikzpicture}[domain=0:12,scale=0.45]
\newcommand\Square[1]{+(-#1,-#1) rectangle +(#1,#1)}
\draw[thick] (2,0) -- (6,0);

\draw (2,0) .. controls (0.5,-0.5) and (0.5,-0.5) .. (-1,0);
\draw (2,0) .. controls (0.5,0.5) and (0.5,0.5) .. (-1,0);

\begin{scope}
 \clip(6,0) rectangle (7,2);
 \draw[thick] (8,0) circle (2cm);
\end{scope}

\begin{scope}
 \clip(5.8,-2.2) rectangle (10,0);
 \draw[thick] (8,0) circle (2cm);
\end{scope}

\path (0: 2cm) coordinate (S1);
\path (180: 1cm) coordinate (R1);
\coordinate (R4) at (4,0);
\coordinate (S4) at (6,0);

\begin{scope}[xshift=8cm]
 \path (120: 2cm) coordinate (R5);
 \path (240: 2cm) coordinate (R6);
 \path (300: 2cm) coordinate (S5);
 \path (0: 2cm) coordinate (R7);

\end{scope}

\fill (S1) circle (6pt);
\fill (S4) circle (6pt);
\fill (S5) circle (6pt);

\fill (R1) \Square{6pt} ;
\fill (R4) \Square{6pt} ;
\fill (R5) \Square{6pt} ;
\fill (R6) \Square{6pt} ;
\fill (R7) \Square{6pt} ;

\node [left] at (R1) {$r_1$};
\node [below] at (R4) {$r_2$};
\node [above left] at (R5) {$r_3$};
\node [below left] at (R6) {$r_4$};
\node [right] at (R7) {$r_5$};

\node [below] at (S1) {$s_1$};
\node [right] at (S4) {$s_2$};
\node [below right] at (S5) {$s_3$};

\node at (-8,0) {$\begin{blockarray}{cccc}
 & s_1 & s_2 & s_3 \\
\begin{block}{c(ccc)}
 r_1 & 2 & 0 & 0\\
 r_2 & 1 & 1 & 0\\
 r_3 & 0 & 1 & 0\\
 r_4 & 0 & 1 & 1\\
 r_5 & 0 & 0 & 1\\
\end{block}
\end{blockarray}$};
\end{tikzpicture}
\end{center}

 \vspace{-0.9cm}
 \end{ex}   

We denote by $|H|$ the number of R-vertices in a subgraph $H \subseteq \mG_A$. We refer to a cycle/path $H$ in $\mG_A$ as even (resp., odd) if $|H|$ is even (resp., odd). We include an isolated R-vertex in the definition of an odd path (but an isolated vertex is not a cycle). 
A path in $\mG_A$ is said to be an {\bf R-path} if its endpoints are R-vertices and it includes all edges incident to each of its R-vertices. In \Cref{exASR}, the path $r_5s_3r_4s_2r_3$ is an R-path of $\mG_A$. An isolated R-vertex is deemed an (odd) R-path.  A cycle in $\mG_A$ is said to be an {\bf R-cycle} if it includes all edges in $\mG_A$ incident to each of its R-vertices.

The following two observations about $\ker A$ and $\ker\, [A\,|\,\bm{1}]^\top$ are immediate, and allow us to pass from considering vectors in matrix kernels to considering labellings of $\mG_A$. In both cases, when summing labels on the neighbours of a vertex $v$, if a neighbour $v'$ is connected by $k$ parallel edges to $v$, we count the label of $v'$ $k$ times. 

\begin{itemize}
\item Each element of $\ker A$ corresponds uniquely to an $\R$-valued labelling of the S-vertices of $\mG_A$  such that the sum of labels on neighbours of any R-vertex is zero. We call such a labelling {\bf consistent}. 
\item Each element of $\ker\,[A\,|\,\bm{1}]^\top$ corresponds uniquely to an $\R$-valued labelling of the R-vertices of $\mG_A$ such that the sum of labels on neighbours of any S-vertex is zero and, additionally, all labels sum to $0$. We call such a labelling {\bf good}.
\end{itemize}

When   labels belong to a set $T\subseteq \R$, we refer to the labelling as a $\bm{T}$\textbf{-labelling}. 
A labelling is \textbf{supported} on a subgraph $H$ if the labels are nonzero only for the (R- or S-)vertices in $H$. 

 We say an exponent matrix $A\in \Z_{\geq 0}^{m\times n}$ is \textbf{quadratic} if its row-sums are all at most $2$. In particular, all its entries are $0$, $1$ or $2$. If $A$ is quadratic, then all cycles in $\mG_A$ are R-cycles. Moreover, there can be at most two parallel edges between an R-vertex and an S-vertex of $\mG_A$, and if this occurs, the R-vertex has only one distinct neighbour. A quadratic network is precisely one with a quadratic exponent matrix. 
We say that a connected component of $\mG_A$ is \textbf{quadratic} if all R-vertices in the component are of degree at most $2$.

\subsection{The solvability matrix of a quadratic network with trivial partition} 
Our next goal is to gain some understanding of the affine dependencies amongst the sources of a quadratic network, encoded in $\ker\,[A\,|\,\bm{1}]^\top$, with a view to getting bounds on the degree of the solvability equations for quadratic networks with few sources. Recall that good labellings of $\mG_A$ correspond to vectors in $\ker\,[A\,|\,\bm{1}]^\top$.

The next lemma addresses, for an arbitrary  exponent matrix $A$, the existence of good labellings that can be derived from simple inspection of the SR-graph $\mG_A$.  

\begin{lem}
 \label[lemma]{lemcomb}
 Consider an exponent matrix $A$ and the corresponding SR-graph $\mG_A$. 
\begin{enumerate}[label=(\roman*)]
\item An isolated odd cycle in $\mG_A$ admits no  good labelling other than the zero labelling.

\item If $\mG_A$ includes an even R-cycle or an even R-path $C$, then $\mG_A$ admits a good $\{-1,0,1\}$-labelling supported on $C$.

\item If $\mG_A$ has an acyclic quadratic connected component with no S-vertices of degree less than $2$, and at least one S-vertex of degree greater than $2$, then $\mG_A$ admits a nonzero good $\{-1,0,1\}$-labelling with support contained in the component. 

\item If $\mG_A$ has an quadratic connected component with one S-vertex of degree $4$ and remaining S-vertices of degree $2$, then $\mG_A$ admits a nonzero good $\{-1,0,1\}$-labelling with support contained in the component. 

\item If $\mG_A$ has an quadratic  connected component with two S-vertices of degree $3$ and remaining S-vertices of degree $2$, then $\mG_A$ admits a nonzero good $\{-2,-1,0,1,2\}$-labelling with support contained in  the component. 

\item If $\mG_A$ includes two quadratic connected components, each being either an odd R-path or having an S-vertex of degree $3$ while its remaining S-vertices have degree $2$, then $\mG_A$ admits a nonzero good $\{-2,-1,0,1,2\}$-labelling  with support contained in  the two components. 
\end{enumerate}
\end{lem}
\vspace{-\baselineskip}\begin{proof}
(i) and (ii) are immediate. For (iii), if $s$ is some S-vertex of degree $\ell>2$, then the component includes at least $\ell$ paths meeting only at $s$ with an R-vertex as one endpoint and $s$ as the other.
As at least two of them have the same parity, the component has an even R-path through $s$ and (ii) applies. 

For (iv), if (ii) and (iii) do not apply,  then the  component consists of an odd cycle $C_1$ sharing exactly one S-vertex $s$ with either an odd cycle or an odd R-path $C_2$. Examples of the desired labellings are shown in  \Cref{fig:1}(a). 

For (v), if (iii) does not apply, the component containing the two degree three S-vertices $s_1,s_2$ must take one of the forms
 \begin{center}
 \vspace{-0.2cm}
\begin{tikzpicture}[domain=0:12,scale=0.26, even odd rule]
\newcommand\Square[1]{+(-#1,-#1) rectangle +(#1,#1)}

\draw[color=gray!50] (-4,-3)--(-4,3)--(4,3)--(4,-3)--cycle;

\draw[thick, dashed] (0,0) circle (2cm);
\begin{scope}
 \clip(-2.2,-1) rectangle (2.2,1);
 \draw[thick] (0,0) circle (2cm);
\end{scope}

\draw[thick] (-2,0) -- (-1,0);
\draw[thick] (1,0) -- (2,0);
\draw[thick, dashed] (-1,0) -- (1,0);
\begin{scope}
 \clip(-1,-2.2) rectangle (1,2.2);
 
\end{scope}

\path (0: 2cm) coordinate (S1);
\path (180: 2cm) coordinate (S2);

\fill (S1) circle (6pt);
\fill (S2) circle (6pt);

\node [below right] at (S1) {$\scriptstyle{s_2}$};
\node [below left] at (S2) {$\scriptstyle{s_1}$};

\begin{scope}[xshift=12cm]
 \draw[color=gray!50] (-6,-3)--(-6,3)--(6,3)--(6,-3)--cycle;

\draw[thick, dashed] (0,0) circle (2cm);
\begin{scope}
 \clip(-2.2,-1) rectangle (2.2,1);
 \draw[thick] (0,0) circle (2cm);
\end{scope}

\draw[thick] (-2.5,0) -- (-2,0);
\draw[thick, dashed] (-4,0) -- (-2.5,0);
\draw[thick, dashed] (-4.5,0) -- (-4,0);
\draw[thick] (2,0) -- (2.5,0);
\draw[thick, dashed] (2.5,0) -- (4,0);
\draw[thick] (4,0) -- (4.5,0);

\path (0: 2cm) coordinate (S1);
\path (180: 2cm) coordinate (S2);
\coordinate (R1) at (-4.5,0);
\coordinate (R2) at (4.5,0);

\fill (S1) circle (6pt);
\fill (S2) circle (6pt);
\fill (R1) \Square{6pt} ;
\fill (R2) \Square{6pt} ;

\node [below right] at (S1) {$\scriptstyle{s_2}$};
\node [below left] at (S2) {$\scriptstyle{s_1}$};

\end{scope}

\begin{scope}[xshift=24cm]
\draw[color=gray!50] (-3.5,-3)--(-3.5,3)--(10.5,3)--(10.5,-3)--cycle;

\draw[thick, dashed] (0,0) circle (2cm);

\draw[thick, dashed] (2,0) -- (5,0);

\begin{scope}
 \clip(5,0) rectangle (6,2);
 \draw[thick, dashed] (7,0) circle (2cm);
\end{scope}

\begin{scope}
 \clip(4.8,-2.2) rectangle (9,0);
 \draw[thick, dashed] (7,0) circle (2cm);
\end{scope}

\path (0: 2cm) coordinate (S1);
\coordinate (S2) at (5,0);

\begin{scope}[xshift=7cm]
 \path (120: 2cm) coordinate (R5);
 \path (0: 2cm) coordinate (R7);
\end{scope}

\fill (S1) circle (6pt);
\fill (S2) circle (6pt);

\fill (R5) \Square{6pt} ;
\fill (R7) \Square{6pt} ;

\node [left] at (S1) {$\scriptstyle{s_1}$};
\node [right] at (S2) {$\scriptstyle{s_2}$};

\end{scope}

\begin{scope}[xshift=40cm]
\draw[color=gray!50] (-3.5,-3)--(-3.5,3)--(10.5,3)--(10.5,-3)--cycle;

\draw[thick, dashed] (0,0) circle (2cm);
 \draw[thick, dashed] (7,0) circle (2cm);
\draw[thick, dashed] (2,0) -- (5,0);

\path (0: 2cm) coordinate (S1);
\coordinate (S2) at (5,0);

\fill (S1) circle (6pt);
\fill (S2) circle (6pt);

\node [left] at (S1) {$\scriptstyle{s_1}$};
\node [right] at (S2) {$\scriptstyle{s_2}$};

\end{scope}

\end{tikzpicture}
\vspace{-0.2cm}
\end{center}
In the first case the component {\em must} include an even R-cycle and (ii) applies. In the second, it {\em must} include an even R-path and (ii) again applies. In the latter two cases, if the component includes an even R-cycle or an even R-path, then (ii) applies. Otherwise it consists of a path $P$ having S-vertices $s_1,s_2$ as endpoints, while attached to $s_1,s_2$ are either: an odd R-path and an odd R-cycle; or two odd R-cycles. Both possible components admit good labellings as illustrated by example in \Cref{fig:1}(b).

Finally, for part (vi), if (ii), hence (iii), does not apply, any component having all S-vertices of degree $2$ except one, say $s$, of degree $3$, must consist of an odd cycle $C$ with a path $P$ attached at $s$ and having an R-vertex as endpoint. 
 We may label the R-vertices in $C$ with $\pm 1$ and those in $P$ with $\pm 2$ 
to obtain a labelling such that label-sums around each S-vertex are zero, and the total label-sum is $\pm 1$. On the other hand, we can label R-vertices of any odd R-path with alternating values $1$ and $-1$, giving a label-sum of $\pm 1$. See \Cref{fig:1}(c). Choosing labels so that one component has label-sum $1$ and the other $-1$, we obtain the desired good labelling. 
\end{proof}

\begin{figure}[!t]
    \centering

 \begin{minipage}[h]{0.28\textwidth}
 \begin{center}
\begin{tikzpicture}[domain=0:12,scale=0.3]
\newcommand\Square[1]{+(-#1,-#1) rectangle +(#1,#1)}

\draw[color=gray!50] (-4,-3)--(-4,3)--(8,3)--(8,-3)--cycle;

\draw[thick] (0,0) circle (2cm);

\begin{scope}
 \clip(2,0) rectangle (3,2);
 \draw[thick] (4,0) circle (2cm);
\end{scope}

\begin{scope}
 \clip(1.8,-2.2) rectangle (6,0);
 \draw[thick] (4,0) circle (2cm);
\end{scope}

\path (0: 2cm) coordinate (S1);
\path (60: 2cm) coordinate (R1);
\path (120: 2cm) coordinate (S2);
\path (180: 2cm) coordinate (R2);
\path (240: 2cm) coordinate (S3);
\path (300: 2cm) coordinate (R3);

\begin{scope}[xshift=4cm]
 \path (120: 2cm) coordinate (R5);
 \path (240: 2cm) coordinate (R6);
 \path (300: 2cm) coordinate (S5);
 \path (0: 2cm) coordinate (R7);

\end{scope}

\fill (S1) circle (6pt);
\fill (S2) circle (6pt);
\fill (S3) circle (6pt);
\fill (S5) circle (6pt);

\fill (R1) \Square{6pt} ;
\fill (R2) \Square{6pt} ;
\fill (R3) \Square{6pt} ;
\fill (R5) \Square{6pt} ;
\fill (R6) \Square{6pt} ;
\fill (R7) \Square{6pt} ;

\node [above] at (R1) {$\scriptstyle{1}$};
\node [left] at (R2) {$\scriptstyle{-1}$};
\node [below] at (R3) {$\scriptstyle{1}$};
\node [above] at (R5) {$\scriptstyle{-1}$};
\node [below] at (R6) {$\scriptstyle{-1}$};
\node [right] at (R7) {$\scriptstyle{1}$};

\begin{scope}[yshift=-7cm]
 \draw[color=gray!50] (-4,-3)--(-4,3)--(8,3)--(8,-3)--cycle;

\draw[thick] (0,0) circle (2cm);
\draw[thick] (4,0) circle (2cm);

\path (0: 2cm) coordinate (S1);
\path (60: 2cm) coordinate (R1);
\path (120: 2cm) coordinate (S2);
\path (180: 2cm) coordinate (R2);
\path (240: 2cm) coordinate (S3);
\path (300: 2cm) coordinate (R3);

\begin{scope}[xshift=4cm]
 \path (60: 2cm) coordinate (S4);
 \path (120: 2cm) coordinate (R5);
 \path (240: 2cm) coordinate (R6);
 \path (300: 2cm) coordinate (S5);
 \path (0: 2cm) coordinate (R7);

\end{scope}

\fill (S1) circle (6pt);
\fill (S2) circle (6pt);
\fill (S3) circle (6pt);
\fill (S4) circle (6pt);
\fill (S5) circle (6pt);

\fill (R1) \Square{6pt} ;
\fill (R2) \Square{6pt} ;
\fill (R3) \Square{6pt} ;
\fill (R5) \Square{6pt} ;
\fill (R6) \Square{6pt} ;
\fill (R7) \Square{6pt} ;

\node [above] at (R1) {$\scriptstyle{1}$};
\node [left] at (R2) {$\scriptstyle{-1}$};
\node [below] at (R3) {$\scriptstyle{1}$};
\node [above] at (R5) {$\scriptstyle{-1}$};
\node [below] at (R6) {$\scriptstyle{-1}$};
\node [right] at (R7) {$\scriptstyle{1}$};

\end{scope}

\end{tikzpicture}

(a) \Cref{lemcomb}(iv)
\end{center}
\end{minipage} \quad
\begin{minipage}[h]{0.4\textwidth}
\begin{center}
\begin{tikzpicture}[domain=0:12,scale=0.3]
\newcommand\Square[1]{+(-#1,-#1) rectangle +(#1,#1)}

\draw[color=gray!50] (-5.5,-3)--(-5.5,3)--(12.5,3)--(12.5,-3)--cycle;

\draw[thick] (0,0) circle (2cm);
\draw[thick] (2,0) -- (5,0);

\begin{scope}
 \clip(5,0) rectangle (6,2);
 \draw[thick] (7,0) circle (2cm);
\end{scope}

\begin{scope}
 \clip(4.8,-2.2) rectangle (9,0);
 \draw[thick] (7,0) circle (2cm);
\end{scope}

\path (0: 2cm) coordinate (S1);
\path (60: 2cm) coordinate (R1);
\path (120: 2cm) coordinate (S2);
\path (180: 2cm) coordinate (R2);
\path (240: 2cm) coordinate (S3);
\path (300: 2cm) coordinate (R3);
\coordinate (R4) at (3.5,0);
\coordinate (S4) at (5,0);

\begin{scope}[xshift=7cm]
 \path (120: 2cm) coordinate (R5);
 \path (240: 2cm) coordinate (R6);
 \path (300: 2cm) coordinate (S5);
 \path (0: 2cm) coordinate (R7);
\end{scope}

\fill (S1) circle (6pt);
\fill (S2) circle (6pt);
\fill (S3) circle (6pt);
\fill (S4) circle (6pt);
\fill (S5) circle (6pt);

\fill (R1) \Square{6pt} ;
\fill (R2) \Square{6pt} ;
\fill (R3) \Square{6pt} ;
\fill (R4) \Square{6pt} ;
\fill (R5) \Square{6pt} ;
\fill (R6) \Square{6pt} ;
\fill (R7) \Square{6pt} ;

\node [above right] at (R1) {$\scriptstyle{1}$};
\node [left] at (R2) {$\scriptstyle{-1}$};
\node [below right] at (R3) {$\scriptstyle{1}$};
\node [below] at (R4) {$\scriptstyle{-2}$};
\node [above left] at (R5) {$\scriptstyle{1}$};
\node [below left] at (R6) {$\scriptstyle{1}$};
\node [right] at (R7) {$\scriptstyle{-1}$};

\begin{scope}[xshift=-1.5cm,yshift=-7cm]

 \draw[color=gray!50] (-4,-3)--(-4,3)--(14,3)--(14,-3)--cycle;

\draw[thick] (0,0) circle (2cm);
\draw[thick] (2,0) -- (8,0);
\draw[thick] (10,0) circle (2cm);

\path (0: 2cm) coordinate (S1);
\path (60: 2cm) coordinate (R1);
\path (120: 2cm) coordinate (S2);
\path (180: 2cm) coordinate (R2);
\path (240: 2cm) coordinate (S3);
\path (300: 2cm) coordinate (R3);
\coordinate (R4) at (3.5,0);
\coordinate (S4) at (5,0);
\coordinate (R8) at (6.5,0);
\coordinate (S8) at (8,0);

\begin{scope}[xshift=10cm]
 \path (120: 2cm) coordinate (R5);
 \path (240: 2cm) coordinate (R6);
 \path (300: 2cm) coordinate (S5);
 \path (0: 2cm) coordinate (R7);
 \path (60: 2cm) coordinate (S6);

\end{scope}

\fill (S1) circle (6pt);
\fill (S2) circle (6pt);
\fill (S3) circle (6pt);
\fill (S4) circle (6pt);
\fill (S5) circle (6pt);
\fill (S6) circle (6pt);
\fill (S8) circle (6pt);

\fill (R1) \Square{6pt} ;
\fill (R2) \Square{6pt} ;
\fill (R3) \Square{6pt} ;
\fill (R4) \Square{6pt} ;
\fill (R5) \Square{6pt} ;
\fill (R6) \Square{6pt} ;
\fill (R7) \Square{6pt} ;
\fill (R8) \Square{6pt} ;

\node [above right] at (R1) {$\scriptstyle{1}$};
\node [left] at (R2) {$\scriptstyle{-1}$};
\node [below right] at (R3) {$\scriptstyle{1}$};
\node [below] at (R4) {$\scriptstyle{-2}$};
\node [below] at (R8) {$\scriptstyle{2}$};
\node [above left] at (R5) {$\scriptstyle{-1}$};
\node [below left] at (R6) {$\scriptstyle{-1}$};
\node [right] at (R7) {$\scriptstyle{1}$};

\end{scope}

\end{tikzpicture}

(b) \Cref{lemcomb}(v)

\end{center}
\end{minipage}\quad
   \begin{minipage}[h]{0.22\textwidth}
   \begin{center}
\begin{tikzpicture}[domain=0:12,scale=0.3]
\newcommand\Square[1]{+(-#1,-#1) rectangle +(#1,#1)}

\draw[color=gray!50] (-4,-3)--(-4,3)--(5,3)--(5,-3)--cycle;

\draw[thick] (0,0) circle (2cm);
\draw[thick] (2,0) -- (3.5,0);

\path (0: 2cm) coordinate (S1);
\path (60: 2cm) coordinate (R1);
\path (120: 2cm) coordinate (S2);
\path (180: 2cm) coordinate (R2);
\path (240: 2cm) coordinate (S3);
\path (300: 2cm) coordinate (R3);
\coordinate (R4) at (3.5,0);

\begin{scope}[yshift=-7cm]
\draw[color=gray!50] (-4,-3)--(-4,3)--(5,3)--(5,-3)--cycle;

\coordinate (S4) at (-1,0);
\coordinate (R7) at (-2.5,0);
\coordinate (R8) at (0.5,0);
\coordinate (S7) at (2,0);
\coordinate (R9) at (3.5,0);

\draw[thick] (-2.5,0) -- (3.5,0);

\end{scope}

\fill (S1) circle (6pt);
\fill (S2) circle (6pt);
\fill (S3) circle (6pt);
\fill (S4) circle (6pt);
\fill (S7) circle (6pt);

\fill (R1) \Square{6pt} ;
\fill (R2) \Square{6pt} ;
\fill (R3) \Square{6pt} ;
\fill (R4) \Square{6pt} ;
\fill (R7) \Square{6pt} ;
\fill (R8) \Square{6pt} ;
\fill (R9) \Square{6pt} ;

\node [above right] at (R1) {$\scriptstyle{1}$};
\node [left] at (R2) {$\scriptstyle{-1}$};
\node [below right] at (R3) {$\scriptstyle{1}$};
\node [below] at (R4) {$\scriptstyle{-2}$};

\node [below] at (R7) {$\scriptstyle{1}$};
\node [below] at (R8) {$\scriptstyle{-1}$};
\node [below] at (R9) {$\scriptstyle{1}$};

\end{tikzpicture}

\smallskip
(c) \Cref{lemcomb}(vi)
\end{center}
\end{minipage}

    \caption{Good labellings in \Cref{lemcomb}(iv)-(vi).  }
    \label{fig:1}
\end{figure}

If the SR-graph $\mG_A$ of the exponent matrix of a  network satisfies any of the conditions in \Cref{lemcomb}, then we obtain information on choices of rows of $W$.  We use this information to make a strong claim about quadratic exponent matrices next.
 
 Given $w \in \Z^n$ such that $\sum_iw_i=0$, define the {\bf degree} of $w$ as $\mathrm{deg}\,w :=\sum_{i=1}^n|w_i|/2$. 
 Equivalently, $\mathrm{deg}\,w$ is the sum of the positive entries of $w$.

\begin{thm}
 \label{thmquadker}
Given a quadratic exponent matrix $A\in \Z^{(n+2)\times n}_{\geq 0}$, there exists a nonzero (row) vector $w \in \ker\, A\,|\,\bm{1}]^\top$ satisfying $\mathrm{deg}\,w\leq n+1$ and $|w_i|\leq 2$ for all $i=1,\dots,n+2$. 
 \end{thm}
 \vspace{-\baselineskip}
\begin{proof}
It is enough to show that there exists a nonzero good $\{-2,-1,0,1,2\}$-labelling of $\mG_A$. Indeed, if this is the case, 
the corresponding nonzero vector $w\in \ker\,[A\,|\,\bm{1}]^\top$
satisfies $|w_i|\leq 2$ for all $i=1,\dots,n+2$ and hence $\mathrm{deg}\,w  \leq n+2$. If we have equality in the latter, then all entries of $w$ must be $\pm 2$ (and $n$ even). Dividing by $2$ we obtain a new vector of degree at most $n+1$ as desired.   
 
 Observe that it is enough to show that a union of connected components of  $\mathcal{G}_A$ admits a nonzero good $\{-2,-1,0,1,2\}$-labelling. As $\mG_A$ has two more R-vertices than S-vertices, it cannot consist of a disjoint union of cycles. We will use that $A$ being quadratic implies that $\sum_{i,j}A_{ij} \leq 2n+4$ and that the degree of any R-vertex in $\mG_A$ is $0, 1$ or $2$.  

We now prove the claim by induction. It is easily checked for $n=1$. Fix $n \geq 2$ and let us assume the claim holds for $(\ell+2)\times \ell$ matrices with $\ell\leq n-1$. We consider first some scenarios where the claim either follows from the inductive hypothesis or is easily confirmed by direct calculation:

(1) If $\mG_A$ has an isolated S-vertex $s$ (corresponding to a zero column of $A$), let $r$ be an arbitrary R-vertex. 
If $\mG_A$ has an S-vertex $s$ that has a unique neighbour, let $r$ be this neighbour. In both cases, label $r$ with $0$ and the claim follows from the inductive hypothesis applied to the induced subgraph of $\mG_A$ obtained by removing $s$ and $r$.

(2) If $\sum_{i,j}A_{ij} > 2n+2$, then,  w.l.o.g., the top $n+1$ rows of $A$ form a matrix, say $B$, whose row-sums are all $2$. By the inductive hypothesis, there exists a row vector $\widetilde{w} \in \ker\,[\widetilde{A}\,|\,\bm{1}]^\top$ of the desired form,  where $\widetilde{A}$ is obtained by deleting the first column and final row of $A$. As the row-sums of $B$ are $2$,  the first column of $[B\,|\,\bm{1}]$ is a linear combination of its other columns, hence orthogonal to $\widetilde{w}$. We obtain $(\widetilde{w},0)[A\,|\,\bm{1}] = 0$ and $w=(\widetilde{w},0)$ is as desired.

(3) If $\mG_A$ includes an isolated odd cycle $C$, say with $\ell\geq 1$ R-vertices, then by \Cref{lemcomb}(i),
any good labelling of $\mG_A$ labels the R-vertices of $C$ with zeros. We may remove $C$ to obtain a smaller nonempty graph with $n-\ell$ S-vertices and $n-\ell+2$ R-vertices, and apply the inductive hypothesis.

 If neither (1) nor (2) apply, hence $2n \leq \sum_{i,j}A_{ij} \leq 2n+2$, there are three possible scenarios:
\begin{enumerate}[label=(\alph*)]
\item $\sum_{i,j}A_{ij} = 2n$: all S-vertices have degree $2$. 
\item $\sum_{i,j}A_{ij} = 2n+1$: one S-vertex has degree $3$ and the rest have degree $2$. 
\item $\sum_{i,j}A_{ij} = 2n+2$:  either one S-vertex has degree $4$  or two have degree $3$, and the rest have degree $2$. 
\end{enumerate}
We treat each case separately, assuming that additionally (3) does not apply.

(a) As all S-vertices have degree $2$, each connected component of $\mathcal{G}_A$ is either a cycle (hence an R-cycle) or an R-path.  If one of these is even, then \Cref{lemcomb}(ii) gives a good $\{-1,0,1\}$-labelling supported on the component. Otherwise, as (3) does not apply, $\mathcal{G}_A$ consists of disjoint odd R-paths. As $\mathcal{G}_A$ has two more R-vertices than S-vertices, it must consist of exactly two odd R-paths, 
and \Cref{lemcomb}(vi) gives the desired nonzero good $\{-2,-1,0,1,2\}$-labelling. 
 
(b) By a counting argument as in (a), $\mG_A$ has   either three R-leaves, or one R-leaf and an isolated R-vertex.
It follows that either part (iii) or (vi) of \Cref{lemcomb} apply, 
and there exists a nonzero good $\{-2,-1,0,1,2\}$-labelling. Indeed, if the vertex of degree $3$ does not belong to an acyclic component, then necessarily this component has exactly one  R-leaf. In that case there must be an additional R-path and \Cref{lemcomb}(vi) applies.

 (c) $\mG_A$ must now include either two R-leaves or an isolated R-vertex. Parts (iv) and (v) of \Cref{lemcomb}  give respectively the scenarios where $\mG_A$ has   an S-vertex of degree $4$ or two S-vertices of degree $3$ in the same connected component. If two S-vertices of degree $3$ are in different connected components, then the claim follows from part (vi) of \Cref{lemcomb}.

As we have exhausted all possibilities, this completes the proof.
\end{proof}

\begin{rem}[Observations related to \Cref{thmquadker}]
 \label[remark]{remquadker}
We gather together some observations about quadratic networks arising directly from \Cref{lemcomb} and \Cref{thmquadker}, or their proof, or the preliminaries. Recall that for a quadratic network, every cycle is an R-cycle.

\begin{enumerate}
\item There are only two topologies of $\mG_A$ leading to $\mathrm{deg}\,w=n+1$ (namely, maximal degree): $n$ entries of $w$ must be $\pm 2$ and the remaining two $1$ and $-1$. This can only arise from the scenario in \Cref{lemcomb}(vi) in which case $\mG_A$ has two components: either both consist of a 2-vertex cycle with an attached path; or one is a 2-vertex cycle with an attached path, and the other is an isolated R-vertex as shown:
\begin{center}
\begin{tikzpicture}[domain=0:12,scale=0.4]
\newcommand\Square[1]{+(-#1,-#1) rectangle +(#1,#1)}
\draw[color=gray!50] (-2,-2.5)--(-2,1)--(9,1)--(9,-2.5)--cycle;

\draw (2,0) .. controls (0.5,-0.5) and (0.5,-0.5) .. (-1,0);
\draw (2,0) .. controls (0.5,0.5) and (0.5,0.5) .. (-1,0);

\path (0: 2cm) coordinate (S1);
\path (180: 1cm) coordinate (R1);
\coordinate (R2) at (8,0);
\draw[thick, dashed] (2,0) -- (8,0);

\fill (S1) circle (6pt);
\fill (R1) \Square{6pt} ;
\fill (R2) \Square{6pt} ;

\begin{scope}[yshift=-1.5cm]
\coordinate (R1) at (4,0);\fill (R1) \Square{6pt} ;
\end{scope}

\begin{scope}[xshift = -14cm]

 \draw[color=gray!50] (-2,-2.5)--(-2,1)--(9,1)--(9,-2.5)--cycle;

\draw (2,0) .. controls (0.5,-0.5) and (0.5,-0.5) .. (-1,0);
\draw (2,0) .. controls (0.5,0.5) and (0.5,0.5) .. (-1,0);

\path (0: 2cm) coordinate (S1);
\path (180: 1cm) coordinate (R1);
\coordinate (R2) at (8,0);
\draw[thick, dashed] (2,0) -- (8,0);

\fill (S1) circle (6pt);
\fill (R1) \Square{6pt} ;
\fill (R2) \Square{6pt} 
;
\begin{scope}[yshift=-1.5cm]
\draw (2,0) .. controls (0.5,-0.5) and (0.5,-0.5) .. (-1,0);
\draw (2,0) .. controls (0.5,0.5) and (0.5,0.5) .. (-1,0);

\path (0: 2cm) coordinate (S1);
\path (180: 1cm) coordinate (R1);
\coordinate (R2) at (8,0);
\draw[thick, dashed] (2,0) -- (8,0);

\fill (S1) circle (6pt);
\fill (R1) \Square{6pt} ;
\fill (R2) \Square{6pt} ;

\end{scope}

\end{scope}

\end{tikzpicture}
\end{center}

\item If some species appears in the source of no more than one reaction, then we can choose $w$ such that $\mathrm{deg}\,w \leq n$.
\item If $\sum_{i,j}A_{ij} > 2n+2$, 
then we can choose $w$ such that $\mathrm{deg}\,w \leq n$.
\item If $\mG_A$ includes isolated odd cycles $C_1, \ldots, C_s$, then we can choose $w$ with  $\mathrm{deg}\,w \leq n+1-\sum_{i=1}^s|C_i|$.
\item If $\mG_A$ includes an even cycle $C$, then as $C$ contains at most $n$ R-vertices, we can choose $w$ such that $\mathrm{deg}\,w = |C|/2 \leq n/2$. 
\item If $\mG_A$ includes an even R-path $P$, then we can choose $w$ such that $\mathrm{deg}\,w \leq |P|/2 \leq (n+1)/2$.
\end{enumerate}

The reader may infer from the proof of \Cref{thmquadker} many more such constraints on the degree of $w$ from the topology of $\mG_A$. Observe that part (6) of \Cref{remquadker} applies to arbitrary exponent matrices, not necessarily quadratic, and the same holds for part (5) of \Cref{remquadker} provided $C$ is an R-cycle. This follows in both cases from \Cref{lemcomb}(ii).
\end{rem}

\subsection{Quadratic full rank networks with few sources}\label{sec:quadratic_full}
The main result of this subsection, \Cref{thmquadnn2n}, concerns full rank quadratic networks with $n$ species and $n+2$ distinct sources (the case of $n+1$ or fewer distinct sources is already covered by \Cref{thmfewsources}). The number of positive nondegenerate equilibria of such a network is at most $n+1$ by \Cref{thmfewsources}; however, here we show, using \Cref{thmquadker}, that this is also the alternative B\'ezout source bound. 

\begin{thm}
\label{thmquadnn2n}
Let $k \geq 2$ and consider a dynamically nontrivial, quadratic $(n,n+k,n)$ network with $n+2$ distinct sources, and which is not source deficient. Such a network has an alternative B\'ezout source bound of at most $n+1$, and consequently admits no more than $n+1$ positive nondegenerate equilibria.
\end{thm}
\vspace{-\baselineskip}
\begin{proof}
The matrix $W$ corresponding to the trivial partition and chosen as in \Cref{lem:Wstruct} has size $(k-1) \times (n+k)$, and $k-2$ linearly independent ``trivial'' rows with only two nonzero entries equal to $1$ and $-1$, corresponding to linear solvability equations. The remaining ``nontrivial'' row of $W$, say $w$, has support, say $P$, which does not include indices of two reactions with the same source. We assume without loss of generality that $P$ does not include the indices $i_1,\dots,i_{k-2}$.
The claim now follows if $w$ can be chosen to have degree no more than $n+1$. This however, is immediate from \Cref{thmquadker}: we remove from $A$ the rows with indices $i_1,\dots,i_{k-2}$
to obtain a new matrix $\tilde{A} \in \Z^{(n+2) \times n}_{\geq 0}$; obtain, via \Cref{thmquadker}, a nonzero row vector $\tilde{w} \in \mathbb{Z}^{n+2}$ of degree no more than $n+1$ such that $\tilde{w}[\tilde{A}\,|\,\bm{1}]=0$; and finally obtain $w \in \mathbb{Z}^{n+k}$ of the same degree as $\tilde{w}$, by augmenting $\tilde{w}$ with zeros at entries $i_1,\dots,i_{k-2}$. 
\end{proof}

\begin{rem}
For $n=1$, the bound in \Cref{thmquadnn2n} is not optimal: by \cite[Lemma 20]{BBH2024smallbif}, quadratic rank-one networks do not admit multiple positive nondegenerate equilibria, regardless of the number of species (see also \cite{pantea:voitiuk:2022}). 
\end{rem}

\begin{rem}[Bounds based on combinatorial properties of networks]
Recall that the bound on positive nondegenerate equilibria in \Cref{thmquadnn2n} is also true for networks which are not necessarily quadratic (c.f. \Cref{prop:fewnomials} and \Cref{thm:n_plus_2}). From the constructive proof of \Cref{thmquadker}, via combinatorial considerations alone, we get bounds on the number of positive nondegenerate equilibria for any full rank quadratic network with $n$ species and at most $n+2$ distinct sources (see \Cref{remquadker}).  These bounds are often considerably lower than the worst case in \Cref{thmquadnn2n}.  For example, \Cref{remquadker}(1) tells us that the bound is not attained unless $A$ has exactly one or two rows with an entry equal to $2$.  The refined bounds are also applicable to nonquadratic networks whose SR-graph $\mG_A$ satisfies some of the properties in \Cref{lemcomb}, for example, having an even R-cycle or a quadratic component with an S-vertex of degree $4$ and remaining S-vertices of degree $2$.
A drawback, however, of the techniques in the proof of \Cref{thmquadnn2n}, is that they do not easily allow us to prove improvements for networks which factor over a nontrivial partition, as was possible in \Cref{thm:n_plus_2} in the general case.

\end{rem}

\begin{rem}[Open questions on sharpness of the bounds in \Cref{thmquadnn2n}]
 \label[remark]{remnn2n}
In the case $n=2$, it can be checked, with the help of \texttt{Mathematica}, that no quadratic $(2,4,2)$ network with product molecularity of less than $6$ admits three positive nondegenerate equilibria. An example with product molecularity of $6$ having three positive nondegenerate equilibria is presented in \cite[Remark~38]{BBH2024smallbif}, and thus the bound in \Cref{thmquadnn2n} is sharp for $n=2$. By the inheritance results in \cite{banajisplitting,feliu:intermediates}, some quadratic $(3,5,3)$ networks must also admit at least three positive, nondegenerate equilibria: we can add an additional species as an intermediate into some reaction of a quadratic $(2,4,2)$ network admitting three positive, nondegenerate equilibria to obtain such a network. More generally, if some quadratic $(n,n+2,n)$ network admits $k$ positive nondegenerate equilibria, then the same holds for some quadratic $(n+1,n+3,n+1)$ network. However, we were not able to find a quadratic $(3,5,3)$ network admitting {\em four} positive nondegenerate equilibria, and thus achieving the bound in \Cref{thmquadnn2n}. We leave open the question of whether the bound in \Cref{thmquadnn2n} can be achieved by quadratic $(n, n+2, n)$ networks for $n \geq 3$.
\end{rem}

We close this subsection with some remarks on quadratic $(n,m,n)$ networks with more than $n+2$ distinct sources. Recall that by \Cref{remaltBezout} for any $(n,m,n)$ network, not necessarily quadratic, with at most $n+2$ distinct sources, the alternative Bézout source bound is never larger than the naive BKK source bound. \Cref{BKKbeatsBezout} illustrated that this result does not extend to the case of $n+3$ distinct sources, even for quadratic networks. On the other hand, the next example provides an instance with  $n+3$ distinct sources where an alternative B\'ezout source bound improves on the naive BKK source bound. 

\begin{ex}[The alternative B\'ezout source bound improves on the naive BKK source bound]
 \label[example]{ex363}
Consider any $(3,6,3)$ network with sources
\[
\mathsf{0}, \quad 2\fX, \quad \fY, \quad 2\fY, \quad \fZ, \quad \fY+\fZ\,.
\]
A mixed volume calculation gives a naive BKK source bound of $6$.
However, with minimal effort, by examining $\ker\,[A\,|\,\bm{1}]^\top$, we obtain an alternative B\'ezout source bound of $4$ for every such network. The details are left to the reader who may also find it interesting to interpret the claim in terms of labellings of $\mG_A$. 
\end{ex}

Although we cannot be sure which bound will be sharpest when there are $n+3$ or more distinct sources, as a general principle, for an $(n,m,n)$ network with $n+k$ sources where $k$ is small relative to $n$, it is often easier to extract information about positive equilibria from the smaller solvability system rather than the original mass action equations.

\subsection{The toricity matrix for quadratic networks with the trivial partition}
Note that we can find a basis of $\ker\,[A\,|\,\bm{1}]$ consisting of vectors of two forms:
 vectors arising from linear dependencies amongst columns of $A$, namely, vectors in $\ker\,A$ augmented with a terminal zero;
  and at most one vector not of this form. 
 The latter only occurs if $\bm{1} \in \im A$. Geometrically, this condition means that the dimension of the affine hull of the sources is strictly smaller than the dimension of their linear span, equivalently $0$ does not belong to the affine hull of the sources.
 
We treat these two types of vectors in the next lemmas.

Given a quadratic exponent matrix $A$, 
we define the following two graphs easily derived from $\mG:=\mG_A$. (i) $\widehat{\mG}$ is the subgraph of $\mG$ obtained by removing all connected components of $\mG$ which include an R-vertex with fewer than two distinct neighbours. As $A$ is quadratic, $\widehat{\mG}$ includes no parallel edges, and all R-vertices in $\widehat{\mG}$ have degree $2$. (ii) $\mG^*$ is derived from $\widehat{\mG}$ by
replacing each R-vertex and its incident edges with an edge between its two neighbours, to obtain a graph on S-vertices alone. ($\mG^*$ may include parallel edges.)

\begin{lem}
 \label[lemma]{lemquadker1}
For any quadratic network, $\ker A$ has a basis of $\{-1,0,1\}$-valued vectors with disjoint support. Any row-reduced basis will be of this form. The nullity of $A$ is precisely the number of bipartite components of $\mG^*$.
\end{lem}

\vspace{-\baselineskip}\begin{proof}
Recall that elements of $\ker A$ are identified with consistent labellings of $\mG:=\mG_A$. Isolated R-vertices 
  do not affect $\ker A$ and we may, without loss of generality, remove these from $\mG$. 

Consider any consistent labelling of $\mG$. Since R-vertices have degree no more than $2$, if the label of an S-vertex is zero, then all labels of the S-vertices in the same component of $\mG$ are also zero. An immediate consequence is that all labels must be zero on any component which includes an R-vertex with only one neighbour, and so elements of $\ker A$ correspond uniquely to consistent labellings of $\widehat{\mG}$. 

It is easily seen that a consistent labelling can have support including some component $\widehat{\mG}_i$ of $\widehat{\mG}$ if and only if the corresponding component, say $\mG^*_i$, of $\mG^*$ is bipartite in which case, since a connected bipartite graph is uniquely two-colourable, the consistent labelling of $\widehat{\mG}_i$ may be chosen to be a $\{-1,1\}$-labelling   and is unique up to scalar multiples. 

For each bipartite component of $\mG^*$, we choose arbitrarily one of the two possible  consistent $\{-1,1\}$-labellings supported on the corresponding component of $\mG$. Each such component thus contributes exactly one nonzero vector in $\ker A$ with values in $\{-1,0,1\}$, and trivially these have disjoint support, hence are linearly independent. By construction, these vectors span $\ker A$. This proves the first and last claims.
For the second claim we note that starting with any basis of $\ker A$ as rows of a matrix, the new rows obtained by row-reduction have minimal support. Consequently each vector of a row-reduced basis must be supported on exactly one component of $\mG$ corresponding to a bipartite component of $\mG^*$, as constructed.  
\end{proof}

\begin{rem}[Connection with the toricity matrix]
For a network such that $\bm{1} \not \in \im A$, the number of bipartite components of $\mG^*$, i.e., the nullity of $A$, is precisely the rank of the toricity matrix $\wQ$ corresponding to the trivial partition.
\end{rem}

In what follows, we write $\bm{2}$ for a vector of twos of length inferred from the context. 
\begin{lem}
 \label[lemma]{lemquadker2}
For any quadratic matrix $A\in \Z^{m\times n}_{\geq 0}$ such that $\bm{1} \in \mathrm{im}\,A$, there exists a $\{0,1,2\}$-valued vector $z$ satisfying $Az = \bm{2}$. 
\end{lem}
\vspace{-\baselineskip}
\begin{proof} 
As $\bm{1} \in \mathrm{im}\,A$, there exists $z$ such that $Az = \bm{2}$. 
As $A$ cannot have a zero row, for each $j=1,\dots,m$, either (i) the $j$th row of $A$ has exactly one nonzero entry $A_{ji}$, in which case $2=(Az)_j=A_{ji}z_i$ and hence $z_i=1$ or $2$; or 
(ii) the $j$th row of $A$ has exactly two nonzero entries $A_{ji},A_{jk}$, $i\neq k$, 
in which case $2=(Az)_j = z_i+z_k$. Whenever the constraints in (i) and (ii) can be satisfied, they can clearly be satisfied by a vector with values in $\{0,1,2\}$.
\end{proof}

\begin{prop} \label[proposition]{lem:kerA1}
For a quadratic network with trivial partition, we can choose $\wQ$ such that (i) its columns include at most one $\{0,1,2\}$-valued vector; (ii) the remaining columns are $\{-1,0,1\}$-valued vectors with disjoint support; and (iii) its row sums are at least $-1$ and at most $2$. 

Additionally, independently of the chosen partition, if $\bm{1} \in \mathrm{im}\,A$ and $\ker A$ is trivial, then $\wQ$ consists of one column and $\mu^{\wQ}$ consists of monomials of degree at most $2$. 
\end{prop}
\vspace{-\baselineskip}
\begin{proof}
Parts (i) and (ii) follow directly from \Cref{lemquadker1,lemquadker2}. 
Part (iii) follows as the $\{-1,0,1\}$-valued vectors from \Cref{lemquadker1} have disjoint support, and whenever the $\{0,1,2\}$-valued vector constructed in the proof of \Cref{lemquadker2} {\em must} have a $2$ in $i$th place, i.e., the network includes some reaction with source $\fX_i$, then all the $\{-1,0,1\}$-valued basis vectors constructed in \Cref{lemquadker1} must have a zero in $i$th place. 

The last statement follows again from \Cref{lemquadker2}, as if $z$ satisfies $Az=\bm{2}$ and is as in the lemma, then, given any $p$-partition $\mP$, the vector with its first $n$ entries agreeing with $z$ and the last $p$ entries equal to $-2$ belongs to $\ker\,[A | \bm{1}_\mP]$. 
\end{proof}

If a network admits a monomial parameterisation, namely $\#\Ak \leq 1$ for all $\k$, then, whenever $\Ak$ is nonempty, say $\Ak = \{\alpha_\k\}$, the equation for positive equilibria on a stoichiometric class, namely $Z[(h(\alpha_\k)/\k)^{\wG}\circ \mu^{\wQ}] = K$, depends only on $\mu^{\wQ}$ (see \Cref{remmonopar}). This includes the special case where $Y$ is a singleton. When, additionally, the network is quadratic, \Cref{lem:kerA1} tells us that the exponents in $\wQ$ are $-1,0,1$ or $2$ and are hence small. Since $\mu$ is a positive variable, this sometimes allows us to obtain an improved BKK bound on the number of positive nondegenerate equilibria, as illustrated by the next example.

\begin{ex}[An alternative BKK source bound from $\wQ$]
\label[example]{ex543_1}
Consider any dynamically nontrivial $(5,4,3)$ network with sources
\[
2\mathsf{V}, \quad 2\mathsf{W},\quad \mathsf{W}+\fX,\quad \fY+\fZ\,,
\]
which are easily seen to be affinely independent. We get a naive B\'ezout source bound of $8$, and a naive BKK source bound of $4$.
However, we can obtain an alternative BKK source bound which improves on the naive bound. As $\ker \Gamma$ is spanned by a positive vector, say $u$, the network can factor only over the trivial partition. 
Computing $\wQ$, but without actually computing $\wG$, we find that, after clearing denominators,
the alternative system, $Z[(u/\k)^{\wG}\circ \mu^{\wQ}] = K$, consists of a pair of cubic equations in $(\mu_1, \mu_2) \in \mathbb{R}^2_+$ of the form
\[
\begin{array}{rcl}
 p_1\mu_1+p_2\mu_2+p_3\mu_1\mu_2+p_4\mu_1^2\mu_2 &=& 0\,,\\
 q_1\mu_1+q_2\mu_2+q_3\mu_1\mu_2+q_4\mu_1^2\mu_2 &=& 0\,.
\end{array}
\]
The alternative B\'ezout bound of $9$ is clearly unhelpful in this case; but by an easy calculation (or, indeed, inspection) the BKK theorem applied to this pair of polynomials gives an alternative BKK source bound of $2$. Note that we reduced a mixed-volume calculation in $\R^5$ to one in $\R^2$ by studying the alternative equations.
\end{ex}
\subsection{Quadratic networks with a single conservation law}\label{sec:quadraticnnn1}
In this subsection, we consider quadratic $(n,m,n-1)$ networks. The bound in \Cref{thmnnn1}(i) can be greatly improved by using \Cref{lem:kerA1}, and made independent of $n$, for quadratic networks with the trivial partition and for other partitions under the additional hypothesis that $\bm{1} \in \im A$. 

\begin{thm}
 \label{thmquadnnn1}
 Consider a quadratic $(n,m,n-1)$ network with $n\leq m \leq n+p_f -1$. Then:
 \begin{itemize}
 \item[(i)] If the source rank is less than $n$ (equivalently, $\mathrm{rank}[A\,|\,\bm{1}] \leq n$, so that either $\bm{1} \in \im A$ or $\rank A<n$), then the network admits no more than two positive nondegenerate equilibria on any stochiometric class.
 \item[(ii)] If $\bm{1} \in \im A$ and there exists a nonzero $\omega \in \ker\Gamma^\top \cap \R^n_{\geq 0}$, then the network admits at most one positive nondegenerate equilibrium on each stoichiometric class. 
  \end{itemize} 
\end{thm}\vspace{-\baselineskip}
\begin{proof}
 
 Assume that the network is nondegenerate, for otherwise the conclusions are trivial. As in the proof of \Cref{thmnnn1}, the assumptions imply that $p_f=m-n+1$, the network is $\mP_f$-toric, and it has $\mP_f$-independent sources. We choose the finest partition. Then, $Z$ is a row vector, $Q$ is a column vector and, by \Cref{thm:degensolvability}, positive nondegenerate equilibria on the stoichiometric class defined by $Zx = K$ correspond to nondegenerate solutions $\mu \in \mathbb{R}_+$ to
\begin{equation}\label{eq:quad}
Z[(v/\k)^{\wG}\circ \mu^{\wQ}] = K\,.
\end{equation}
We consider the two cases 
 $\bm{1} \in \im A$ or $\rank A<n$ separately:
\begin{enumerate}[align=left,leftmargin=*]
\item[(a)] If $\bm{1} \in \mathrm{im}\,A$, then by \Cref{lem:kerA1}, $\wQ$ can be chosen such that $\mu^{\wQ}$ consists of monomials of degree at most $2$. 
\item[(b)] If $\mathrm{rank}\,A <n$, then we must have $\mathrm{rank}\,A = n-1$ (otherwise the network would be source deficient, hence degenerate). By \Cref{lemquadker1}, we can choose $Q$ (hence $\wQ$) to be a $(-1,0,1)$ vector. 
\end{enumerate}
After clearing denominators if necessary, we obtain that in both cases
\eqref{eq:quad} is at-most-quadratic in $\mu$. This proves (i). 
Part (ii) now follows from \Cref{thmnnn1}(iii), as case (a) applies and hence $\wQ$ has nonnegative integer entries. 
\end{proof}

\begin{cor}
 \label[corollary]{corquadnnn1}
 Quadratic $(n,n,n-1)$ networks admit no more than two positive nondegenerate equilibria on each stochiometric class. The smallest $n$ for which this bound is achieved is $n=3$. The bound becomes $1$ if $\bm{1} \in \mathrm{im}\,A$ and there exists a nonzero $\omega \in \ker \Gamma^\top \cap \R^n_{\geq 0}$.
 \end{cor}\vspace{-\baselineskip}
\begin{proof}
As usual, it is enough to consider networks which are nondegenerate, hence dynamically nontrivial and not source deficient. With the trivial partition we have $\mathrm{rank}\,[A\,|\,\bm{1}]=n$, and hence either $\mathrm{rank}\,A = n$ (and hence $\bm{1} \in \mathrm{im}\,A$); or $\mathrm{rank}\,A = n-1$ and $\bm{1} \not \in \mathrm{im}\,A$. The first and last claims now follow from \Cref{thmquadnnn1}. 
 
We observed in \Cref{remnnn1} that, by \cite[Lemma 20]{BBH2024smallbif}, no quadratic $(2,2,1)$ network admits multiple, positive, nondegenerate equilibria on any stoichiometric class. An example of a quadratic -- in fact bimolecular -- $(3,3,2)$ network with two positive, nondegenerate equilibria on a stoichiometric class was presented in \cite[Section~5.1]{banaji:boros:hofbauer:2024a}. \end{proof}

We close with a result which demonstrates how we may combine various results above to obtain bounds on positive nondegenerate equilibria: in this case for quadratic $\mP$-toric networks with few sources and a single conservation law.
\begin{thm}
\label{thmquadnn2nA}
Let $k \geq 2$. Consider a dynamically nontrivial, quadratic $(n,n+k,n-1)$ network with source rank $n-1$ (i.e., such that $\mathrm{rank}[A\,|\,\bm{1}] = n$).
\begin{enumerate}
\item If the network has $n$ distinct sources, we have an alternative B\'ezout source bound of $2$, and consequently the network admits no more than $2$ positive nondegenerate equilibria on any stoichiometric class. 
\item If the network has $n+1$ distinct sources, we have an alternative B\'ezout source bound of $2n$, and consequently the network admits no more than $2n$ positive nondegenerate equilibria on any stoichiometric class.
\end{enumerate}
\end{thm}
\vspace{-\baselineskip}\begin{proof}
We consider the network with trivial partition $\mP$ and observe that the network is $\mP$-toric by assumption, and so, by considering the two equations of \eqref{gensol} sequentially, we automatically get an alternative B\'ezout bound, see \Cref{remtoricBezout}. By \Cref{thm:degensolvability} a solution to this pair of equations is nondegenerate if and only if the corresponding equilibrium is. Furthermore, $\wQ$ consists of a single column, which has one of the forms described in \Cref{lem:kerA1}. We choose $W$ as in \Cref{lem:Wstruct}. 

In Part (1),  the solvability system $h(\alpha)^W = \k^W$ is equivalent to a linear system in $\alpha$ (see the proof of \Cref{thmfewsources}), while the equation $Z[(h(\alpha)/\k)^{\wG}\circ \mu^{\wQ}] = K$, can be written as an at-most quadratic equation in $\mu$ (see \Cref{lem:kerA1}), and the first claim follows.

  In Part (2), $W$ has $k-1$ rows which give rise to linear equations, and one nontrivial row, say $w$. We may remove some column of $A$ which is in the span of the remaining columns of $[A\,|\,\bm{1}]$, to obtain an $(n+k) \times n$ matrix, say $[\tilde{A}\,|\,\bm{1}]$, of rank $n$; this does not affect $W$. As in the proof of \Cref{thmquadnn2n}, we can choose  $w$ to have degree no more than $n$. The solvability system thus has a B\'ezout bound of $n$. Moreover, as in Part (1), the equation $Z[(h(\alpha)/\k)^{\wG}\circ \mu^{\wQ}] = K$, can be written as an at-most quadratic equation in $\mu$, leading to a combined B\'ezout bound of $2n$.
\end{proof}

\section{Conclusions}

We have presented approaches characterising the positive equilibrium set of a mass action network based on the network's basic properties, such as the number of species and reactions, the rank, and the configuration of its sources. These approaches begin with writing down alternative systems of equations, often simpler than the original mass action equations, whose solutions are in smooth, one-to-one correspondence with positive equilibria. Moreover, degeneracy or nondegeneracy of equilibria can also be inferred from the alternative systems.

We demonstrated the usefulness of these alternative systems via a number of results and examples, some drawn from the biological literature (see, in particular, \Cref{ex:2site,,exBIOMD,,ex664}). We saw the usefulness of these techniques in parameterising equilibria; in providing bounds on positive nondegenerate equilibria; in characterising the parameter set for multistationarity; and in studying bifurcations. Many of the claims follow from direct calculation or simple applications of B\'ezout's theorem, Descartes' rule of signs, the BKK theorem, and so forth, to the alternative systems. 

We observed that networks may admit more than one partition, and the alternative systems of equations are dependent on a choice of partition; however, the finest partition is often an advantageous choice. Open questions remain about how best to choose various objects such as the generalised inverse $G$ of $[A\,|\,\bm{1}_\mP]$, to obtain the sharpest possible results. 

In the important special case of quadratic networks, we showed that bounds can often be improved. The particularly simple nature of the exponent matrices of quadratic networks opened up new combinatorial techniques for exploring multistationarity, which we only touched on here, but propose as a basis for further theory and applications.

This paper has focussed on developing and illustrating a basic set of techniques; we expect these techniques to find applications in the study of biological networks, especially in conjunction with results on inheritance of nondegenerate behaviours. Much of the work here was originally motivated by the study of bifurcations in mass action networks, and we expect the results here will, in turn, deepen our understanding of bifurcations in mass action networks.

{\bf Acknowledgements}
MB's work was supported by Research England under the Expanding Excellence in England (E3) funding stream awarded to MARS: Mathematics for AI in Real-world Systems in the School of Mathematical Sciences at Lancaster University. EF has been supported by the European Union under the Grant Agreement number 101044561, POSALG. Views and opinions expressed are those of the authors only and do not necessarily reflect those of the European Union or European Research Council (ERC). Neither the European Union nor ERC can be held responsible for them. MB would like to thank Amlan Banaji for several useful conversations on Hausdorff dimension. MB and EF would like to thank Alicia Dickenstein for helpful comments on a draft of this work.

{\bf Conflicts of interest}
The authors have no relevant competing interests to declare.

{\bf Data availability}
Code used to generate and analyse networks is available on github at https://github.com/CRNcode/CRN. The mass action models from the ODEbase repository analysed in this paper are available at https://odebase.org/.

\small
\bibliographystyle{plain}

\end{document}